\documentclass[journal]{IEEEtran}
\usepackage{graphicx}
\usepackage[utf8]{inputenc}
\usepackage{graphics}
\usepackage{epsfig} 
\usepackage{mathptmx} 
\usepackage{amsmath} 
\usepackage{amssymb}  
\usepackage{amsthm}
\usepackage{cite}
\usepackage{multicol}
\usepackage{mathtools, cuted}
\usepackage[cmintegrals]{newtxmath}
\usepackage{epstopdf}

\usepackage{helvet}
\usepackage{courier}
\usepackage{subfigure}
\usepackage{type1cm}
\usepackage{makeidx}         
\usepackage{algorithmic}

\theoremstyle{definition}

\theoremstyle{theorem}
\newtheorem{theorem}{Theorem}
\theoremstyle{lemma}
\newtheorem{lemma}{Lemma}
\theoremstyle{remark}
\newtheorem{remark}{Remark}
\theoremstyle{proposition}
\newtheorem{proposition}{Proposition}
\theoremstyle{corollary}
\newtheorem{corollary}{Corollary}
\theoremstyle{assumption}
\newtheorem{assumption}{Assumption}
\usepackage{multicol}        
\usepackage[bottom]{footmisc}
\usepackage{caption}
\usepackage{nomencl}
\usepackage[usenames, dvipsnames]{color}

\begin{document}
%
\title{Scalable Vehicle Team Continuum Deformation Coordination with Eigen Decomposition}
%
%
%

\author{Hossein~Rastgoftar, Ella M. Atkins,~\IEEEmembership{Senior Member},~and Ilya Kolmanovsky,~\IEEEmembership{{\color{black}Fellow}}
\thanks{{\color{black}All authors are with the Aerospace Engineering Department, University of Michigan, Ann Arbor,
MI, 48109 USA, e-mails: \{hosseinr, ematkins, ilya\}@umich.edu}.}
}
%
%

\markboth{}%
{Shell \MakeLowercase{\textit{et al.}}: Bare Demo of IEEEtran.cls for Journals}
%



\maketitle

\begin{abstract}
The continuum deformation leader-follower cooperative control strategy models vehicles in a multi-agent system as particles of a deformable body. A desired continuum deformation is defined based on leaders' trajectories and acquired by followers in real-time through local communication. The existing continuum deformation theory requires followers to be placed inside the convex simplex defined by leaders. This constraint is relaxed in this paper. 
We prove that under suitable assumption{\color{black}s} any $n+1$ {\color{black}($n=1,2,3$)} vehicles forming an $n$-D simplex can be selected as leaders while followers, arbitrarily positioned inside or outside the leading simplex, can acquire a desired continuum deformation in a decentralized fashion. The paper's second contribution is to assign a one-to-one mapping between leaders' smooth trajectories and homogeneous deformation features obtained by continuum deformation eigen-decomposition. Therefore, a safe and smooth continuum deformation coordination can be planned either by shaping homogeneous transformation features or by choosing appropriate leader trajectories. This is beneficial to efficiently plan and guarantee 
collision avoidance in a large-scale group. 
A simulation case study is reported in which a virtual convex simplex contain{\color{black}s} a {\color{black}quadcopter} vehicle team at any time $t$; A* search {\color{black}is} applied to optimize quadcopter team continuum deformation in an obstacle-laden environment.
\end{abstract}


%
\IEEEpeerreviewmaketitle

\section{Introduction}
Formation and cooperative control algorithms \cite{lin2014distributed} have been applied to problems in biology \cite{wang2015bswarm}, computer science \cite{boissier2013multi}, aerospace engineering \cite{dong2015time, nazari2016decentralized}, and elsewhere.
Virtual structure \cite{ren2002virtual, low2011flexible}, consensus \cite{ren2007information, feng2014group, zhu2010leader, xiao2006state, liu2018exponential, li2018nonlinear, cao2015leader, shao2018leader, olfati2004consensus, zhang2010consensus}, containment \cite{li2016containment, li2015containment, liu2014containment, wang2014distributed, liu2015containment, zhao2015finite, zhao2015robust, notarstefano2011containment}, and continuum deformation \cite{ rastgoftar2016continuum, rastgoftar2017continuum, lal2006continuum} are available {\color{black}multi-agent system (MAS)} coordination methods. While the virtual structure (VS) method is commonly exploited for centralized coordination, the other three methods provide decentralized solutions. The VS method treats MAS as particles of a virtual rigid body; rigid body translation and rotation prescribes agents' trajectories in a 3-D motion space. 
Consensus algorithm stability has been analyzed under fixed and switching communication topologies \cite{feng2014group, yu2010group} and in the presence of fixed and time-varying delays \cite{peng2007distributed, zhu2010leader, xiao2006state, zhang2010consensus}. Finite-time consensus under fixed and switching communication topologies {\color{black}is} developed in Refs. \cite{liu2018exponential, li2018nonlinear}, while leader-follower consensus is investigated in Refs. \cite{cao2015leader, shao2018leader}.

In containment control, leaders independently guide collective motion, and followers acquire the desired coordination via local communication. Containment control stability and convergence \cite{ji2008containment} with fixed \cite{li2016containment} and switching \cite{li2015containment, liu2015containment, notarstefano2011containment} communication topologies have been analyzed in {\color{black}the} existing literature.  Retarded containment with fixed \cite{li2016containment} and time-varying \cite{liu2014containment, zhao2015robust} time-delays and finite-time containment control and coordination \cite{wang2014distributed, zhao2015finite} have {\color{black}also} been investigated. 

{\color{black}Defining agent coordination by continuum deformation\footnote{A continuum is {\color{black}defined as} a continuous domain containing an infinite number of particles with infinitesimal size \cite{lai2009introduction}.} was first introduced in Ref. \cite{lal2006continuum}.} {\color{black}Leader-follower continuum deformation proposed in \cite{rastgoftar2016continuum, rastgoftar2018asymptotic}} treats leader and follower agents as particles of a deformable body. Leaders form an $n$-D simplex containing follower agents during MAS evolution ($n=1,2,3$). A desired formation is defined by a homogeneous transformation (deformation) uniquely related to the trajectories of $n+1$ leaders. Followers acquire the desired homogeneous transformation in real-time through local communication and apply communication weights consistent with each agent's position in the reference configuration. {\color{black}Also,  Refs. \cite{zhao2018affine, xu2018affine} offer a leader-follower affine transformation method for  multi-agent coordination  where graph rigidity is explained to specify followers' communication weights based on agents' reference configuration.} Continuum deformation supports fixed \cite{rastgoftar2016continuum} and switching \cite{rastgoftar2017continuum} communication topologies. Ref. \cite{rastgoftar2016continuum} analyzes continuum deformation stability in {\color{black}the} presence of communication delays. Alignment and polyhedral communication topologies are analyzed in \cite{rastgoftar2016continuum, rastgoftar2014continuum}, and continuum deformation with more than $n+1$ moving leaders is studied in \cite{rastgoftar2018safe}.  Containment control and continuum deformation are both leader-follower methods. Continuum deformation extends containment control by assuring inter-agent collision avoidance as well as containment. 

{\color{black}
This paper offers a novel eigen-decomposition method  for continuum deformation coordination of a multi-vehicle system (MVS) in a 3-D motion space. {\color{black}This eigen-decomposition leads to a computationally-efficient and scalable continuum deformation coordination approach in an obstacle-free environment and a less conservative safety condition for inter-agent collision avoidance.} By relaxing limitations considered in the previous work, we advance {\color{black}the theory of} continuum deformation acquisition in obstacle-laden and obstacle-free environments. Furthermore, we study continuum deformation planning and optimization in a cluttered environment.} 
{\color{black} Compared to the available {\color{black}literature} and the authors' {\color{black}prior} work \cite{rastgoftar2018safe, rastgoftar2017continuum, rastgoftar2016continuum, rastgoftar2018asymptotic}, this article offers the following novel contributions:
\begin{itemize}
    \item{This paper advances continuum deformation coordination theory \cite{rastgoftar2016continuum, rastgoftar2017continuum, rastgoftar2014continuum, rastgoftar2018safe} by relaxing the containment requirement considered previously. {\color{black}Specifically,} we show {\color{black}in this paper} that any $n+1$ {\color{black}($n=1,2,3$)} {{vehicle}}s forming an $n$-D simplex can be selected as leaders. Followers, placed either inside or outside of the leading simplex, infer the desired continuum deformation in real time through local communication. 
    }
     \item{We advance {\color{black}the theory of} continuum deformation {\color{black}for} integrator agents toward continuum deformation of vehicles with {\color{black}input-output} linearizable dynamics. Assuming each vehicle has minimum-phase dynamics, {\color{black}this} paper guarantees inter-agent collision {\color{black}avoidance} in a motion governed by the continuum deformation algorithm with significant rotation and deformation {\color{black}possible}.}
     \item{{\color{black}The existing continuum deformation coordination method ensures inter-agent collision avoidance by assigning a single lower-limit for all deformation eigenvalues. This could make continuum deformation coordination restrictive specifically when agents are non-uniformly distributed in the reference configuration as agent minimum separation constraints are related to deformation matrix eigenvalues. This paper guarantees inter-agent collision avoidance by assigning a lower-limit {\color{black} on one of the eigenvalues of the pure deformation matrix that characterizes the minimum separation distance in the reference configuration of the vehicles,} while the {\color{black}other two} eigenvalues only need to be positive to maintain the requirement of the continuum deformation coordination (see Theorem 5 below). This paper also relaxes agent spacing requirements in regions where the single minimum eigenvalue was unnecessarily restrictive. This new less conservative strategy is advantageous because more aggressive continuum deformation maneuvers are possible with distinct lower limits for the deformation eigenvalues.} }
\end{itemize}

}

{\color{black}In this paper, MVS desired homogeneous transformation is uniquely represented by the following features: (i) A rotation matrix, (i) A positive definite deformation matrix defining principle deformations (eigenvalues) and their orientations (eigenvectors) along with a rigid-body displacement vector. A one-to-one mapping is obtained to relate leader trajectories defining an $n$-D homogeneous transformation to homogeneous transformation features. Safe continuum deformations will be planned either by shaping homogeneous transformation features or by choosing desired leader trajectories. In an obstacle-free environment, a large-scale continuum deformation coordination is planned strictly by shaping homogeneous transformation features. This is beneficial because safe leaders' trajectories, ensuring inter-agent collision avoidance, can be determined at low computational cost. Alternatively, desired homogeneous transformations through an obstacle-laden environment can be planned by optimizing leaders' trajectories such that the prescribed lower limit on homogeneous deformation eigenvalues required for obstacle collision avoidance is satisfied.  Case studies are presented illustrating how safety constraints {\color{black}can be incorporated into optimizing}
continuum deformation given initial and target MAS formations.}

This paper is organized as follows. Preliminaries presented in 
Section \ref{Preliminaries} are followed by {\color{black}inter-agent communication topology and graph theory definitions in Section \ref{Inter-Agent Communication}.} {\color{black}Section \ref{Problem Statement} presents the formulations and statements of the problems considered in this paper. MVS collective dynamics is obtained in Section \ref{MAS Collective Dynamics}.} Safety requirements of MVS continuum deformation coordination are obtained in Section \ref{SimpleSafety}. Continuum deformation planning is formulated in Section \ref{Continuum Deformation Optimization}. Case study results in Section \ref{simulation} are followed by a conclusion in Section \ref{Conclusion}.

\section{Preliminaries}
\label{Preliminaries}
\subsection{Position Notations}\label{Position Notations}
 Agent positions are expressed with respect to  a Cartesian frame with unit basis vectors $\hat{\mathbf{e}}_1$,  $\hat{\mathbf{e}}_2$, and  $\hat{\mathbf{e}}_3$. {\color{black}Expressing $\hat{\mathbf{e}}_1=\left[1~0~0\right]^T$,  $\hat{\mathbf{e}}_2=\left[0~1~0\right]^T$,  $\hat{\mathbf{e}}_3=\left[0~0~1\right]^T$,} the  paper defines the following position notations for every agent $i\in \mathcal{V}$:
 \\
 \textbf{Actual Position} vector of vehicle $i\in \mathcal{V}$  denoted by $\mathbf{r}_i=\left[x_i~y_i~z_i\right]^T$ is considered as the output of the control system of every vehicle $i\in \mathcal{V}$.
 \\
 \textbf{Initial Position} vector of vehicle $i\in \mathcal{V}$ is denoted by $\mathbf{r}_{i,s}=\left[x_{i,s}~y_{i,s}~z_{i,s}\right]^T\in \Omega_s$, where $\Omega_s\subset \mathbb{R}^3$ is a finite set.
 \\
 \textbf{Reference Position} vector of vehicle $i\in \mathcal{V}$ is denoted by $\mathbf{r}_{i,0}=\left[x_{i,0}~y_{i,0}~z_{i,0}\right]^T\in \Omega_0$, where $\Omega_0\subset \mathbb{R}^3$ is a finite set.
 \\
 \textbf{Global Desired Position} vector of vehicle $i\in \mathcal{V}$  denoted by $\mathbf{r}_{i,HT}=\left[x_{i,HT}~y_{i,HT}~z_{i,HT}\right]^T$ is defined by homogeneous deformation that is presented in Section \ref{MVS Homogeneous Deformation Coordination}.
 \\
 \textbf{Local Desired Position} vector of vehicle $i\in \mathcal{V}$ denoted by $\mathbf{r}_{d,i}=\left[x_{d,i}~y_{d,i}~z_{d,i}\right]^T$ is defined in Section \ref{MAS Collective Dynamics}.

{\color{black} 
 \subsection{Motion Space Discretization}
 Let $\mathbf{p}_1\in \mathbb{R}^{n+1}$, $\cdots$, $\mathbf{p}_{n+1}\in \mathbb{R}^{n+1}$, and $\mathbf{c}$ be position vectors of $n+2$ points in an $n$-D hyperplane. Defining a scalar function
 \begin{equation}
     \Psi_n\left(\mathbf{p}_1,\cdots,\mathbf{p}_{n+1}\right)=\mathrm{rank}\left(
     \begin{bmatrix}
         \mathbf{p}_2-\mathbf{p}_1&\cdots&\mathbf{p}_{n+1}-\mathbf{p}_1
     \end{bmatrix}
     \right), 
 \end{equation}
vectors  $\mathbf{p}_1\in \mathbb{R}^{n+1}$, $\cdots$, $\mathbf{p}_{n+1}\in \mathbb{R}^{n+1}$ assign positions of vertices of an $n$-D simplex, if $\Psi_n\left(\mathbf{p}_1,\cdots,\mathbf{p}_{n+1}\right)=n$.

If $\Psi_n\left(\mathbf{p}_1,\cdots,\mathbf{p}_{n+1}\right)=n$,  we can define a vector function
 \begin{equation}
 \label{Thetaaan}
     \mathbf{\Theta}_n\left(\mathbf{p}_1,\cdots,\mathbf{p}_{n+1},\mathbf{c}\right)=
     \begin{bmatrix}
     \mathbf{p}_1&\cdots&\mathbf{p}_{n+1}\\
     1&\cdots&1
     \end{bmatrix}
     ^{-1}
     \begin{bmatrix}
         \mathbf{c}\\
         1
     \end{bmatrix}
     .
 \end{equation}
 The following properties hold for the vector $\mathbf{\Theta}_n\in \mathbb{R}^{\left(n+1\right)\times 1}$:
 \begin{enumerate}
     \item{The sum of the entries of $\mathbf{\Theta}_n$ is $1$, i.e. $\mathbf{1}_{1\times \left(n+1\right)}\mathbf{\Theta}_n=1$, where $\mathbf{1}_{1\times \left(n+1\right)}\in \mathbb{R}^{1\times \left(n+1\right)}$ is a vector with all components equal to $1$.} 
     \item{ If $\mathbf{\Theta}_n\left(\mathbf{p}_1,\cdots,\mathbf{p}_{n+1},\mathbf{c}\right)>\mathbf{0}$, then, point $\mathbf{c}$ is inside the simplex made by  $\mathbf{p}_1$, $\cdots$, and $\mathbf{p}_{n+1}$. Otherwise $\mathbf{c}$ is outside.}
 \end{enumerate}

 }

\subsection{Rotation {\color{black}Matrix} }Angles $\beta_1$, $\beta_2$, and $\beta_3$ define rotation {\color{black}matrix} 
\begin{equation}\label{rotationmatrix}
\resizebox{0.99\hsize}{!}{%
$
    \mathbf{R}\left(\beta_1,\beta_2,\beta_3\right)=\begin{bmatrix}
    C_{\beta_2} C_{\beta_3}&  C_{\beta_2} S_{\beta_3} &-S_{\beta_2}\\
S_{\beta_1}S_{\beta_2} C_{\beta_3}-C_{\beta_1}S_{\beta_3}&S_{\beta_1}S_{\beta_2} S_{\beta_3}+C_{\beta_1}C_{\beta_3}&S_{\beta_1}C_{\beta_2}\\
   C_{\beta_1}S_{\beta_2} C_{\beta_3}+S_{\beta_1}S_{\beta_3} &C_{\beta_1}S_{\beta_2} S_{\beta_3}-S_{\beta_1}C_{\beta_3}&C_{\beta_1}C_{\beta_2}
\end{bmatrix}
,
$
}
\end{equation}
where $C_{\left(\cdot\right)}$ and $S_{\left(\cdot\right)}$ abbreviate $\cos{\left(\cdot\right)}$ and $\sin{\left(\cdot\right)}$, respectively. Rotation {\color{black}matrix} $\mathbf{R}\left(\beta_1,\beta_2,\beta_3\right)$ has the following properties:
\begin{enumerate}
    \item{Orthonormal: $\mathbf{R}^T\left(\beta_1,\beta_2,\beta_3\right)\mathbf{R}\left(\beta_1,\beta_2,\beta_3\right)=\mathbf{I}_3$ where $\mathbf{I}_3\in \mathbb{R}^{3\times 3}$ is the identity matrix.}
    \item{$\mathbf{R}(\beta_1,\beta_2,\beta_3)=\mathbf{R}(\beta_1,0,0)\mathbf{R}(0,\beta_2,0)\mathbf{R}(0,0,\beta_3)$.}
    \item{$\mathbf{R}(0,0,0)=\mathbf{I}_3$.}
\end{enumerate}


\subsection{MVS Homogeneous Deformation Coordination}\label{MVS Homogeneous Deformation Coordination}
 {\color{black}We consider collective motion of an MVS consisting of $N$ vehicles that are identified by index numbers defined by set $\mathcal{V}_R$.} {\color{black}We treat the {\color{black}vehicles constituting} the MVS as particles of a deformable body, where}
the global desired position of vehicle $i\in \mathcal{V}_R$, denoted $\mathbf{r}_{i,HT}=x_{i,HT}\hat{\mathbf{e}}_1+y_{i,HT}\hat{\mathbf{e}}_2+z_{i,HT}\hat{\mathbf{e}}_3$, is defined by homogeneous transformation. 

\begin{equation}
\label{homogtrans}
   \mathbf{r}_{i,HT}(t)=\mathbf{Q}(t)\mathbf{r}_{i,0}+\mathbf{d}(t),
\end{equation}
where {\color{black}$t$ is the current time,} $t_0$ is the reference time, $\mathbf{Q}\in \mathbb{R}^{3\times 3}$ is the Jacobian matrix, and {\color{black}$\mathbf{d}=\left[d_1~d_2~d_3\right]\in  \mathbb{R}^{3\times 1}$} is a  \emph{rigid body displacement vector}. Note that $\mathbf{Q}{\color{black}(t)}$ is nonsingular at time $t$.
Schematics of $1$-D, $2$-D, and $3$-D MVS reference configurations are illustrated, where reference position of agent $i\in \mathcal{V}$ is given by
\begin{equation}
    \mathbf{r}_{i,0}=
    \begin{cases}
    \left[x_{i,0}~0~0\right]^T&n=1\\
    \left[x_{i,0}~y_{i,0}~0\right]^T&n=2\\
    \left[x_{i,0}~y_{i,0}~z_{i,0}\right]^T&n=3\\
    \end{cases}
    .
\end{equation}
{\color{black}
\begin{assumption}
Vehicles' initial positions are given by 
\begin{equation}
\label{initialformation}
  \mathbf{r}_{i,s}=\mathbf{Q}_s\mathbf{r}_{i,0}+\mathbf{d}_s,
\end{equation}
where $\mathbf{Q}_s=\mathbf{Q}(t_s)$ and $\mathbf{d}_s=\mathbf{d}(t_s)$ denote  Jacobian matrix and  rigid-body displacement vector, respectively at the initial time instant $t=t_s$. {\color{black}This paper assumes that $\mathbf{Q}_s$ is an orthogonal matrix.}
\end{assumption}
}

\begin{figure}
 \centering
 \subfigure[]{\includegraphics[width=0.3\linewidth]{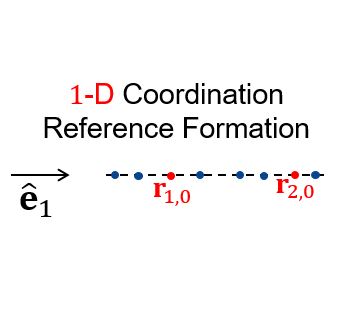}}
  \subfigure[]{\includegraphics[width=0.3\linewidth]{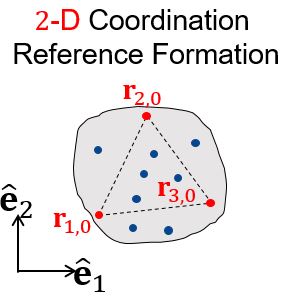}}
  \subfigure[]{\includegraphics[width=0.3\linewidth]{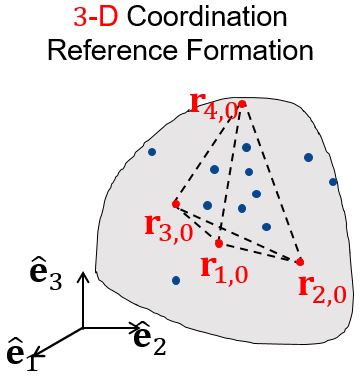}}
     \caption{Schematic of (a) $1$-D, (b) $2$-D, and (c) $3$-D reference configurations. }
\label{ReferenceSchematic}
\end{figure}

{

\subsubsection{Leader-Follower Homogeneous Deformation Coordination}
If $\mathbf{Q}$ and $\mathbf{d}$ are known at a time $t$,
leaders' trajectories can be assigned using the homogeneous transformation given in \eqref{homogtrans}. However, obstacle collision avoidance may not be necessarily guaranteed when Eq. \eqref{homogtrans} is directly used to define a continuum deformation coordination. {\color{black}This issue can be handled by} defining a large-scale continuum deformation coordination  by choosing appropriate leaders' trajectories {\color{black}that avoid} obstacles rather than shaping $\mathbf{Q}$ and $\mathbf{d}$ at any time $t$.
{\color{black}
\begin{assumption}\label{leaderrankcondition}
It is assumed that leader vehicles  $1$, $2$, $\cdots$, ${n+1}$ form an $n$-D simplex {\color{black}in the reference configuration}. Therefore,
\begin{equation}
\begin{split}
 \Psi_n\left(\mathbf{r}_{1,0},\cdots,\mathbf{r}_{{n+1},0}\right)=n.
\end{split}
\end{equation}
\end{assumption}
Because leaders' reference positions satisfy Assumption \ref{leaderrankcondition}}
global desired position of vehicle $i\in \mathcal{V}_R$ {\color{black}at $t=t_0$} is expressed as a linear combination of leader positions \cite{rastgoftar2016continuum}:
\begin{equation}
\label{leadersexpression}
    i\in \mathcal{V}_R,~t\geq t_s,\qquad \mathbf{r}_{i,HT}=\sum_{j=1}^{n+1}\alpha_{i,j}\mathbf{r}_{j,HT}(t),
\end{equation}
where $\alpha_{i,1}$, $\alpha_{i,2}$, $\cdots$, $\alpha_{i,n+1}$ are  {\color{black}called \textit{reference}} \emph{$\alpha$-parameters} and obtained by 

\begin{equation}
\label{communicationwithleaders}
    \begin{bmatrix}
    \alpha_{i,1}&
    \cdots&
    \alpha_{i,n+1}
    \end{bmatrix}
    ^T
    =\mathbf{\Theta}_n\left(\mathbf{r}_{1,0},\cdots,\mathbf{r}_{n+1,0},\mathbf{r}_{i,0}\right).
\end{equation}
}
\begin{table*}
    \small
    \centering
    \caption{$n$-dimensional Homogeneous Deformation Parameters}
    \begin{tabular}{|c| c| c| c| c| c| c| c| c |c |c |c| c| }
    \hline
    $n$&$\lambda_1$&$\lambda_2$&$\lambda_3$&$\phi_u$&$\theta_u$&$\psi_u$&$\phi_r$&$\theta_r$&$\psi_r$&$d_1$&$d_2$&$d_3$  \\
    \hline
         $1$&$>0$&$=1$&$=1$&$=0$&$=0$&$=0$&$=0$&$\in[-\pi/2,\pi/2]$&$\in[0,2\pi]$&$\in \mathbb{R}$&$\in \mathbb{R}$& $\in \mathbb{R}$ \\
         \hline
         $2$&$>0$&$>0$&$=1$&$=0$&$=0$&$\in[0,2\pi]$&$[0,\pi]$&$[0,2\pi]$&$[0,2\pi]$&$\in \mathbb{R}$&$\in \mathbb{R}$& $\in \mathbb{R}$ \\
         \hline
         $3$&$>0$&$>0$&$>0$&$\in[0,\pi]$&$\in[0,2\pi]$&$\in[0,2\pi]$&$\in[0,\pi]$&$\in[0,2\pi]$&$\in[0,2\pi]$&$\in \mathbb{R}$&$\in \mathbb{R}$& $\in \mathbb{R}$ \\
         \hline
    \end{tabular}
    \label{tab:l}
\end{table*}

\subsubsection{Homogeneous Deformation Decomposition}
{\color{black}Matrix $\mathbf{Q}(t)$ in \eqref{homogtrans} can be decomposed as
\begin{equation}
\label{DECOM}
        \mathbf{Q}(t)=\mathbf{R}_D(t)\mathbf{U}_D(t),
\end{equation}
where
\begin{subequations}
\label{ROTDEF}
\begin{equation}
    \mathbf{R}_D(t)=\mathbf{R}\left(\phi_r(t),\theta_r(t),\psi_r(t)\right),
\end{equation}
\begin{equation}
    \mathbf{U}_D(t)=\sum_{i=1}^3\lambda_i\hat{\mathbf{u}}_i\left(\phi_u(t),\theta_u(t),\psi_u(t)\right)\hat{\mathbf{u}}_i^T\left(\phi_u(t),\theta_u(t),\psi_u(t)\right),
\end{equation}
\begin{equation}
\label{ui}
i=1,2,3,\qquad    \hat{\mathbf{u}}_i=\mathbf{R}^T\left(\phi_u(t),\theta_u(t),\psi_u(t)\right)\hat{\mathbf{e}}_i.
\end{equation}
\end{subequations}
{\color{black}Note that $\hat{\mathbf{u}}_1$, $\hat{\mathbf{u}}_2$, and $\hat{\mathbf{u}}_3$ are the eigenvectors of $\mathbf{U}_D$ while $\hat{\mathbf{e}}_1$, $\hat{\mathbf{e}}_2$, and $\hat{\mathbf{e}}_3$ are the base vectors of the inertial coordinate system defined in Section \ref{Position Notations}.}
A desired homogeneous transformation \eqref{homogtrans} can thus be uniquely expressed by the following features: (i) Rotation angles $\phi_r{\color{black}(t)}$, $\theta_r{\color{black}(t)}$, $\psi_r{\color{black}(t)}$, (ii) {\color{black}Deformation} eigenvalues $\lambda_1{\color{black}(t)}$, $\lambda_2{\color{black}(t)}$, $\lambda_3{\color{black}(t)}$, (iii) {\color{black}Deformation} angles $\phi_u$, $\theta_u$, and $\psi_u$, and {\color{black}(iv)} Rigid body displacement components $d_1{\color{black}(t)}$, $d_2{\color{black}(t)}$, and $d_3{\color{black}(t)}$. {\color{black}{\color{black}The} decomposition of $n$-D homogeneous deformation coordination {\color{black}into such features} is detailed in Appendix \ref{APA}, where $n=1,2,3$.} 
The features are all real-valued signals with {\color{black}ranges} given in Table \ref{tab:l}.
}

\section{\hspace{0.1cm}Inter-Agent Communication}\label{Inter-Agent Communication}

Suppose an MVS consists of $N$ vehicles moving in a 3-D motion space. The set $\mathcal{V}_R=\{1,2,\cdots,N\}$ {\color{black}defining} identification (index) numbers of the vehicles {\color{black}is expressed as} $\mathcal{V}_R=\mathcal{V}_L\bigcup \mathcal{V}_F$ where $\mathcal{V}_L$ and $\mathcal{V}_F$ define index numbers of leaders and followers, respectively. The paper considers cases in which vehicles are distributed on an $n$-D ($n=1,2,3$) {\color{black}Euclidean} space in $\mathbb{R}^3$. The MVS is guided by $n+1$ leaders with index numbers $\mathcal{V}_L=\{1,2,\cdots,n+1\}$.
Followers' index numbers are defined by the set $ \mathcal{V}_F=\{n+2,\cdots,N\}$.

\begin{figure}
\center
\includegraphics[width=2.0 in]{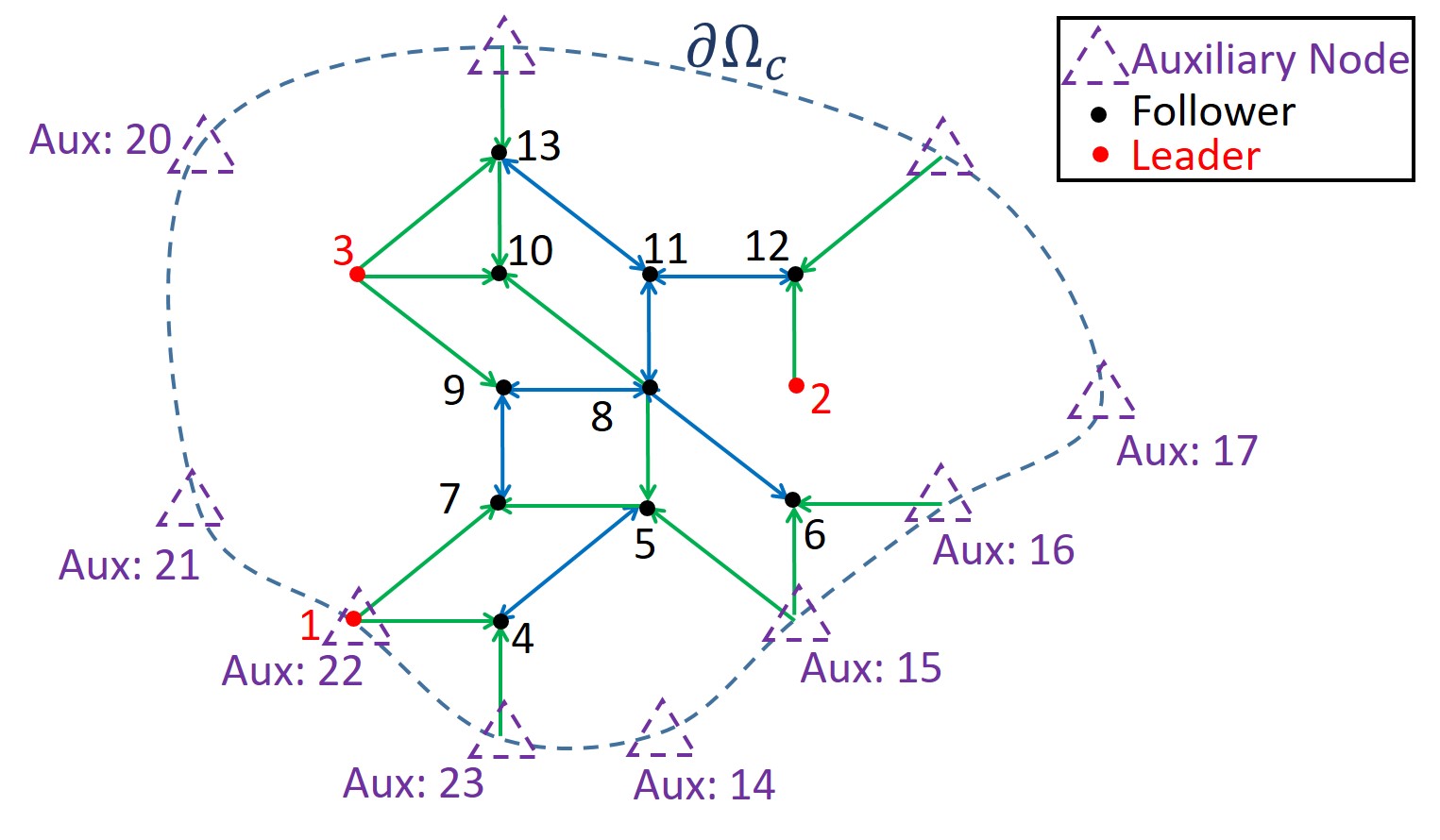}
\caption{Schematic of a communication graph with real and auxiliary nodes used in a $2$-D continuum deformation.  }
\label{schematicauxiliaryyyy}
\end{figure}

{\color{black}Let $\Omega_c$ be an arbitrary closed domain enclosing  all real vehicles at a reference configuration, where $N_a$ auxiliary nodes are arbitrarily distributed on the boundary $\partial \Omega_c$ and identified by the set $\mathcal{V}_{aux}=\{N+1,N+2,\cdots,N+N_a\}$.
}
{\color{black}
\begin{remark}
Auxiliary nodes do not represent real agents and they are defined only to ensure MVS collective motion stability. 
\end{remark}
}

\subsection{Reference Communication Weights and Weight Matrix} 
Inter-agent communication is defined by graph $\mathcal{G}_w=\mathcal{G}_w\left(\mathcal{V},\mathcal{E}_w\right)$ with  node set $\mathcal{V}$ and edge set $\mathcal{E}_w\in \mathcal {V}\times \mathcal{V}$.  $\mathcal{V}$ defines real and auxiliary (virtual) agents, e.g.  $\mathcal{V}=\mathcal{V}_R\bigcup\mathcal{V}_{aux}$ where $\mathcal{V}_R$ and $\mathcal{V}_{aux}$ define real and auxiliary vehicle index numbers, respectively. {\color{black}For every node $i\in \mathcal{V}$, \textit{reference in-neighbor} set $\mathcal{N}_i=\{j\in \mathcal{V}\big|(j,i)\in \mathcal{E}_w\}$ defines the in-neighbor nodes in the reference configuration.}

{\color{black}An example communication graph $\mathcal{G}_w$ for a $2$-D MVS coordination is shown in Fig. \ref{schematicauxiliaryyyy}. Real nodes are defined by $\mathcal{V}_{R}=\{1,\cdots,13\}$, where $\mathcal{V}_L=\{1,2,3\}$ and $\mathcal{V}_F=\{4,\cdots,13\}$ define leaders and followers, respectively. $\mathcal{V}_{aux}=\{14,\cdots,23\}$ defines auxiliary nodes. {\color{black}The set of all nodes in given by $\mathcal{V}=\{1,\cdots,23\}$.} Each auxiliary node communicates with all three leaders, {\color{black}and} communication between auxiliary nodes and leaders is not shown in Fig. \ref{schematicauxiliaryyyy}. {\color{black}An} auxiliary node {\color{black}may or may not} be {\color{black}coincident with} a real node positioned at boundary $\partial \Omega_c$ at reference time $t_0$. As shown in Fig. \ref{schematicauxiliaryyyy}, real agent 1 and auxiliary node 22 are {\color{black}coincident}.}
{\color{black}
\begin{assumption}
It is assumed that every leader $i\in \mathcal{V}_L$ moves independently. Therefore, $\mathcal{N}_i=\emptyset$, if $i\in \mathcal{V}_L$. {\color{black} Furthermore, considering that leaders are in-neighbors of auxiliary nodes,  we can say
\begin{equation}
    i\in \mathcal{V}_{aux},\qquad \mathcal{N}_i=\mathcal{V}_L.
\end{equation}}
\end{assumption}
\begin{assumption}\label{graphassumption}
{\color{black}Graph $\mathcal{G}_w\left(\mathcal{V},\mathcal{E}_w\right)$ defines at least one directed path from every leader $j\in \mathcal{V}_L$ toward vehicle $i\in \mathcal{V}_F\bigcup \mathcal{V}_{aux}$, where $\left|\mathcal{N}_i\right|={\color{black}n+1}$, i.e. every non-leader vehicle communicates with three in-neighbors defined by $\mathcal{N}_i$.}
\end{assumption}
}

{\color{black}

\begin{assumption}\label{graphpositive}
It is assumed that in-neighbor vehicles  $i_1$, $i_2$, $\cdots$, $i_{n+1}$ form an $n$-D simplex enclosing follower $i\in \mathcal{V}_F$ in the reference configuration. Therefore,
\begin{subequations}
\begin{equation}
\label{INNEIGHBOR}
\begin{split}
   \forall i\in \mathcal{V}_R,\qquad \Psi_n\left(\mathbf{r}_{i_1,0},\cdots,\mathbf{r}_{i_{n+1},0}\right)=n,
\end{split}
\end{equation}
\begin{equation}
\begin{split}
   \forall i\in \mathcal{V}_R,\qquad \mathbf{\Theta}_n\left(\mathbf{r}_{i_1,0},\cdots,\mathbf{r}_{i_{n+1},0},\mathbf{r}_{i,0}\right)>\mathbf{0}.
\end{split}
\end{equation}
\end{subequations}

\end{assumption}
}

{\color{black}Defining reference in-neighbor set of follower $i\in  \left(\mathcal{V}_F\bigcup \mathcal{V}_{aux}\right)$ as $\mathcal{N}_i=\{i_1,\cdots,i_{n+1}\}$,}
the communication weight between  $i\in \left(\mathcal{V}_F\bigcup \mathcal{V}_{aux}\right)$ and in-neighbor vehicle $i_k\in \mathcal{V}$ ($k=1,2,\cdots,n+1$) is denoted by $w_{i,i_k}$ and obtained as follows:
{\color{black}
\begin{equation}
    \label{communicationwitfollowers}
       \begin{bmatrix}
        w_{i,i_1}&
        \cdots&
        w_{i,i_{n+1}}
    \end{bmatrix}
    ^T
    =
    \mathbf{\Theta}_n\left(\mathbf{r}_{i_1,0},\cdots,\mathbf{r}_{i_{n+1},0},\mathbf{r}_{i,0}\right),
\end{equation}
where $ n=1,2,3$ is the dimension of the homogeneous deformation coordination and $\mathbf{\Theta}_n$ is defined by  \eqref{Thetaaan}.
}
\begin{remark}
If vehicle $i$ is a follower ($i\in \mathcal{V}_F$), communication weights $w_{i,i_1}$ through $w_{i,i_{n+1}}$ are all positive. This is because the communication simplex defined by $i_1$, $i_2$, $\cdots$, $i_{n+1}$ encloses follower $i\in \mathcal{V}_F$.
\end{remark}

 We define the weight matrix $\mathbf{W}=\left[W_{ij}\right]\in \mathbb{R}^{\left(N+N_a\right)\times \left(N+N_a\right)}$ as follows:
\begin{equation}
\label{comweightdefinition}
    \mathbf{W}_{ij}=
    \begin{cases}
    -1&i=j\\
    w_{i,j}>0&i\in \left(\mathcal{V}_F\bigcup\mathcal{V}_{aux}\right),~j\in \mathcal{N}_i\\
    0&\mathrm{{\color{black}otherwise}}.
    \end{cases}
\end{equation}
The matrix $\mathbf{W}$ can be partitioned as follows:
\begin{equation}
    \mathbf{W}=
    \begin{bmatrix}
    -\mathbf{I}_{n+1}&\mathbf{0}_{\left(n+1\right)\times\left(N-n-1\right)}&\mathbf{0}_{\left(n+1\right)\times N_a}\\
    \mathbf{W}_{f,l}&\mathbf{A}&\mathbf{W}_{f,a}\\
    \mathbf{W}_{a,l}&\mathbf{0}_{N_a\times\left(N-n-1\right)}&-\mathbf{I}_{N_a}\\
    \end{bmatrix}
    ,
\end{equation}
where $\mathbf{I}_{n+1}\in \mathbb{R}^{\left(n+1\right)\times \left(n+1\right)}$ and $\mathbf{I}_{N_a}\in \mathbb{R}^{N_a\times N_a}$ are the identity matrices,  $\mathbf{0}_{\left(n+1\right)\times \left(N-n-1\right)}\in \mathbb{R}^{\left(n+1\right)\times \left(N-n-1\right)}$, $\mathbf{0}_{\left(n+1\right)\times N_a}\in \mathbb{R}^{\left(n+1\right)\times N_a}$, and $\mathbf{0}_{N_a\times \left(N-n-1\right)}\in \mathbb{R}^{N_a\times \left(N-n-1\right)}$ are the zero-entry matrices, $\mathbf{W}_{f,l}\in \mathbb{R}^{\left(N-n-1\right)\times \left(n+1\right)}$ and $\mathbf{W}_{f,a}\in \mathbb{R}^{\left(N-n-1\right)\times N_a}$ are non-negative matrices, and $\mathbf{A}\in \mathbb{R}^{\left(N-n-1\right)\times \left(N-n-1\right)}$. Also,
\begin{equation}
    \mathbf{W}_{a,l}
    =
    \begin{bmatrix}
    \mathbf{\Theta}_n^T\left(\mathbf{r}_{1,0},\cdots,\mathbf{r}_{n+1,0},\mathbf{r}_{N+1,0}\right)\\
    \vdots\\
    \mathbf{\Theta}_n^T\left(\mathbf{r}_{1,0},\cdots,\mathbf{r}_{n+1,0},\mathbf{r}_{N+N_a,0}\right)\\
    \end{bmatrix}
    \in \mathbb{R}^{N_a\times 3},
\end{equation}
is one-sum row. where $\alpha_{i,k}$ ($i\in \mathcal{V}_{aux}$, $k\in \mathcal{V}_L$) is assigned by Eq. \eqref{communicationwithleaders}. 
\begin{proposition}\label{proposssssition1}
Let $\mathbf{z}_{q,l,0}=[q_{1,0}~\cdots~q_{n+1,0}]^T$, $\mathbf{z}_{q,f,0}=[q_{n+2,0}~\cdots~q_{N,0}]^T$, and $\mathbf{z}_{q,aux,0}=[q_{N+1,0}~\cdots~q_{N+N_a,0}]^T$
define component $q\in \{x,y,z\}$ of leaders, followers, and auxiliary nodes, respectively. Then, $\mathbf{z}_{q,f,0}$ is related to $\mathbf{z}_{q,l,0}$ by
\begin{equation}
\label{leaderfollowerreferencerelation}
\mathbf{L}
\begin{bmatrix}
\mathbf{z}_{q,l,0}\\
\mathbf{z}_{q,f,0}\\
\end{bmatrix}
=\begin{bmatrix}
-\mathbf{z}_{q,l,0}\\
{\color{black}\mathbf{0}}
\end{bmatrix}
\end{equation}
where
\begin{equation}
\label{LLLLL}
    \mathbf{L}=
    \begin{bmatrix}
 -\mathbf{I}_{n+1}&\mathbf{0}_{\left(n+1\right)\times\left(N-n-1\right)}\\
    {\color{black}\mathbf{B}}&\mathbf{A}\\
\end{bmatrix}
,
\end{equation}
and 
\begin{equation}
    {\color{black}\mathbf{B}=\mathbf{W}_{f,l}+\mathbf{W}_{f,a}\mathbf{W}_{a,l}.}
\end{equation}
\end{proposition}
\begin{proof}
When {\color{black}reference} communication weights 
and $\alpha$ parameters 
are consistent with agents' reference positions and assigned using Eqs. \eqref{communicationwitfollowers} and \eqref{communicationwithleaders},
\begin{equation}
\label{INEFFICIENTT}
\mathbf{W}
\begin{bmatrix}
\mathbf{z}_{q,l,0}^T&
\mathbf{z}_{q,f,0}^T&
\mathbf{z}_{q,aux,0}^T
\end{bmatrix}
^T
=\mathbf{0}.
\end{equation}
Substituting $\mathbf{z}_{q,aux,0}= \mathbf{W}_{a,l}\mathbf{z}_{q,l,0}$, Eq. \eqref{INEFFICIENTT} simplifies to {\color{black}relation \eqref{leaderfollowerreferencerelation}.}

\end{proof}



{\color{black}
\begin{theorem}\label{theorem:3}
If Assumptions \ref{leaderrankcondition}, \ref{graphassumption}, and \ref{graphpositive} are satisfied, then the following properties hold:
\begin{enumerate}
    \item{Matrix $\mathbf{A}\in \mathbb{R}^{\left(N-n-1\right)\times \left(N-n-1\right)}$ is Hurwitz.}
    \item{For an arbitrary placement of the auxiliary agents on $\partial \Omega_c$, 
\begin{equation}
\label{originalWL}
    \mathbf{W}_L=\mathbf{A}^{-1}{\color{black}\mathbf{B}}\in \mathbb{R}^{\left(N-n-1\right)\times \left(n+1\right)}
\end{equation}
is {\color{black} i.e., } {\color{black}the} sum of {\color{black}the} row-elements is one for every row of matrix $\mathbf{W}_L$, where $\Omega_c$ is an arbitrary closed domain enclosing MVS reference configuration ($\Omega_0\subset\Omega_c$).}
\end{enumerate}

\end{theorem}

\begin{proof}
{\color{black}By Assumption \ref{graphassumption}}, the position information can be transmitted from a leader to every follower. Consequently, $\mathbf{W}\in\mathbb{R}^{\left(N+N_a\right)\times \left(N+N_a\right)}$ is not reducible. {\color{black}By Assumption \ref{graphpositive} and due to how communication weights are defined by Eq. \ref{communicationwitfollowers}},  {\color{black}it follows that:} (i) {\color{black}The sum} of the elements {\color{black}is} zero {\color{black}for} rows $n+2$ through $N+N_a$ of $\mathbf{W}$. (ii) Except for diagonal elements of $\mathbf{W}$ that are all $-1$, off-{\color{black}diagonal} elements are non-negative at rows $n+2$ through $N$. (ii) Because $\mathbf{W}$ is not reducible,  the sums of the row elements are negative in at least $n+1$ rows of matrix $\mathbf{A}\in \mathbb{R}^{\left(N-n-1\right)\times\left(N-n-1\right)}$. Hence, the matrix $\mathbf{A}$ can be expressed as  $\mathbf{A}=-\mathbf{I}_{N-n-1}+\mathbf{H}$ where $\mathbf{H}\in \mathbb{R}^{\left(N-n-1\right)\times\left(N-n-1\right)}$ has {\color{black}a} {\color{black}spectral} radius $\rho<1$, i.e. eigenvalues of matrix $\mathbf{H}$ are all located inside a disk with radius $\rho<1$ centered at the origin {\color{black}on a complex plane}. Therefore, $-\mathbf{A}$ is {\color{black}a} nonsingular M-matrix \cite{qu2009cooperative} {\color{black}and} matrix $\mathbf{A}$ is Hurwitz.

Because Assumption  \ref{leaderrankcondition} is satisfied, {\color{black}the} position of follower $i\in \mathcal{V}_F$ can be expressed with respect to leaders as given in Eq. \eqref{leadersexpression} where $\alpha_{1,1}$, $\cdots$, $\alpha_{i,n+1}$ are unique and assigned by Eq. \eqref{communicationwithleaders}. In other words, component $q$ of {\color{black}the} followers' and leaders' reference positions can be related by
$
\mathbf{z}_{q,f,0}=\mathbf{W}_L\mathbf{z}_{q,l,0},
$
where
\begin{equation}
\label{WL}
\mathbf{W}_L=\begin{bmatrix}
\mathbf{\Theta}_n^T\left(\mathbf{r}_{1,0},\cdots,\mathbf{r}_{n+1,0},\mathbf{r}_{n+2,0}\right)\\
\vdots\\
\mathbf{\Theta}_n^T\left(\mathbf{r}_{1,0},\cdots,\mathbf{r}_{n+1,0},\mathbf{r}_{N,0}\right)\\
\end{bmatrix}=
\begin{bmatrix}
\alpha_{n+2,1}&\cdots&\alpha_{n+2,n+1}\\
\vdots&\ddots&\vdots\\
\alpha_{N,1}&\cdots&\alpha_{N,n+1}\\
\end{bmatrix}
.
\end{equation}
Because $\alpha$-parameters $\alpha_{i,1}$, $\cdots$, $\alpha_{i,n+1}$ are assigned by Eq. \eqref{communicationwithleaders}, $\mathbf{W}_L$ is one-sum row. On the other hand, $\mathbf{z}_{q,f,0}$ and $\mathbf{z}_{q,l,0}$ can be also related {\color{black}with} Eq. \eqref{leaderfollowerreferencerelation}:
\begin{equation}
\label{findrel}
    \mathbf{z}_{q,f,0}=\mathbf{A}^{-1}{\color{black}\mathbf{B}} \mathbf{z}_{q,,l,0}.
\end{equation}
By equating right-hand sides of Eqs. \eqref{WL} and \eqref{findrel}{\color{black}, {\color{black}the following {\color{black}properties hold:} (i) $\mathbf{W}_L$ is} obtained as given in Eq. \eqref{originalWL}, and (ii) ${\mathbf{W}_L}_{i,j}=\alpha_{i+n+1,j}$ ($(i+n+1)\in \mathcal{V}_F,~j\in \mathcal{V}_L)$.}
\end{proof}
}


\begin{remark}\label{RM3}
  Let $\mathbf{z}_{q,l,HT}\left(t\right)=[q_{1,HT}~\cdots~q_{n+1,HT}]^T\in \mathbb{R}^{n+1}$ and $\mathbf{z}_{q,f,HT}\left(t\right)\left(t\right)=[q_{n+2,HT}~\cdots~q_{N,HT}]^T\in \mathbb{R}^{N-n-1}$ denote the vector formed from component $q$ ($q\in\{x,y,z\}$) of {\color{black}the} global desired positions of leaders and followers, e.g. $\mathbf{r}_{i,HT}=x_{i,HT}\hat{\mathbf{e}}_1+y_{i,HT}\hat{\mathbf{e}}_2+z_{i,HT}\hat{\mathbf{e}}_3$ is the global desired position of vehicle $i\in \left(\mathcal{V}_L\bigcup\mathcal{V}_F\right)$. Then, $\mathbf{z}_{q,f,HT}$ and $\mathbf{z}_{q,l,HT}$ are related by
 \begin{equation}
 \label{followerglobaldesired}
 t\geq t_s,\qquad    \mathbf{z}_{q,f,HT}(t)=\mathbf{W}_L\mathbf{z}_{q,l,HT}(t).
 \end{equation}
 Therefore, 
  \[
j=0,1,\cdots,\rho_q,\qquad \dfrac{d^j\mathbf{z}_{q,HT}}{dt^j}=\begin{bmatrix}
    \mathbf{I}_{n+1}\\
    \mathbf{W}_L
\end{bmatrix}
\begin{bmatrix}
\frac{d^{j}q_{1,HT}}{dt^j}\\
\vdots\\
\frac{d^{j}q_{n+1,HT}}{dt^j}\\
\end{bmatrix}
,
 \]
 where  $
 \mathbf{z}_{q,HT}(t)=
 \begin{bmatrix}
     q_{1,HT}(t)&
     \cdots&
     q_{N,HT}(t)
 \end{bmatrix}
 ^T
$.
\end{remark}\label{RM4}

{\color{black}
\begin{assumption}\label{leaderglobaldesired}
It is assumed that leader vehicles  $1$, $2$, $\cdots$, ${n+1}$ form an $n$-D simplex at any time $t\in [t_s,t_f]$, where $t_s$ and $t_f$ denote the initial and final times, respectively. Therefore,
\begin{equation}
\begin{split}
 \Psi_n\left(\mathbf{r}_{1,HT}{\color{black}(t)},\cdots,\mathbf{r}_{{n+1},HT}{\color{black}(t)}\right)=n.
\end{split}
\end{equation}
\end{assumption}
}

{\color{black}\subsection{Coordination Graph and Real Communication Weights}
As {\color{black}mentioned in Section \ref{Inter-Agent Communication}}, auxilliary nodes do not move and they are just defined in the reference configuration in order to obtain communication weights and ensure stability of collective motion. {\color{black}When the MVS is moving}, follower vehicles only communicate with real in-neighbor where inter-agent communication is defined by coordination graph $\mathcal{G}_c\left(\mathcal{V}_R,\mathcal{E}_c\right)$ with edge set $\mathcal{E}_c\subset \mathcal{V}_R\times \mathcal{V}_R$. \underline{\textit{Real in-neighbors}} of follower $i\in \mathcal{V}_F$ {\color{black}are} defined by time-invariant set 
$\mathcal{I}_i=\{j\big|(j,i)\in \mathcal{E}_c\}$ which is called \textit{real in-neighbor set}. We define the local desired position of vehicle $i\in \mathcal{V}_R$ by
\begin{equation}
     \label{localdesiredpositionnnnnnnnn}
         \mathbf{r}_{d,i}=
         \begin{cases}
         \mathbf{r}_{i,HT}&i\in \mathcal{V}_L\\
        \sum_{j\in \mathcal{I}_i}\varpi_{i,j}\mathbf{r}_j=&i\in \mathcal{V}_F
         \end{cases}
         ,
     \end{equation}
 where $\varpi_{i,j}>0$ is the \underline{\textit{real communication weight}}  between follower $i\in \mathcal{V}_F$ and vehicle $j\in \mathcal{I}_i\subset \mathcal{V}_R$ and 
 \[
 \sum_{j\in \mathcal{I}_i}\varpi_{i,j}=1.
 \]

By considering Proposition \eqref{proposssssition1} and Theorem \eqref{theorem:3}, the following properties hold about the real in-neighbor set $\mathcal{I}_i$ and real communication weight $\varpi_{i,j}$:
\begin{enumerate}
    \item{If $\mathcal{N}_i\bigcap \mathcal{V}_{aux}=\emptyset$, then, $\mathcal{I}_i=\mathcal{N}_i$ and $\varpi_{i,j}=w_{i,j}$ for every vehicle $i\in \mathcal{V}_F$ with every in-neighbor vehicle $j\in \mathcal{I}_i=\mathcal{N}_i$.}
     \item{If $\mathcal{N}_i\bigcap \mathcal{V}_{aux}\neq\emptyset$, then, $\mathcal{I}_i\neq\mathcal{N}_i$ and $\varpi_{i,j}$ is defined as follows:
     \[
     \varpi_{i,j}=
     \begin{cases}
     w_{i,j}&j\in (\mathcal{I}_i\bigcap \mathcal{N}_i)\\
     \sum_{h\in \left(\mathcal{N}_i\bigcap \mathcal{V}_{aux}\right)}w_{i,h}\alpha_{h,j}&\mathrm{otherwise}\\
     \end{cases}
     .
     \]
     }
\end{enumerate}
}

\section{\hspace{0.2cm}Problem Formulation}
\label{Problem Statement}
{\color{black}This paper studies the properties of our homogeneous deformation approach for an $N$-vehicle MVS {\color{black}where vehicles are} {\color{black}treated as particles of an $n$-D deformable body ($n=1,2,3$)}. The desired {\color{black}MVS vehicle positions} are guided by $n+1$ leaders {\color{black}and acquired by the remaining followers through local communication, where followers {\color{black}are} {\color{black}arbitrarily} distributed in a $3D$ motion space.} 
}
{\color{black}Dynamics of the vehicle $i\in \mathcal{V}_R$ is {\color{black}represented} by a nonlinear model}
 \begin{equation}\label{Dynamicsofagenti}
     \begin{cases}
     \dot{\mathbf{x}}_i=\mathbf{f}_i\left(\mathbf{x}_i\right)+\mathbf{g}_i\left(\mathbf{x}_i\right)\mathbf{u}_i\\
     \mathbf{r}_i=\mathbf{r}_i\left(\mathbf{x}_i\right),
     \end{cases}
\end{equation}
where $\mathbf{u}_i\in \mathcal{U}_i\subset \mathbb{R}^{n_u}$ and $\mathcal{U}_i$ is compact.  
In Eq. \eqref{Dynamicsofagenti}, $\mathbf{x}_i\in \mathbb{R}^{n_{x}\times 1}$ and $\mathbf{u}_i\in \mathbb{R}^{n_{u}\times 1}$ are state and control input, respectively, {\color{black}the} actual position $\mathbf{r}_i\in \mathbb{R}^{3\times 1}$ is the output of vehicle dynamics \eqref{Dynamicsofagenti}, and $\mathbf{f}_i:\mathbb{R}^{n_x}\rightarrow \mathbb{R}^{n_x}$ and $\mathbf{g}_i:\mathbb{R}^{n_x}\rightarrow \mathbb{R}^{n_x\times n_u}$ are smooth. 
{\color{black}Defining  $\mathbf{F}_L=\left[\mathbf{f}_1^T~\cdots~\mathbf{f}_{n+1}^T\right]^T$,  $\mathbf{F}_F=\left[\mathbf{f}_{n+2}^T~\cdots~\mathbf{f}_{N}^T\right]^T$, $\mathbf{G}_L=\left[\mathbf{g}_1^T~\cdots~\mathbf{g}_{n+1}^T\right]^T$,  $\mathbf{G}_F=\left[\mathbf{g}_{n+2}^T~\cdots~\mathbf{g}_{N}^T\right]^T$, the leader state vector by $\mathbf{X}_L=\left[\mathbf{x}_1^T~\cdots~\mathbf{x}_{n+1}^T\right]^T$, the follower state vector by $\mathbf{X}_F=\left[\mathbf{x}_{n+2}^T~\cdots~\mathbf{x}_{N}^T\right]^T$, the leader input vector by
$\mathbf{U}_L=\left[\mathbf{u}_1^T~\cdots~\mathbf{u}_{n+1}^T\right]^T$, the follower input vector $\mathbf{U}_F=\left[\mathbf{u}_{n+2}^T~\cdots~\mathbf{u}_{N}^T\right]^T$, the MVS collective dynamics can be expressed by
\begin{equation}
\begin{cases}
    \dot{\mathbf{X}}_L=\mathbf{F}_L\left(\mathbf{X}_L\right)+\mathbf{G}_L\left(\mathbf{X}_L\right)\mathbf{U}_L\\
    \dot{\mathbf{X}}_F=\mathbf{F}_F\left(\mathbf{X}_F\right)+\mathbf{G}_F\left(\mathbf{X}_F\right)\mathbf{U}_F\\
\end{cases}
.
\end{equation}
We also define
\[
\begin{split}
    \mathbf{P}_L=\mathrm{vec}\left(
    \begin{bmatrix}
    \mathbf{r}_1&\cdots&\mathbf{r}_{n+1}&\cdots&\frac{d^\Xi\mathbf{r}_1}{ dt}&\cdots&\frac{d^\Xi\mathbf{r}_{n+1}}{ dt^\Xi}
    \end{bmatrix}^T\right)
    \\
    \mathbf{P}_F=\mathrm{vec}\left(
    \begin{bmatrix}
    \mathbf{r}_{n+2}&\cdots&\mathbf{r}_{N}&\cdots&\frac{d^\Xi \mathbf{r}_{n+2}}{ dt}&\cdots&\frac{d^\Xi\mathbf{r}_{N}}{dt^\Xi}
    \end{bmatrix}^T\right)
\end{split}
\]
are the outputs of follower and leader MVS dynamics, respectively. ``vec'' is the matrix vectorization operator. Also,
\[
\begin{split}
\resizebox{0.99\hsize}{!}{%
$
    \mathbf{P}_{L,d}=\mathrm{vec}\left(
    \begin{bmatrix}
    \mathbf{r}_{1,HT}&\cdots&\mathbf{r}_{n+1,HT}&\cdots&\frac{d\mathbf{r}_{1,HT}}{ dt}&\cdots&\frac{d^\Xi\mathbf{r}_{n+1,HT}}{ dt^\Xi}
    \end{bmatrix}^T\right)
    $
    }
\end{split}
\]
in the reference input of the MVS control system, and
\[
\begin{split}
    \mathbf{P}_{F,d}(t)=\mathbf{D}_{\mathrm{MVS}}\mathbf{P}_F(t)+\mathbf{B}_{\mathrm{MVS}}\mathbf{P}_L(t)
\end{split}
\]
is the reference input of the follower MVS dynamics, where $\mathbf{D}_{\mathrm{MVS}}=\mathbf{I}_{3\Xi}\otimes\mathbf{D}$, $\mathbf{B}_{\mathrm{MVS}}=\mathbf{I}_{3\Xi}\otimes\mathbf{B}$, $\mathbf{D}=\mathbf{I}_{N-n-1}+\mathbf{A}$, $\otimes$ is the Koronecker product symbol, $\Xi\leq n_x$ is known, and $\mathbf{I}_{3\Xi}\in \mathbb{R}^{3\Xi\times 3\Xi}$ and $\mathbf{I}_{N-n-1}\in \mathbb{R}^{\left(N-n-1\right)\times \left(N-n-1\right)}$ are the identity matrices.

\begin{figure}
\center
\includegraphics[width=3.2 in]{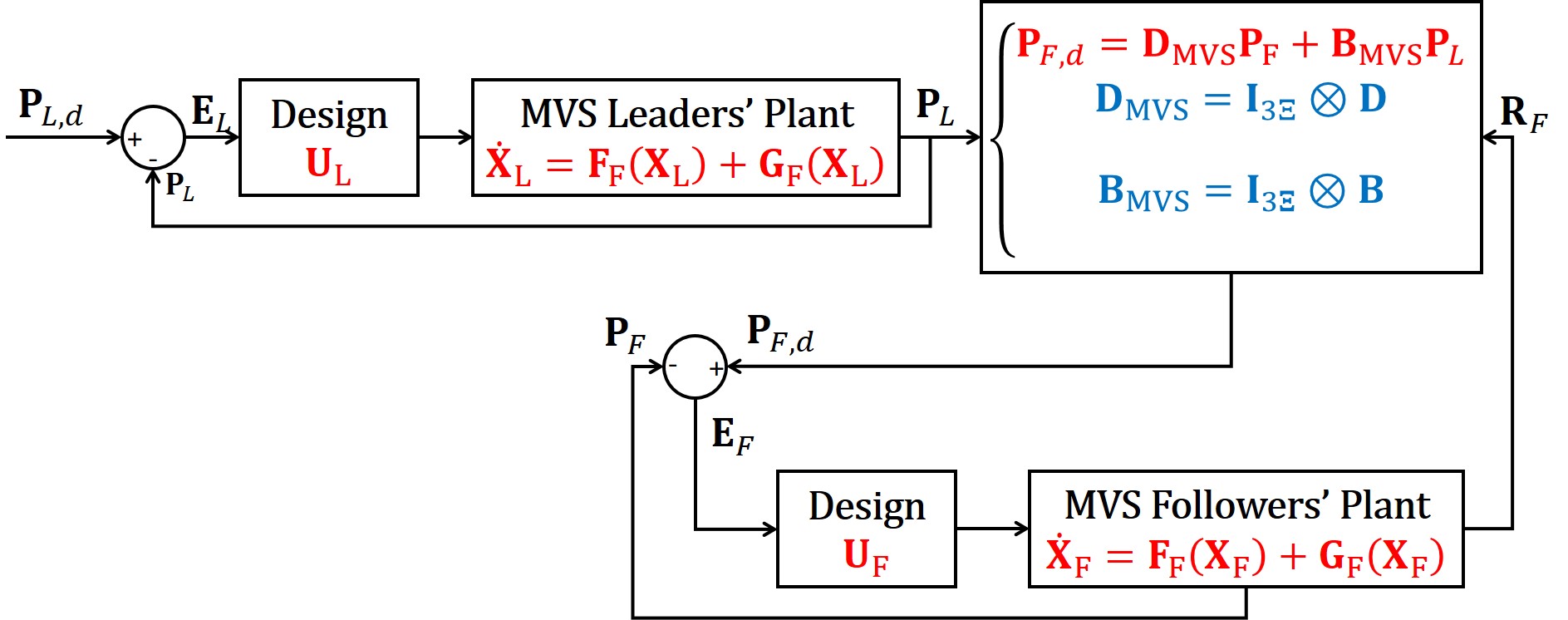}
\caption{Continuum deformation coordination block diagram.}
\label{1stcontrollerblock}
\end{figure}
This paper studies the following problems:

\textbf{\textit{Problem 1 (MVS Continuum Deformation {\color{black}Coordinated Control}):}} {\color{black}The first problem to be addressed is to determine} $\mathbf{U}_L(t)$ and $\mathbf{U}_F(t)$ such that {\color{black}position} deviation of every vehicle $i\in \mathcal{V}_L\bigcup \mathcal{V}$ is less than or equal to $\delta>0$ at any time $t$ ($\|\mathbf{r}_i(t)-\mathbf{r}_{i,HT}\|\leq \delta,~\forall i\in \mathcal{V},~t\in [t_s,t_f]$).
Note {\color{black}that the} local desired position $\mathbf{r}_{d,i}$, defined in Section \ref{Position Notations} is the reference input for every vehicle $i\in \mathcal{V}_L\bigcup \mathcal{V}_F$. Therefore, leader $i\in \mathcal{V}_L$ knows the global desired position $\mathbf{r}_{i,HT}$ while follower $i\in \mathcal{V}$ does not know $\mathbf{r}_{i,HT}$ but acquires it through local communication. The block diagram of the MVS collective dynamics is shown in Fig. \ref{1stcontrollerblock}. 
{\color{black}
\begin{assumption}\label{feedbacklinearizable}
{\color{black}The following assumption is made:} Dynamics of every vehicle $i\in \mathcal{V}_L\bigcup \mathcal{V}_F$ is input-output linearizable. 
\end{assumption}
{\color{black}Stability and convergence of the MVS collective dynamics are investigated in Section \ref{SimpleSafety}.}
}

\textbf{\textit{Problem 2 ({\color{black}Guaranteeing} Continuum Deformation Safety Specification):}} {\color{black}The objective is to mathematically} specify safety conditions for large-scale MVS coordination in obstacle-{\color{black}free} and obstacle-{\color{black}laden} motion spaces.  In particular, a large-scale coordination is labelled safe if inter-agent collision is avoided, obstacle collision avoidance is avoided, and vehicle input constraints are satisfied. {\color{black}This paper also advances the existing condition for assurance of inter-agent collision avoidance in a large-scale continuum deformation coordination by proposing a new {\color{black}and less restrictive} inter-agent collision avoidance condition. {\color{black}Specifically, we} show that how inter-agent collision can be avoided only by constraining the {\color{black}one of the eigenvalues} of the deformation matrix, if the shear-deformation angles remain constant at all times $t$. {\color{black}This can significantly improve the flexibility and maneuverability of the large-scale continuum deformation coordination.}}
\\

\textbf{\textit{Problem 3 (MVS Continuum Deformation Coordination Planning):}} 
{\color{black}For the planning of a continuum deformation coordination, the global desired trajectories of leaders are assigned in obstacle-free and obstacle-laden environments such that the initial conditions at time $t_s$, final conditions at time $t_f$, and  continuum deformation safety condition at any time $t\in [t_s,t_f]$ are satisfied. We define the} desired trajectory of every leader $i\in \mathcal{V}_L$ by
\begin{equation}
\label{riiihttttt}
   \mathbf{r}_{i,HT}=\mathbf{r}_{i,HT}\left(\mathbf{s}_\varrho^n(t)\right), 
\end{equation}
where 
\begin{equation}
\label{generalplease}
    \mathbf{s}_{\varrho}^n(t)=\left[s_{1,\varrho}^n(t)~\cdots~s_{g_{n,\varrho},\varrho}^n(t)\right]^T\in \mathbb{R}^{g_{n,\varrho}\times 1}
\end{equation}
is the generalized coordinate vector, $n\in\{1,2,3\}$ is the dimension of the homogeneous deformation coordination and 
\begin{equation}
    g_{n,\varrho}=
    \begin{cases}
    n(n+1)&\varrho=\mathrm{OL},~n=1,2,3\\
    6&\varrho=\mathrm{OF},~n=1\\
    8&\varrho=\mathrm{OF},~n=2\\
    9&\varrho=\mathrm{OF},~n=3\\
    \end{cases}
\end{equation}
is the number of generalized coordinates. Also,
$s_{h,\varrho}^n:[t_s,t_f]\rightarrow \mathbb{R}$ is the $h$-th generalized coordinate ($h=1,\cdots,n(n+1)$), $\varrho\in \{\mathrm{OF},\mathrm{OL}\}$ is a discrete variable, and ``$\mathrm{OF}$'' and ``$\mathrm{OL}$'' stand for ``obstacle free'' and ``obstacle laden'', respectively. Note that $\mathbf{s}_{\mathrm{OL}}^n$ {\color{black}is} defined based on leaders' position components while  $\mathbf{s}_{\mathrm{OF}}$ specifies the homogeneous transformation features. The generalized coordination vector $\mathbf{s}_\varrho^n(t)$ is given by
\begin{equation}
\label{barsss}
t\in [t_k,t_{k+1}],\qquad \mathbf{s}_{\varrho}^n(t)=\bar{\mathbf{s}}_{k,\varrho}^n(1-\beta(t,T_k))+\beta(t,T_k)\bar{\mathbf{s}}_{k+1,\varrho}^n,
\end{equation}
where $k=1,\cdots,n_\tau,~\varrho\in \{\mathrm{OL},\mathrm{OF}\}$, {\color{black}$t_1=t_s$, $t_{n_\tau+1}=t_f$,} $\bar{\mathbf{s}}_{1,\varrho}^n$, $\cdots$, $\bar{\mathbf{s}}_{n_\tau+1,\varrho}^n$ determine intermediate configurations for the generalized coordinate, $T_k=t_{k+1}-t_k$, and $\beta(t,T_k)$ is defined by the following fifth order polynomial:
\begin{equation}
 \label{poly1}
t\in [t_k,t_{k+1}],\qquad \beta(t,T_k)=\sum_{j=0}^5\zeta_{j,k}\left(\frac{t-t_k}{ T_k}\right)^5,
\end{equation}
where $\beta(t_k,T_k)\in [0,1]$ is an increasing function of $t$ over $[t_k,t_{k+1}]$, $\beta(t_k,T_k)=0$, $\beta(t_{k+1},T_k)=1$, $\zeta_{0,k}$ through $\zeta_{5,k}$ are constant and determined based on boundary conditions at times $t_k$ and $t_{k+1}$. 
Note that $T_1$ through $T_n$ are {\color{black}the} design parameters determined such that the continuum deformation safety conditions are all satisfied.

{\color{black}For continuum deformation planning in obstacle-free environments, $n_\tau=1$ and Eq. \eqref{barsss} is used to assign the global desired trajectories of the leaders by shaping the homogeneous deformation features.} For motion planning in obstacle-laden environments, $\bar{\mathbf{s}}_{1,\mathrm{OL}}^n$, $\cdots$, $\bar{\mathbf{s}}_{n_\tau+1,\mathrm{OL}}^n$ are determined using A* search  such that the leaders' travel distance is minimized. 
{\color{black}This paper determines the leaders' paths through connecting intermediate way-points assigned by the A* planner. Therefore, the additional conditions $\dot{\beta}(t_k,T_k)=0$, $\ddot{\beta}(t_k,T_k)=0$, $\dot{\beta}(t_{k+1},T_k)=0$,  $\ddot{\beta}(t_{k+1},T_k)=0$ are all satisfied in order to ensure the continuity of the desired leaders' velocity and acceleration at all times $t$.}
}

 {\color{black}
\section{\hspace{0.1cm}Problem 1: {\color{black}MVS Continuum Deformation Coordinated Control}}
\label{MAS Collective Dynamics}
{\color{black}{\color{black}By} Assumption \ref{feedbacklinearizable}, we can {\color{black}find a}} state transformation $\mathbf{x}_i\rightarrow \left(\mathbf{r}_i,\varpi_i\right)$, {\color{black}such that}  \eqref{Dynamicsofagenti} {\color{black}is decomposed into external dynamics}
\begin{equation}
\label{controllable}
 {\color{black}q\in \{x,y,z\},\qquad}   \dfrac{{d}^{\rho_q} q_i}{{d}t^{\rho_q}}={v}_{q,i}
\end{equation}
and internal dynamics
  \begin{equation}
  \label{zero}
  \dot{\omega}_i={\mathbf{f}}_{\mathrm{int},i}\left({\omega}_i,x_i,\cdots,L_f^{\rho_x-1}x_i,y_i,\cdots,L_f^{\rho_y-1}y_i,z_i,\cdots,L_f^{\rho_z-1}z_i\right),
  \end{equation}
{\color{black}where  $v_{q,i}$ is assigned as}
\begin{equation}
    {v}_{q,i}=-\sum_{k=1}^{\rho_q-\Xi-1}\gamma_{i,k}L_f^{\rho_q-k}q_i\\
     +\sum_{k=\rho_q-\Xi}^{\rho_q}\gamma_{i,k}L_f^{\rho_q-k}\left(q_{d,i}-q_{i}\right),
\end{equation}
{\color{black}and where $L_{\mathbf{F}}{G}=\left(\bigtriangledown G\right)^T \mathbf{F}$ is the Lie derivative of {\color{black}a} smooth function $G$ with respect to {\color{black}a vector field} $\mathbf{F}$, $q_{d,i}$ is the component $q\in \{x,y,z\}$ of local desired position $\mathbf{r}_{d,i}$ defined by Eq. \eqref{localdesiredpositionnnnnnnnn}. 
     } 
 Control gains $\gamma_{i,1}$ through  $\gamma_{i,\rho_q}$ are assigned such that {\color{black}the} MVS collective dynamics is stable. 
In Eq. \eqref{controllable}, $\rho_q$ ($q\in \{x,y,z\}$) is the relative degree, $\rho=\rho_x+\rho_y+\rho_z\leq n_x$ is the total relative degree, {\color{black}$\Xi< \rho_q$ is constant}, and
   \[
   \resizebox{0.99\hsize}{!}{%
$
  \begin{split}
    \begin{bmatrix}
      v_{x,i}\\
      v_{y,i}\\
      v_{z,i}
  \end{bmatrix}
  =&
  \mathbf{M}_i\mathbf{u}_i+\mathbf{n}_i,\\
  \mathbf{M}_i=&
  \begin{bmatrix}
     L_{\mathbf{g}_i}L_{\mathbf{f}_i}^{\rho_x-1}x_i\\
     L_{\mathbf{g}_i}L_{\mathbf{f}_i}^{\rho_y-1}y_i\\
     L_{\mathbf{g}_i}L_{\mathbf{f}_i}^{\rho_z-1}z_i
  \end{bmatrix}
  ,~\mathbf{n}_i=
  \begin{bmatrix}
     \left( L_{\mathbf{f}_i}^{\rho_x}x_i\right)^T&
      \left(L_{\mathbf{f}_i}^{\rho_y}y_i\right)^T&
      \left( L_{\mathbf{f}_i}^{\rho_z}z_i\right)^T
  \end{bmatrix}
  ^T.
  \end{split}
  $
  }
  \]
  {\color{black}
\begin{assumption}
For the system \eqref{Dynamicsofagenti} and \eqref{controllable}-\eqref{zero}, (i) $n_u\geq 3$, (ii) $\mathrm{Rank}\left(\mathbf{M}_i\right)=3$, and (iii) zero dynamics of vehicle $i$, $\dot{\omega}_i={\mathbf{f}}_{\mathrm{int},i}\left({\omega}_i,0,\cdots,0\right)$,  is locally asymptotically stable {\color{black}at the origin}.
\end{assumption}
}

{\color{black}Define
\[
\resizebox{0.99\hsize}{!}{%
$
\mathbf{X}_{\mathrm{SYS},q}=
\begin{bmatrix}
    q_1&\cdots&q_N&\cdots&L_{\mathbf{f}_1}^{\rho_q}q_1&\cdots&L_{\mathbf{f}_N}^{\rho_q}q_N
\end{bmatrix}
^T\in \mathbb{R}^{\rho_qN\times 1},\\
$
}
\]
\[
\mathbf{X}_{\mathrm{SYS}}=
\begin{bmatrix}
   \mathbf{X}_{\mathrm{SYS},x}^T& \mathbf{X}_{\mathrm{SYS},y}^T& \mathbf{X}_{\mathrm{SYS},z}^T
\end{bmatrix}
^T\in \mathbb{R}^{\left(\rho_x+\rho_y+\rho_z\right)N\times 1},
\]
\[
\begin{split}
 {\mathbf{U}}_{\mathrm{SYS},q}=&\sum_{j=\rho_q-\Xi}^{\rho_q}\mathbf{\Gamma}_{j,l}
    \begin{bmatrix}
        L_f^{\rho_q-j}q_{1,HT}\\
       \vdots\\
        L_f^{\rho_q-j}q_{n+1,HT}\\
    \end{bmatrix}
    ,
\end{split}
\]
\[
\mathbf{U}_{\mathrm{SYS}}=
\begin{bmatrix}
   \mathbf{U}_{\mathrm{SYS},x}^T& \mathbf{U}_{\mathrm{SYS},y}^T& \mathbf{U}_{\mathrm{SYS},z}^T
\end{bmatrix}
^T\in \mathbb{R}^{\left(\rho_x+\rho_y+\rho_z\right)(n+1)\times 1},
\]
where $q_i$ and $q_{i,HT}$ denote component $q\in\{x,y,z\}$ of actual and global desired positions of vehicle $i$. The MVS collective dynamics can be expressed by the following normal form:}
\begin{subequations}
\label{MVSD}
\begin{equation}
\label{External}
\mathrm{{External~Dynamics}:}~\dot{\mathbf{X}}_{\mathrm{SYS}}=\mathbf{A}_{\mathrm{SYS}}{\mathbf{X}}_{\mathrm{SYS}}+\mathbf{B}_{\mathrm{SYS}}{\mathbf{U}}_{\mathrm{SYS}},
\end{equation}
\begin{equation}
\label{Internal}
\mathrm{{Internal~Dynamics}:}~~
\frac{d{\omega}}{ dt}=\mathbf{F}_{\mathrm{INT}}\left(\omega,\mathbf{X}_{\mathrm{SYS}}\right),
\end{equation}
\end{subequations}
where 
$
\omega=
\begin{bmatrix}
    \omega_1^T&\cdots&\omega_N^T
\end{bmatrix}
^T
,
~
\mathbf{F}_{\mathrm{INT}}=
\begin{bmatrix}
    \mathbf{f}_{\mathrm{int},1}^T&\cdots&\mathbf{f}_{\mathrm{int},N}^T
\end{bmatrix}
^T,
$
\[
\resizebox{0.99\hsize}{!}{%
$
\mathbf{A}_{\mathrm{SYS}}=
\begin{bmatrix}
\mathbf{A}_{\mathrm{SYS},x}&\mathbf{0}&\mathbf{0}\\
\mathbf{0}&\mathbf{A}_{\mathrm{SYS},y}&\mathbf{0}\\
\mathbf{0}&\mathbf{0}&\mathbf{A}_{\mathrm{SYS},z}
\end{bmatrix}
,~
\mathbf{B}_{\mathrm{SYS}}=
\begin{bmatrix}
\mathbf{B}_{\mathrm{SYS},x}&\mathbf{0}&\mathbf{0}\\
\mathbf{0}&\mathbf{B}_{\mathrm{SYS},y}&\mathbf{0}\\
\mathbf{0}&\mathbf{0}&\mathbf{B}_{\mathrm{SYS},z}
\end{bmatrix}
,
$
}
\]
\[
\resizebox{0.99\hsize}{!}{%
$
    j=\rho_q-\Xi+1,\cdots,\rho_q,\qquad \mathbf{\Gamma}_{j,l}=\mathrm{diag}(\gamma_{j,1},\cdots,\gamma_{j,n+1})\in \mathbb{R}^{\left(n+1\right)\times \left(n+1\right)}
   ,
   $
   }
\]
\[
    j=1,\cdots,\rho_q \qquad \mathbf{\Gamma}_{j}=\mathrm{diag}(\gamma_{j,1},\cdots,\gamma_{j,N})\in \mathbb{R}^{N\times N}
   ,
\]
\[
\resizebox{0.99\hsize}{!}{%
$
q\in \{x,y,z\},\qquad  \mathbf{B}_{\mathrm{SYS},q}=
    \begin{bmatrix}
        \mathbf{0}_{\left(n+1\right)\times\left(\rho_q-1\right)N}&\mathbf{I}_{n+1}&\mathbf{0}_{\left(n+1\right)\times\left(N-n-1\right)}
    \end{bmatrix}
    ^T
$,
}
\]
\[
\resizebox{0.99\hsize}{!}{%
$
q\in \{x,y,z\},\qquad \mathbf{A}_{\mathrm{SYS},q}=
    \begin{bmatrix}
        \mathbf{0}_N&\mathbf{I}_N&\cdots&\mathbf{0}_N\\
        \vdots&\vdots&\ddots&\vdots\ddots\\
        \mathbf{0}_N&\mathbf{0}_N&\cdots&\mathbf{I}_N\\
        \mathbf{\Gamma}_{\rho_q}\mathbf{H}_{\rho_q,q}&\mathbf{\Gamma}_{\rho_q-1}\mathbf{H}_{\rho_q-1,q}&\cdots&\mathbf{\Gamma}_{1}\mathbf{H}_{1,q}
    \end{bmatrix}
$.
}
\]
Note that $\mathbf{I}_{N}\in \mathbb{R}^{N\times N}$ and $\mathbf{0}_{N}\in \mathbb{R}^{N\times N}$ {\color{black}are the identity and zero-entry matrices.}
Matrix $\mathbf{H}_{i,q}\in \mathbb{R}^{N\times N}$ ($i=1,\cdots,\rho_q$) is defined as follows:
\[
q\in \{x,y,z\},\qquad
\mathbf{H}_{i,q}=
\begin{cases}
    -\mathbf{I}_N&1\leq i\leq \rho_q-\Xi-1\\
    \mathbf{L}&\rho_q-\Xi\leq i\leq \rho_q\\
\end{cases}
.
\]

\subsection{Error Dynamics} {\color{black}Defining
\[
\resizebox{0.99\hsize}{!}{%
$
\mathbf{X}_{\mathrm{SYS},HT,q}=
\begin{bmatrix}
    q_{1,HT}&\cdots&q_{N,HT}&\cdots&L_{\mathbf{f}_1}^{\rho_q}q_{1,HT}&\cdots&L_{\mathbf{f}_N}^{\rho_q}q_{N,HT}
\end{bmatrix}
$
}
\]
\[
\mathbf{X}_{\mathrm{SYS},HT}=
\begin{bmatrix}
   \mathbf{X}_{\mathrm{SYS},HT,x}^T& \mathbf{X}_{\mathrm{SYS},HT,y}^T& \mathbf{X}_{\mathrm{SYS},HT,z}^T
\end{bmatrix},
\]
the transient error $\mathbf{E}_{\mathrm{SYS}}=\mathbf{X}_{\mathrm{SYS}}-\mathbf{X}_{\mathrm{SYS},HT}$ assigning deviation vehicles' actual position components from their global desired position components are updated by}
{
\begin{equation}
\label{ERROR}
\dot{\mathbf{E}}_{\mathrm{SYS}}=\mathbf{A}_{\mathrm{SYS}}{\mathbf{E}}_{\mathrm{SYS}}+{\color{black}\begin{bmatrix}
\mathbf{0}\\
\mathbf{I}
\end{bmatrix}
}
{\mathbf{V}}_{\mathrm{SYS}}
,
\end{equation}
where
\[
\label{AAAAAAAAAAAAAAAAA}
\mathbf{V}_{\mathrm{SYS}}=
\begin{bmatrix}
   \mathbf{V}_{\mathrm{SYS},x}^T& \mathbf{V}_{\mathrm{SYS},y}^T& \mathbf{V}_{\mathrm{SYS},z}^T
\end{bmatrix}
^T\in \mathbb{R}^{3N\times 1},\tag{AA}
\]
\[
\resizebox{0.99\hsize}{!}{%
$
\begin{split}
\mathbf{V}_{{\color{black}\mathrm{SYS}},q}=\sum_{j=0}^{\rho_q-\Xi}\mathbf{\Gamma}_j\dfrac{d^{\rho_q-j}}{dt^{\rho_q-j}}\left(\left[
\begin{array}{cc}
\mathbf{z}_{q,l,HT}\\
\mathbf{z}_{q,f,HT}\\
\end{array}
\right]
\right)
=
-\sum_{j=0}^{\rho_q-\Xi}\mathbf{\Gamma}_j\begin{bmatrix}
L_{\mathbf{f}_1}^{\rho_q-j}q_{1,HT}\\
\vdots\\
L_{\mathbf{f}_N}^{\rho_q-j}q_{N,HT}
\end{bmatrix}
\end{split}
$
}
,
\]
{\color{black}and $q\in \{x,y,z\}$.} Considering Remark \ref{RM3} and substituting $\mathbf{z}_{q,f,HT}$ by Eq. \eqref{followerglobaldesired}, $\mathbf{V}_{\mathrm{SYS},q}$ simplifies to
\begin{equation}
\label{VMUS}
\begin{split}
q\in \{x,y,z\},~
\mathbf{V}_{{\color{black}\mathrm{SYS}},q}=-\sum_{j=0}^{\rho_q-\Xi}\mathbf{\Gamma}_{j}
\begin{bmatrix}
    \mathbf{I}_{n+1}\\
    \mathbf{W}_L
\end{bmatrix}
\begin{bmatrix}
L_{\mathbf{f}_1}^{\rho_q-j}q_{1,HT}\\
\vdots\\
L_{\mathbf{f}_{n+1}}^{\rho_q-j}q_{n+1,HT}
\end{bmatrix}
\end{split}
,
\end{equation}
where $\mathbf{\Gamma}_{0}=\mathbf{I}_N$ and $\mathbf{I}_{n+1}$ are the identity matrices.}

\begin{theorem}\label{Theorem55}
Assume control gains, $\gamma_{1,i}$ through $\gamma_{\rho_q,i}$ are chosen such that eigenvalues of {\color{black}the} characteristic equation of external dynamics \eqref{External},
$
 \bigg|s\mathbf{I}-\mathbf{A}_{\mathrm{SYS},q}\bigg|=0
$ ($q\in \{x,y,z\}$),
are all located in the open left-half $s$-plane. If zero dynamics of the MVS internal dynamics \eqref{Internal} is \underline{locally} asymptotically stable {\color{black}and $\mathbf{F}_{\mathrm{I},\omega}=\frac{\partial \mathbf{F}_{\mathrm{INT}}}{\partial \omega}$ is Hurwitz}, then, {\color{black}the} error zero dynamics, obtained by substituting $\mathbf{V}_{\mathrm{SYS},q}=\mathbf{0}$ in Eq. \eqref{ERROR}, is a locally asymptotically stable {\color{black} at the origin}.
 \end{theorem}
\begin{proof}
The proof in inspired by theorem 6.3 in Ref. \cite{slotine1991applied}. By linearization, MVS collective zero dynamics is given by
\[
\begin{cases}
    \dot{\mathbf{E}}_{\mathrm{SYS}}=\mathbf{A}_{\mathrm{SYS}}{\mathbf{E}}_{\mathrm{SYS}}\\
    \frac{d\omega}{ dt}=\mathbf{F}_{\mathrm{I},E}\mathbf{E}_{{\color{black}\mathrm{SYS},y}}+\mathbf{F}_{\mathrm{I},z}\mathbf{E}_{{\color{black}\mathrm{SYS},z}}+\mathbf{F}_{\mathrm{I},\omega}\omega+\mathrm{H.O.T}.
\end{cases}
\]
$\mathbf{A}_{\mathrm{SYS}}$ is exponentially stable.
{\color{black}Because} $\mathbf{F}_{\mathrm{I},\omega}$ is {\color{black} a Hurwitz matrix}, MVS collective zero dynamics is locally asymptotically stable, {\color{black}which follows by} Lyapunov's first method. 
\end{proof}
\begin{corollary}
If $\Xi=\rho_q$ for $q\in\{x,y,z\}$, $\mathbf{V}_{\mathrm{SYS}}=\mathbf{0}$. Then, Theorem \ref{Theorem55} proves stability of the MVS collective dynamics in a continuum deformation coordination.
\end{corollary}
\begin{corollary}
If the internal dynamics \eqref{Internal} is linear, then, {\color{black}the} MVS collective dynamics can be represented by linear external and internal dynamics. Under this {\color{black}assumption}, Theorem \ref{Theorem55} can be used to ensure stability of the MVS continuum deformation and assign an upper bound $\delta$ for deviation of every vehicle from {\color{black}the} desired position given by {\color{black}a} {\color{black}homogeneous} deformation, e.g. $\|\mathbf{r}_i(t)-\mathbf{r}_{i,HT}(t)\|\leq \delta,~\forall i\in \mathcal{V},~\forall t\geq t_0$. 
\end{corollary}
}

\section{\hspace{0.2cm}Problem 2: Continuum Deformation Safety Specification}
\label{SimpleSafety}
\textit{A continuum deformation coordination is called \underline{valid} over time interval $[t_s,t_{f}]$ if the safety requirements are all satisfied at any time $t_s\leq t\leq t_f$. Required safety conditions are specified in Section \ref{Safety Conditions for MVS Continuum Deformation Coordination}. Also, sufficient conditions for inter-agent collision avoidance and vehicle containment are provided in Section \ref{Sufficient conditions for Collision Avoidance}.}

\subsection{Safety Conditions for MVS Continuum Deformation Coordination}\label{Safety Conditions for MVS Continuum Deformation Coordination}
This paper considers the following four safety requirements for a valid continuum deformation coordination:

\textbf{Safety Condition 1 (Bounded Deviation):} Transient error of every follower vehicle $i\in \mathcal{V}_F$ {\color{black}must be bounded} at any time $t\in [t_s,t_{f}]$. This condition can be expressed as follows:
\begin{equation}
\label{deltainequality}
    \forall t\in[t_s,t_f],~\forall i\in \mathcal{V},\qquad \|\mathbf{r}_{i}(t)-\mathbf{r}_{i,HT}(t)\|_2\leq \delta,
\end{equation}
where $\delta>0$ is constant.

\textbf{Safety Condition 2 (Inter-Agent Collision Avoidance):} Inter-agent collision {\color{black}must be} avoided in an MVS continuum deformation coordination {\color{black}by satisfying the condition}
\begin{equation}
\label{collisionavoidancecondition}
    \forall t\in[t_s,t_f],~\forall i,j\in \mathcal{V},~i\neq j,\qquad \|\mathbf{r}_{i}(t)-\mathbf{r}_{i}(t)\|_2\geq 2\epsilon,
\end{equation}
{\color{black}where $\epsilon>0$ is the radius of the ball enclosing each vehicle.}

\textbf{Safety Condition 3 (Vehicle Containment Condition):} {\color{black}For continuum deformation coordination in an obstacle-laden environment, the MVS {\color{black}needs to be contained} contained by} a virtual $n$-D simplex, called \emph{virtual containment simplex ({\color{black}VCS})}, at any time $t$. {\color{black}By ensuring MVS containment, obstacle collision avoidance is guaranteed if the {\color{black}VCS} boundary surfaces {\color{black}do not} hit obstacles in the motion space.} 

Let $\mathbf{h}_{i,0}=h_{x,i,0}{\color{black}\hat{\mathbf{e}}_1}+h_{y,i,0}{\color{black}\hat{\mathbf{e}}_2}+h_{z,i,0}{\color{black}\hat{\mathbf{e}}_3}$ and $\mathbf{h}_{i}\left(t\right)=h_{x,i}{\color{black}\hat{\mathbf{e}}_1}+h_{y,i}{\color{black}\hat{\mathbf{e}}_2}+h_{z,i}{\color{black}\hat{\mathbf{e}}_3}$($i=1,\cdots,n+1$) denote positions of vertex $i$ of the {\color{black}VCS} at reference time $t_0$ and current time $t$, respectively. {\color{black}VCS} evolution is defined by a {\color{black}homogeneous} transformation, therefore, $\mathbf{h}_{i,0}$ and $\mathbf{h}_{i}\left(t\right)$ are related by
\begin{equation}
    i=1,\cdots,n+1,\qquad \qquad \mathbf{h}_{i}\left(t\right)=\mathbf{Q}\left(t\right)\mathbf{h}_{i,0}+\mathbf{d}\left(t\right),
\end{equation}
where $\mathbf{Q}$ and $\mathbf{d}$ are {\color{black}computed based on} leaders' positions {\color{black}in the reference configuration and the current configuration at current time $t$} per Section \ref{MVS Homogeneous Deformation Coordination}. {\color{black}The MVS containment condition is mathematically specified by
\begin{equation}
    t\in [t_s,t_f],~\forall i\in \mathcal{V},\qquad \mathbf{\Theta}_n\left(\mathbf{h}_1\left(t\right),\cdots,\mathbf{h}_{n+1}(t),\mathbf{r}_i(t)\right)> \mathbf{0},
\end{equation}
{\color{black}where $\mathbf{\Theta}_n$ is an appropriately defined function
defined by Eq.  \eqref{Thetaaan}.}
} 

\textbf{Safety Condition 4 (Feasibility of Vehicular Inputs):} A desired continuum deformation is planned by leaders  and must be followed by every vehicle in the team. In other words, the generalized coordinate $\mathbf{s}_\varrho^n(t)$ ($\varrho\in\{\mathrm{OF},\mathrm{OL}\},~n=1,2,3$) must be planned such that every vehicle can follow the desired coordination by employing {\color{black}control inputs which satisfy control constraints}
\begin{equation}
t\in [t_s,t_f],~\forall i\in \mathcal{V},\qquad     \mathbf{u}_i(t)\in \mathcal{U}_i.
\end{equation}

\subsection{Sufficient conditions for Collision Avoidance}\label{Sufficient conditions for Collision Avoidance} It is computationally expensive to ensure inter-agent collision avoidance by satisfying the safety condition \eqref{deltainequality} for every vehicle pair $i$ and $j$. Because MVS desired coordination is defined by an affine transformation, inter-agent collision avoidance and follower containment conditions can be {\color{black}enforced} by assigning lower limits on eigenvalues of {\color{black}the} matrix $\mathbf{U}_D=\left(\mathbf{Q}^T\mathbf{Q}\right)^\frac{1}{ 2}$. Section \ref{Sufficient conditions for Collision Avoidance}-1 provides a conservative condition for {\color{black}follower} containment and inter-agent collision avoidance in an obstacle-laden environment. Furthermore, an aggressive condition for inter-agent collision avoidance in obstacle-free environments is provided in Section \ref{Sufficient conditions for Collision Avoidance}-2.

\subsubsection{Conservative Collision Avoidance Condition}\label{Conservative Collision Avoidance Condition}
{\color{black}We can ensure collision avoidance and vehicle containment by assigning a {single} lower-limit for \underline{all} eigenvalues of matrix $\mathbf{U}_D(t)$, denoted by $\lambda_1(t)$, $\lambda_2(t)$, and $\lambda_3(t)$, at any time $t$.}


\begin{proposition}\label{prop5}
Assume agents' desired positions are defined by the homogeneous transformation given in Eq. \eqref{homogtrans}. Suppose that inequality \eqref{deltainequality} is satisfied at any time $t$. Inter-agent collision avoidance, defined by Eq. \eqref{collisionavoidancecondition}, is guaranteed, if
\begin{equation}
\label{originalcollisionavoidance}
   t\in[t_s,t_f],~\forall i,j\in \mathcal{V},~i\neq j,~~ \|\mathbf{r}_{i,HT}(t)-\mathbf{r}_{j,HT}(t)\|_2\geq 2\left(\delta+\epsilon\right).
\end{equation}
\end{proposition}
\begin{proof}
{\color{black}Note that}
\[
\left(\mathbf{r}_{i}-\mathbf{r}_{j}\right)=\left(\mathbf{r}_{i,HT}-\mathbf{r}_{j,HT}\right)-\left(\mathbf{r}_{i,HT}-\mathbf{r}_i\right)-\left(\mathbf{r}_j-\mathbf{r}_{j,HT}\right).
\]
{\color{black}Hence,}
\[
\begin{split}
    \|\mathbf{r}_{i}-\mathbf{r}_{j}\|_2
     \geq&\|\mathbf{r}_{i,HT}-\mathbf{r}_{j,HT}\|_2-\|\mathbf{r}_{i,HT}-\mathbf{r}_i\|_2-\|\mathbf{r}_j-\mathbf{r}_{j,HT}\|_2\\
     \geq & 2\left(\delta+\epsilon\right)-\delta-\delta=2\epsilon.
\end{split}
\]
\end{proof}
Proposition \ref{prop5} implies {\color{black}that, assuming \eqref{deltainequality} holds,} Eq. \eqref{originalcollisionavoidance} can be used instead of Eq. \eqref{collisionavoidancecondition}
to ensure inter-agent collision avoidance. Substituting $\mathbf{r}_{i,HT}=\mathbf{Q}\mathbf{r}_{i,0}+\mathbf{d}$ and $\mathbf{r}_{j,HT}=\mathbf{Q}\mathbf{r}_{j,0}+\mathbf{d}$, \eqref{originalcollisionavoidance} holds if and only if
\begin{equation}
\label{eQ22}
\begin{split}
    \forall i,j\in \mathcal{V},~i\neq j\qquad 4\left(\delta+\epsilon\right)^2\leq
    &
    \left(\mathbf{r}_{i,0}-\mathbf{r}_{j,0}\right)^T\mathbf{U}_D^2(t)\left(\mathbf{r}_{i,0}-\mathbf{r}_{j,0}\right).
\end{split}
\end{equation}
{\color{black}for all $t$ in $[t_s,t_f]$} 
where $\mathbf{U}_D^2(t)=\mathbf{Q}^T(t)\mathbf{Q}(t)$.
\begin{theorem}\label{conservativecollisionavoidancetheorem}
Define
\begin{equation}
\label{dlmin}
d_{s}=\min\limits_{i,j\in\mathcal{V}_R,i\neq j}\|\mathbf{r}_{i,0}-\mathbf{r}_{j,0} \|
\end{equation}
as the minimum separation distance between any two agents in the reference configuration. Inter-agent collision avoidance is guaranteed if
\begin{equation}
\label{colavoidcond}
    \lambda_i\geq  \lambda_{\mathrm{min}},\qquad i=1,2,3,
\end{equation}
where $\lambda_1$, $\lambda_2$, and $\lambda_3$ are the eigenvalues of {\color{black}$\mathbf{U}_D(t)$} and
\begin{equation}
\label{lambdaminnn}
   \lambda_{\mathrm{min}}= 2 \dfrac{\delta+\epsilon}{d_s}.
\end{equation}
\end{theorem}
\begin{proof}
The right-hand side of Eq. \eqref{eQ22} satisfies the following inequality:
\begin{equation}
    \forall i,j\in \mathcal{V},~i\neq j,\qquad \lambda_{\mathrm{min}}^2 d_s^2\leq \left(\mathbf{r}_{i,0}-\mathbf{r}_{j,0}\right)^T\mathbf{U}_D^2(t)\left(\mathbf{r}_{i,0}-\mathbf{r}_{j,0}\right).
\end{equation}
Inter-agent collision avoidance is guaranteed {\color{black}when}
\[
 4\left(\delta+\epsilon\right)^2\leq \lambda_i^2d_s^2,\qquad  i=1,2,3.
\]
{\color{black}Thus, collision} avoidance is ensured if Eq. \eqref{colavoidcond} is satisfied.
\end{proof}

\begin{theorem}\label{thm4}
Assume each vehicle is enclosed by a ball with radius $\epsilon$. Given deviation upper-bound $\delta$ (Eq. \eqref{deltainequality}), we define
\begin{equation}
\label{deltaaamax}
\delta_{\mathrm{max}}=\min\big\{\left(d_b-\epsilon\right),\frac{1}{2}\left(d_s-2\epsilon\right)\big\},
\end{equation} 
where {$d_b$ is the minimum distance from the boundary of the leading simplex at in the reference configuration.} {\color{black}The vehicle} containment and inter-agent collision avoidance are guaranteed {\color{black}at time instant $t$} if eigenvalues $\lambda_1,~\lambda_2,~\lambda_3$ {\color{black}of $\mathbf{U}_D(t)$} satisfy the following inequality constraint \cite{rastgoftar2018asymptotic}:
\begin{equation}
\label{NOinteragent}
j=1,2,3,\qquad    \lambda_{j}\geq \lambda_{\mathrm{CD,min}}
\end{equation}
where
\begin{equation}
\label{lcmin}
    \lambda_{\mathrm{CD,min}}=\dfrac{\delta+\epsilon}{\delta_{\mathrm{max}}+\epsilon}.
\end{equation}
\end{theorem} 
{\color{black}
\begin{proof}
See the proof of Theorem \ref{thm4} in Ref. \cite{rastgoftar2018safe}.
\end{proof}

\begin{remark}\label{RMK5}
Theorems \ref{thm4} provides a {\color{black}conservative} collision avoidance guarantee condition independent of the total number of agents ($N$). Therefore, inter-agent collision avoidance can be guaranteed for a large-scale MVS  by checking {\color{black}this} single safety condition. However, {\color{black}the conditions \eqref{NOinteragent}-\eqref{lcmin}} can be overly restrictive when agents are not uniformly distributed in the reference configuration. For example, consider the non-uniform initial configuration shown in Fig. \ref{consvAGGS}  where $\delta=\delta_{\mathrm{max}}$. Here, $\lambda_{\mathrm{CD,min}}=1$ and the MVS continuum deformation coordination can be either rigid or expansive. However, contraction along the $\hat{\mathbf{e}}_{2}$ axis is safe but it is not permitted because eigenvalues $\lambda_1(t)$, $\lambda_2(t)$, and $\lambda_3(t)$ {\color{black}of $\mathbf{U}_D(t)$} must be greater than or equal to $1$ at any time $t$. 
\end{remark}
\begin{figure}
 \centering
 \subfigure[]{\includegraphics[width=0.45\linewidth]{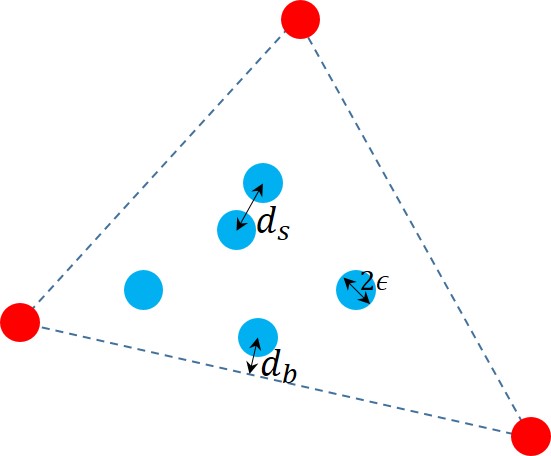}}
  \subfigure[]{\includegraphics[width=0.45\linewidth]{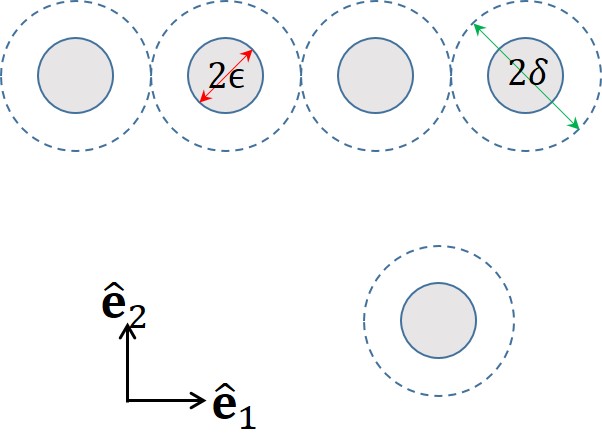}}
     \caption{(a) Schematic for MVS {\color{black}reference} configuration and the VCS containing all vehicles. (b) Schematic of a non-uniform MVS reference configuration where $\delta=\delta_{\mathrm{max}}$ and $\lambda_{\mathrm{CD,min}}=1$.}
\label{consvAGGS}
\end{figure}

\subsubsection{Relaxation of the Collision Avoidance Condition}\label{Aggressive Collision Avoidance Condition}
{\color{black}The issue in Remark \ref{RMK5} can be dealt with, if we only constrain {\color{black}the eigenvalue $\lambda_1$} of matrix $\mathbf{U}_D{\color{black}(t)}$ and direct eigenvector $\hat{\mathbf{u}}_1$ along the vector connecting a pair of agents with the minimum separation distance.} Given MVS reference positions, reference eigenvector $\hat{\mathbf{u}}_{1,0}$ is defined as follows:
\begin{equation}
    \hat{\mathbf{u}}_{1,0}=\min\limits_{i,j\in\mathcal{V}_R,i\neq j}\left\{\dfrac{\mathbf{r}_{i,0}-\mathbf{r}_{j,0}}{\|\mathbf{r}_{i,0}-\mathbf{r}_{j,0}\|}\right\}.
\end{equation}

 Because we use the $3-2-1$ {\color{black}Euler angles} to define a rigid-body rotation, {\color{black}the} eigenvector $\hat{\mathbf{u}}_1$ {\color{black}of ${\mathbf{U}}_D(t)$} only depends on $\theta_u$ and $\psi_u$ at any time $t$ (see Eq. \eqref{ui}). Thus, the reference deformation angles $\theta_{u,0}$ and $\psi_{u,0}$ are obtained by
 \begin{subequations}
 \label{thspssassignment}
 \begin{equation}
     \theta_{u,0}=-\sin^{-1}\left(\hat{\mathbf{u}}_{1,0}\cdot\hat{\mathbf{e}}_3\right),
 \end{equation}
 \begin{equation}
     \psi_{u,0}=\tan^{-1}\left(\dfrac{\hat{\mathbf{u}}_{1,0}\cdot\hat{\mathbf{e}}_2}{\hat{\mathbf{u}}_{1,0}\cdot\hat{\mathbf{e}}_1}\right).
 \end{equation}
 \end{subequations}
{\color{black}Suppose now that the} desired continuum deformation is designed such that deformation angles $\phi_u(t)=\phi_{u,0}$,  $\theta_u(t)=\theta_{u,0}$ and $\psi_u(t)=\psi_{u,0}$ remain constant {\color{black}for} $t\in [t_s,t_f]$, {\color{black}where $\phi_{u,0}$ takes {\color{black}an} arbitrary value between $0$ and $2\pi$.} 
\begin{lemma}\label{lemmmma1}
If  $\phi_u(t)=\phi_{u,0}^*$ $\theta_u(t)=\theta_{u,0}^*$ and $\psi_u(t)=\psi_{u,0}^*$ are constant {\color{black}for} $t\in [t_s,t_f]$ and $\mathbf{r}_{i.HT}$ and $\mathbf{r}_{j.HT}$ are defined by the homogeneous transformation given in \eqref{homogtrans} ($\forall i,j\in \mathcal{V}_R$), then, the following relation holds:
\begin{equation}
\label{lemmaaaaproven}
\resizebox{0.99\hsize}{!}{%
$
    \forall i,j\in \mathcal{V}_R,\qquad \left(\mathbf{r}_{i,HT}(t)-\mathbf{r}_{j,HT}(t)\right)^T \left(\mathbf{R}_D(t)\hat{\mathbf{u}}_{l,0}\right)=\lambda_l(t)\left(\mathbf{r}_{i.0}-\mathbf{r}_{j,0}\right)^T{\hat{\mathbf{u}}_{l,0}}
$
}
\end{equation}
\end{lemma}
\begin{proof}
Given global desired positions of vehicles $i,j\in \mathcal{V}_R$ defined by Eq. \eqref{homogtrans}, we can write
\begin{equation}\label{prflemme}
    \resizebox{0.99\hsize}{!}{%
$
\begin{split}
    \left(\mathbf{r}_{i,HT}(t)-\mathbf{r}_{j,HT}(t)\right)=&\mathbf{Q}(t)\left(\mathbf{r}_{i,0}(t)-\mathbf{r}_{j,0}\right)\\
    =&\mathbf{R}_D(t)\underbrace{\sum_{h=1}^3\lambda_h(t)\hat{\mathbf{u}}_{h,0}{\hat{\mathbf{u}}}_{h,0}^T}_{\mathbf{U}_D{\color{black}(t)}}\left(\mathbf{r}_{i,0}-\mathbf{r}_{j,0}\right)\\
    =&\mathbf{R}_D(t)\sum_{h=1}^3\lambda_h(t)\hat{\mathbf{u}}_{h,0}\left(\mathbf{r}_{i,0}-\mathbf{r}_{j,0}\right)^T{\hat{\mathbf{u}}_{h,0}}\\
    =&\sum_{h=1}^3\lambda_h(t)\left[\left(\mathbf{r}_{i,0}-\mathbf{r}_{j,0}\right)^T\hat{\mathbf{u}}_{h,0}\right]
    \mathbf{R}_D(t)\hat{\mathbf{u}}_{h,0}
\end{split}
$
}
\end{equation}
Note that $\lambda_h(t)\left[\left(\mathbf{r}_{i,0}-\mathbf{r}_{j,0}\right)^T\hat{\mathbf{u}}_{h,0}\right]\in \mathbb{R}$ at any time $t$ ($h=1,2,3$). Thus, pre-multiplying both sides of Eq. \eqref{prflemme} by ${\hat{\mathbf{u}}_{l,0}}^T\mathbf{R}_D^T(t)$ and substituting 
\[
{\hat{\mathbf{u}}_{l,0}}^T\mathbf{R}_D^T(t)\mathbf{R}_D(t)\hat{\mathbf{u}}_{h,0}=\begin{cases}
1&h=l\\
0&h\neq l
\end{cases}
,
\]
Eq. \eqref{lemmaaaaproven} {\color{black}follows}.

\end{proof}
\begin{theorem}
Assume every vehicle is enclosed by a ball of radius $\epsilon$. MVS inter-agent collision avoidance is guaranteed if
\begin{equation}
\label{EQ64}
    \forall t\in [t_s,t_f],\qquad \lambda_1(t)\geq\lambda_{\mathrm{min}},
\end{equation}
where {\color{black}$\lambda_1(t)$ is the first eigenvalues of $\mathbf{U}_D(t)$,} $d_s$ is the minimum separation distance defined in \eqref{dlmin} and
\begin{equation}
\label{relaxedsafetyy}
\lambda_{\mathrm{min}}=2\dfrac{\delta+\epsilon}{d_{s}}.
\end{equation}
\end{theorem}
\begin{proof}
Vectors $\mathbf{R}_D(t)\hat{\mathbf{u}}_{1,0}$, $\mathbf{R}_D(t)\hat{\mathbf{u}}_{2,0}$, $\mathbf{R}_D(t)\hat{\mathbf{u}}_{3,0}$ are {\color{black}parallel to} $\hat{\mathbf{u}}_1(t)$, $\hat{\mathbf{u}}_2(t)$, $\hat{\mathbf{u}}_3(t)$, respectively. Because the lowest minimum separation distance is along the unit vector $\hat{\mathbf{u}}_1(t)=\mathbf{R}_D(t)\hat{\mathbf{u}}_{1,0}$, inter-agent collision is avoided if 
\begin{equation}
\label{COLAvoidtheoremmm}
    \forall i,j\in \mathcal{V}_R,\qquad \left(\mathbf{r}_{i,HT}(t)-\mathbf{r}_{j,HT}(t)\right)^T \left(\mathbf{R}_D(t)\mathbf{u}_{1,0}\right)\geq 
2\left(\delta+\epsilon\right).
\end{equation}
By Lemma \ref{lemmmma1}, {\color{black}it follows that}
\[
\resizebox{0.99\hsize}{!}{%
$
\begin{split}
    \min\limits_{i,j\in \mathcal{V}_R}\left\{\left(\mathbf{r}_{i,HT}(t)-\mathbf{r}_{j,HT}(t)\right)^T \left(\mathbf{R}_D(t)\mathbf{u}_{1,0}\right)\right\}=&\lambda_1(t)\min\limits_{i,j\in \mathcal{V}_R}\left\{\left(\mathbf{r}_{i,0}-\mathbf{r}_{j,0}\right)^T\hat{\mathbf{u}}_{1,0}\right\}\\
   =&\lambda_1(t)d_s,
\end{split}
$
}
\]
{\color{black}Thus the} inter-agent collision is avoided if 
\[
t\in [r_s,t_f],\qquad \lambda_1(t)d_s\geq 2\left(\delta+\epsilon\right).
\]
{\color{black}This is ensured if \eqref{EQ64} and} \eqref{relaxedsafetyy} are satisfied.
\end{proof}
{\color{black}
\begin{remark}\label{importantremark}
Eigenvalues $\lambda_2(t)$ and $\lambda_3(t)$ of matrix $\mathbf{U}_D$ must be positive at any time $t$ as is required for continum deformation coordination. Furthermore, shear deformation angle $\phi_{u,0}^*$ can be arbitrarily selected. Without loss of generality, this paper chooses $\phi_{u}(t)=\phi_{u,0}^*=0$ at any time $t$.
\end{remark}

}

}

\section{\hspace{0.4cm}Problem 3: Continuum Deformation Planning}
\label{Continuum Deformation Optimization}
The paper offers two {\color{black}distinct} strategies for continuum deformation planning in obstacle-free and obstacle-laden environments. In an obstacle-free environment, safe continuum deformation coordination is planned by shaping homogeneous transformation features. In an obstacle-laden environment, A* search is used to plan a safe continuum deformation by choosing optimal leaders' paths given MVS initial and target formations.

\subsection{Obstacle-Free Environment}
 {\color{black}For continuum deformation coordination in an obstacle-free environment, generalized coordinate vector $\mathbf{s}_{\mathrm{OF}}^n$ is defined as follows:}
\begin{subequations}
\begin{equation}\setlength\arraycolsep{3pt}
    {\color{black}\mathbf{s}_{\mathrm{OF}}^1}=
    \begin{bmatrix}
    \lambda_1&\theta_r&\psi_r&d_1&d_2&d_3
    \end{bmatrix}
    ^T,
\end{equation}
\begin{equation}\setlength\arraycolsep{3pt}
\label{HomTransFeat2D}
    {\color{black}\mathbf{s}_{\mathrm{OF}}^2}=
    \begin{bmatrix}
    \lambda_1&\lambda_2&
    \phi_r&\theta_r&\psi_r&d_1&d_2&d_3
    \end{bmatrix}
    ^T,
\end{equation}
\begin{equation}\setlength\arraycolsep{2pt}
{\color{black}\mathbf{s}_{\mathrm{OF}}^3}=
\begin{bmatrix}
    \lambda_1&\lambda_2&\lambda_3
    &\phi_r&\theta_r&\psi_r&d_1&d_2&d_3
    \end{bmatrix}
    ^T,
\end{equation}
\end{subequations}
{\color{black}
where $\mathbf{s}_{\mathrm{OF}}^n$ is defined by quintic polynomial \eqref{barsss} with $t_1=t_s$ and $t_2=t_f$. The global desired position of vehicle $i$ is given by \eqref{homogtrans}, where $\mathbf{d}=\left[d_1(t)~d_2(t)~d_3(t)\right]^T$ for $n=1,2,3$, and
\begin{equation}
    \mathbf{Q}=
    \begin{cases}
        \mathbf{Q}\left(\lambda_1,\theta_r,\psi_r\right)&n=1\\
        \mathbf{Q}\left(\lambda_1,\lambda_2,\phi_r,\theta_r,\psi_r\right)&n=2\\
        \mathbf{Q}\left(\lambda_1,\lambda_2,\lambda_3,\phi_r,\theta_r,\psi_r\right)&n=3\\
    \end{cases}
    .
\end{equation}
\begin{remark}
For coordination in an obstacle-free environment, shear deformation angles $\phi_u$, $\theta_u$, and $\psi_u$ are either $0$ or constant at any time $t\in [t_s,t_f]$ (See Section \ref{Sufficient conditions for Collision Avoidance}-2). Therefore, 
shear deformation angles are not included as generalized coordinates.
\end{remark}
}

\subsection{Obstacle-Laden Environment}
\label{Obstacle-Laden Environment}
{\color{black}
For coordination in an obstacle-laden environment, 
\begin{equation}
    \mathbf{s}_{\mathrm{OL}}^n(t)=\mathrm{vec}\left(
    \begin{bmatrix}
    \mathbf{h}_1(t)&\cdots\mathbf{h}_{n+1}(t)
    \end{bmatrix}
    ^T
    \right),
\end{equation}
is the generalized coordinate vector which is defined based on the $n(n+1)$ leader position components. Note that ``$\mathrm{vec}$'' is the matrix vectorization symbol. Also,  
\[
\mathbf{r}_{i,HT}(t)=\left(\mathbf{I}_3\otimes
\mathbf{\Theta}_n\left(\mathbf{h}_{1,0},\cdots,\mathbf{h}_{n+1,0},\mathbf{r}_{i,0}\right)\right)\mathbf{s}_{\mathrm{OL}}^n,
\]
where $n\in \{1,2,3\}$ is the dimension of the homogeneous deformation, $\mathbf{r}_{i,0}$ is the reference position of vehicle $i\in \mathcal{V}_R$, $\mathbf{h}_{1,0}$ through $\mathbf{h}_{n+1,0}$ denote reference positions of VCS vertices.
}

This paper applies A* search to optimally {\color{black}determine intermediate generalized coordinate vectors $\bar{\mathbf{s}}_{2,\mathrm{OL}}^n$, $\cdots$, $\bar{\mathbf{s}}_{n_\tau,\mathrm{OL}}^n$ given initial and final generalized coordinate vectors $\bar{\mathbf{s}}_{1,\mathrm{OL}}^n$ and $\bar{\mathbf{s}}_{n_\tau+1,\mathrm{OL}}^n$, respectively. }  
We define the following:
\[
\begin{split}
   \Omega_s^n=\left(\mathbf{h}_{1,s},\cdots, \mathbf{h}_{n+1,s}\right):=&\mathrm{Initial~{\color{black}VCS}}\\
   \Omega_f^n=\left(\mathbf{h}_{1,f},\cdots, \mathbf{h}_{n+1,f}\right):=&\mathrm{Final~{\color{black}VCS}}\\
   \Omega_c^n=\left(\mathbf{h}_{1,c},\cdots, \mathbf{h}_{n+1,c}\right):=&\mathrm{Current~{\color{black}VCS}}\\
   \Omega_{\mathrm{next}}^n\left(\mathbf{h}_{1,\mathrm{next}},\cdots, \mathbf{h}_{n+1,\mathrm{next}}\right):=&\mathrm{Next~{\color{black}VCS}}\\
\end{split}
\]
where $\mathbf{h}_{i,0}$, $\mathbf{h}_{i,f}$, $\mathbf{h}_{i,c}$, and $\mathbf{h}_{i,n}$ denote position of vertex $i$ of the {\color{black}VCS} at initial time $t_s$, final time $t_f$, current time, and next time, respectively.
Leaders' optimal paths are determined by minimizing continuum deformation cost, given by
\begin{equation}
n=1,2,3,\qquad    F\left(\Omega_{\mathrm{next}}^n\right)=G\left(\Omega_{\mathrm{next}}^n\right)+H\left(\Omega_{\mathrm{next}}^n\right),
\end{equation}
where $\Omega_{\mathrm{next}}^n$ is a \underline{valid} continuum deformation from $\Omega_{\mathrm{c}}^n$ and 
\begin{equation}
 n=1,2,3,\qquad    H\left(\Omega_{\mathrm{next}}^n\right)=\sqrt{\sum_{l=1}^{n+1}\bigg\|\mathbf{h}_{l,\mathrm{next}}-\mathbf{h}_{l,f}\bigg\|_2^2},
\end{equation}
is the heuristic cost.
Furthermore, 
\begin{equation}
n=1,2,3,\qquad     G\left(\Omega_{\mathrm{next}}^n\right)=\min\big\{G\left(\Omega_{c}^n\right)+C_{c,{\mathrm{next}}}^n\big\},
\end{equation}
is the minimum estimated cost from $\Omega_{s}^n$ to $\Omega_{\mathrm{next}}^n${\color{black},} where $C_{c,{\mathrm{next}}}^n=\sum_{l=1}^{n+1}\bigg\|\mathbf{h}_{l,\mathrm{next}}-\mathbf{h}_{l,c}\bigg\|_2$ ($n=1,2,3$).

\subsection{Planning of travel time}
{\color{black}{\color{black}From \eqref{ERROR}}
\[
\label{BB}
\mathbf{E}_{\mathrm{SYS}}(t)=
    \mathrm{e}^{\mathbf{A}_{\mathrm{SYS}}{\color{black}\left(t-t_s\right)}}\mathbf{E}_{\mathrm{SYS}}\left(t_s\right)+\int_{t_s}^t\mathrm{e}^{\mathbf{A}_{\mathrm{SYS}}\left(t-\tau\right)}{\color{black}\mathbf{V}_{\mathrm{SYS}}\left(\tau\right)}\mathrm{d}\tau.\tag{BB}
\]
{\color{black}Define} $\mathbf{C}_{i}=[{\mathbf{C}_i}_{_{lh}}]\in {\color{black}\mathbb{R}^{1\times \left(\rho_x+\rho_y+\rho_z\right) N}}$, with elements
\[
\resizebox{0.99\hsize}{!}{%
$
{\mathbf{C}_i}_{_{lh}}=
\begin{cases}
    1&\left(l=1\wedge h=i\right)\vee \left(l=2\wedge h=i+\rho_xN\right)\vee \left(l=3\wedge h=i+\left(\rho_x+\rho_y\right)N\right)\\
   0&\mathrm{else}
\end{cases}
.$
}
\]
Assume consecutive generalized coordinate vectors  $\bar{\mathbf{s}}_{k,\varrho}^n$ and $\bar{\mathbf{s}}_{k+1,\varrho}^n$ are known for $\varrho\in \{\mathrm{OF},\mathrm{OL}\}$ and $k=1,\cdots,n_\tau$). Then, the travel time planning problem is defined as follows. Choose $T_k\geq T_k^*$, where minimum travel time $T_k^*$ is assigned by solving the following optimization problem:
\begin{equation}
    T_k^*=\min ~T_k
\end{equation}
{\color{black}subject to \eqref{BB} and}
\begin{equation}
\label{safetyfirst}
\begin{split}
     \mathbf{s}_{\varrho}^n(t)&=\bar{\mathbf{s}}_{k,\varrho}^n(1-\beta(t,T_k))+\beta(t,T_k)\bar{\mathbf{s}}_{k+1,\varrho}^n,\qquad t\in [t_k,t_{k+1}],\\
     \mathbf{E}_{\mathrm{SYS}}^T(t)&\mathbf{C}_i^T\mathbf{C}_i\mathbf{E}_{\mathrm{SYS}}(t)\leq \delta^2,\qquad \forall t\in [t_k,t_{k+1}],\forall i\in \mathcal{V},\\
     \mathbf{u}_i(t)&\in \mathcal{U},\qquad \forall t\in [t_k,t_{k+1}],\forall i\in \mathcal{V},\\
     \beta(t,T_k)&=\sum_{j=0}^5\zeta_{j,k}\left({t-t_k\over T_k}\right)^5,
     \end{split}
\end{equation}
where {\color{black}$t_s=t_1$, $t_f=t_{n_\tau+1}$, $t_{k+1}=t_k+T_k$, and $\zeta_{0,k}$ through $\zeta_{5,k}$ are specified coefficients satisfying assumptions discussed after \eqref{poly1}. Note that typically $\mathbf{E}_{\mathrm{SYS}}(t_s)=\mathbf{0}$.
Because $\beta (t,T_k)$ is assigned by Eq. \eqref{poly1} for $t\in [t_k,t_{k+1}]$, the following properties hold:
\begin{enumerate}
    \item{$\left|\dfrac{d^j\beta (t,T_k)}{dt^j}\right|$
is a decreasing function with respect to $T_k$ at any time $t\in [t_k,t_{k+1}]$ for $j\in\{\rho_q-\Xi,\cdots\rho_q\}$ ($q\in\{x,y,z\}$), and}
\item{$\mathbf{V}_{\mathrm{SYS}}(t)$, defined by \eqref{AAAAAAAAAAAAAAAAA}, and $\mathbf{E}_{\mathrm{SYS}}(t)$, defined by \eqref{BB}, tend to $\mathbf{0}\in \mathbb{R}^{3N\times 1}$ and $\mathbf{0}\in \mathbb{R}^{\left(\rho_x+\rho_y+\rho_z\right)N\times 1}$, respectively, at any time $t\in [t_k,t_{k+1}]$ as $T_k\rightarrow \infty$, } 
\end{enumerate}
Consequently, there exists a minimum travel time $T_k^*$ such that the safety inequality constraints given in \eqref{safetyfirst} {\color{black}are} all satisfied.}

}

\section{\hspace{0.3cm}Simulation Results}
\label{simulation}
Case studies of $2$-D, and $3$-D continuum deformation are presented below with and without obstacles. 
We assume each vehicle is a quadcopter. 
The variables $\phi_i$, $\theta_i$, and $\psi_i$ {\color{black}denote the $i$th quadcopter} Euler angles, {\color{black}the variables} $x_i$, $y_i$, $z_i$ denote its position coordinates, $v_{x,i}$, $v_{y,i}$, $v_{z,i}$ denote its velocity components, $m_i$ {\color{black}is} mass, $F_{T,i}$ {\color{black}is the total thrust magnitude}, $\bar{F}_{T,i}={F_{T,i}\over m_i}$ {\color{black}is the} thrust force per unit mass, and $g=981m/s^2$ {\color{black}is} the gravity acceleration. Quadcopter $i$ dynamics is then given by Eq. \eqref{Dynamicsofagenti}
where
\[
 \mathbf{X}_{i}=[x_i~y_i~z_i~v_{x,i}~v_{y,i}~v_{z,i}~\bar{F}_{T,i}~\phi_i~\theta_i~\psi_i~\dot{\bar{F}}_{T,i}~\dot{\phi}_i~\dot{\theta}_i~\psi_i]^T 
\]
is the state, $\mathbf{r}_i$ is the control output, {\color{black}and}
$\mathbf{u}_{i}=[u_{T,i}~u_{\phi,i}~u_{\theta,i}]$ is the control input vector. Then,
\[
\begin{split}
     \mathbf{f}_i=&[v_{x,i}~v_{y,i}~v_{z,i}~f_{4,i}~f_{5,i}~f_{6,i}~\dot{\bar{F}}_{T,i}~\dot{\phi}_i~\dot{\theta}_i~\dot{\psi}_i~0~0~0~f_{14,i}]^T,\\
     \begin{bmatrix}
         f_{4,i}\\
         f_{5,i}\\
         f_{6,i}\\
     \end{bmatrix}=&
     \begin{bmatrix}
        0\\
        0\\
        -g
        \end{bmatrix}
        +
        \bar{F}_{T,i}
    \begin{bmatrix}
        C_{\phi_{i}}S_{\theta_{i}}C_{\psi_{i}}+S_{\phi_{i}}S_{\psi_{i}}\\
        C_{\phi_{i}}S_{\theta_{i}}S_{\psi_{i}}-S_{\phi_{i}}C_{\psi_{i}}\\
        C_{\phi_{i}}C_{\theta_{i}}\\
    \end{bmatrix},
    \\
    \mathbf{g}_i=&
    \begin{bmatrix}
        \mathbf{0}_{3\times 9}&
        \mathbf{I}_3&
        \mathbf{0}_{3\times 1}
    \end{bmatrix}
    ^T.
\end{split}
\]
Using feedback linearization, quadcopter dynamics \eqref{Dynamicsofagenti} can be expressed as
\begin{subequations}
\begin{equation}
\label{mainequation}
\dfrac{\mathrm{d}^4\mathbf{r}_i}{\mathrm{d}t^4}=\mathbf{v}_{i}
\end{equation}
\begin{equation}
\label{INTDYNAMICSS}
    \dfrac{\mathrm{d}^2\psi_i}{\mathrm{d}t^2}=\ddot{\psi}_{i,HT}+k_{\dot{\psi}_i}\left(\dot{\psi}_{d,i}-\dot{\psi}_{i}\right)+k_{{\psi}_i}\left({\psi}_{d,i}-{\psi}_{i}\right),
\end{equation}
\end{subequations}
where $\rho_x=\rho_y=\rho_z=4$ and \eqref{INTDYNAMICSS} represents the internal dynamics of the quadcopter stytem. {\color{black}Parameters} $k_{\psi_i}$ and $k_{\dot{\psi}_i}$ are positive constants and we choose $\psi_{d,i}=0$, $\dot{\psi}_{d,i}=0$, and $\ddot{\psi}_{d,i}=0$. Therefore, the quadcopter zero-dynamics is stable.
For the quadcopter model \cite{rastgoftar2018safe}, 
\begin{equation}
   \mathbf{v}_{i}=\mathbf{M}_{i}
     \mathbf{u}_{i}
    +\mathbf{n}_{i}
    ,
\end{equation}
where
\[
\begin{split}
    \mathbf{M}_{i}=&
    \begin{bmatrix}
        \Lambda_{1,i}&\Lambda_{2,i}&\Lambda_{3,i}
    \end{bmatrix}
    ,\\
     \mathbf{n}_{i}=&\dot{\Lambda}_{0,i}+\dot{\Lambda}_{1,i}\dot{\bar{F}}_{T,i}+\dot{\Lambda}_{2,i}\dot{\phi}_{i}+\dot{\Lambda}_{3,i}\dot{\theta}_{T,i},
    \end{split}
\]
\[
\resizebox{0.99\hsize}{!}{%
$
\begin{split}
       \Lambda_{0,i}=&\dot{\psi}_{i}
    \begin{bmatrix}
        -C_{\phi_{i}}S_{\theta_{i}}S_{\psi_{i}}+S_{\phi_{i}}C_{\psi_{i}}\\
        C_{\phi_{i}}S_{\theta_{i}}C_{\psi_{i}}+S_{\phi_{i}}S_{\psi_{i}}\\
        0\\
    \end{bmatrix}
    ,~
    \Lambda_{1,i}=
    \begin{bmatrix}
        C_{\phi_{i}}S_{\theta_{i}}C_{\psi_{i}}+S_{\phi_{i}}S_{\psi_{i}}\\
        C_{\phi_{i}}S_{\theta_{i}}S_{\psi_{i}}-S_{\phi_{i}}C_{\psi_{i}}\\
        C_{\phi_{i}}C_{\theta_{i}}\\
    \end{bmatrix}
    ,\\
    \Lambda_{2,i}=&
    \begin{bmatrix}
        -S_{\phi_{i}}S_{\theta_{i}}C_{\psi_{i}}+C_{\phi_{i}}S_{\psi_{i}}\\
        -S_{\phi_{i}}S_{\theta_{i}}S_{\psi_{i}}-C_{\phi_{i}}C_{\psi_{i}}\\
        -S_{\phi_{i}}C_{\theta_{i}}\\
    \end{bmatrix}
    ,~
    \Lambda_{3,i}=
    \begin{bmatrix}
        C_{\phi_{i}}C_{\theta_{i}}C_{\psi_{i}}\\
        C_{\phi_{i}}C_{\theta_{i}}S_{\psi_{i}}\\
        -S_{\theta_{i}}C_{\phi_{i}}\\
    \end{bmatrix}
    .
\end{split}
$
}
\]

\subsection{2-D Continuum Deformation Coordination}
\label{2-D Continuum Deformation Coordination}
Consider an MVS with reference formation shown in Fig. \ref{obstacleladen} (a). The MVS consists of $N=27$ quadcopters;  quadcopters $1$, $2$, and $3$ are leaders and the {\color{black}remaining} quadcopters $4,5,\cdots,27$ are followers, e.g. $\mathcal{V}_L=\{1,2,3\}$ and $\mathcal{V}_F=\{4,\cdots,27\}$. Boundary quadcopters are defined by {\color{black}the} set 
$
\mathcal{V}_B=\{1,4,5,10,12,14,17,24,25,26,27\}
$
and $N_a=11$ auxiliary nodes are defined by the set 
$
\mathcal{V}_{aux}=\{28,29,\cdots,38\}$.
Without loss of generality, reference positions of auxiliary nodes and boundary quadcopters are the same (Fig. \ref{obstacleladen} (a)). {\color{black}Unidirectional} and bidirectional communications are shown by one-sided green arrows and double-sided blue arrows, respectively. Leaders move independently. {\color{black}}Therefore, no edge is incident to a node $i\in \mathcal{V}_L$. Every follower communicates with three in-neighbor nodes
{\color{black}where} followers' {\color{black}reference} communication weights are determined by using Eq. \eqref{communicationwitfollowers}  given quadcopters' reference positions shown in Fig. \ref{obstacleladen} (a). Every auxiliary node communicates with leader agents $1$, $2$, and $3$, e.g. $\mathcal{N}_i=\{1,2,3\}$ if $i\in \mathcal{V}_{aux}$. Note that communication links between auxiliary nodes and leaders are not shown in Fig. \ref{obstacleladen} (a).
As shown in Fig. \ref{obstacleladen} (a), reference {\color{black}VCS} is a containing triangle with vertices positioned at $\mathbf{h}_{1,0}=(0,40)$, $\mathbf{h}_{2,0}=(60,100)$, and $\mathbf{h}_{3,0}=(0,140)$. Given reference {\color{black}VCS} configuration and {\color{black}the} MVS formation, $d_s=5.5875$ is the minimum separation distance between two quadcopters and $d_b=4.5358$ is the minimum distance from the {\color{black}VCS} sides. Assuming every quadcopter is enclosed by a ball with radius $\epsilon=0.5m$, $\delta_{\mathrm{max}}=2.2938m$ is obtained by using Eq. \eqref{deltaaamax}. 
\begin{figure}
 \centering
 \subfigure[$t=0s$]{\includegraphics[width=0.47\linewidth]{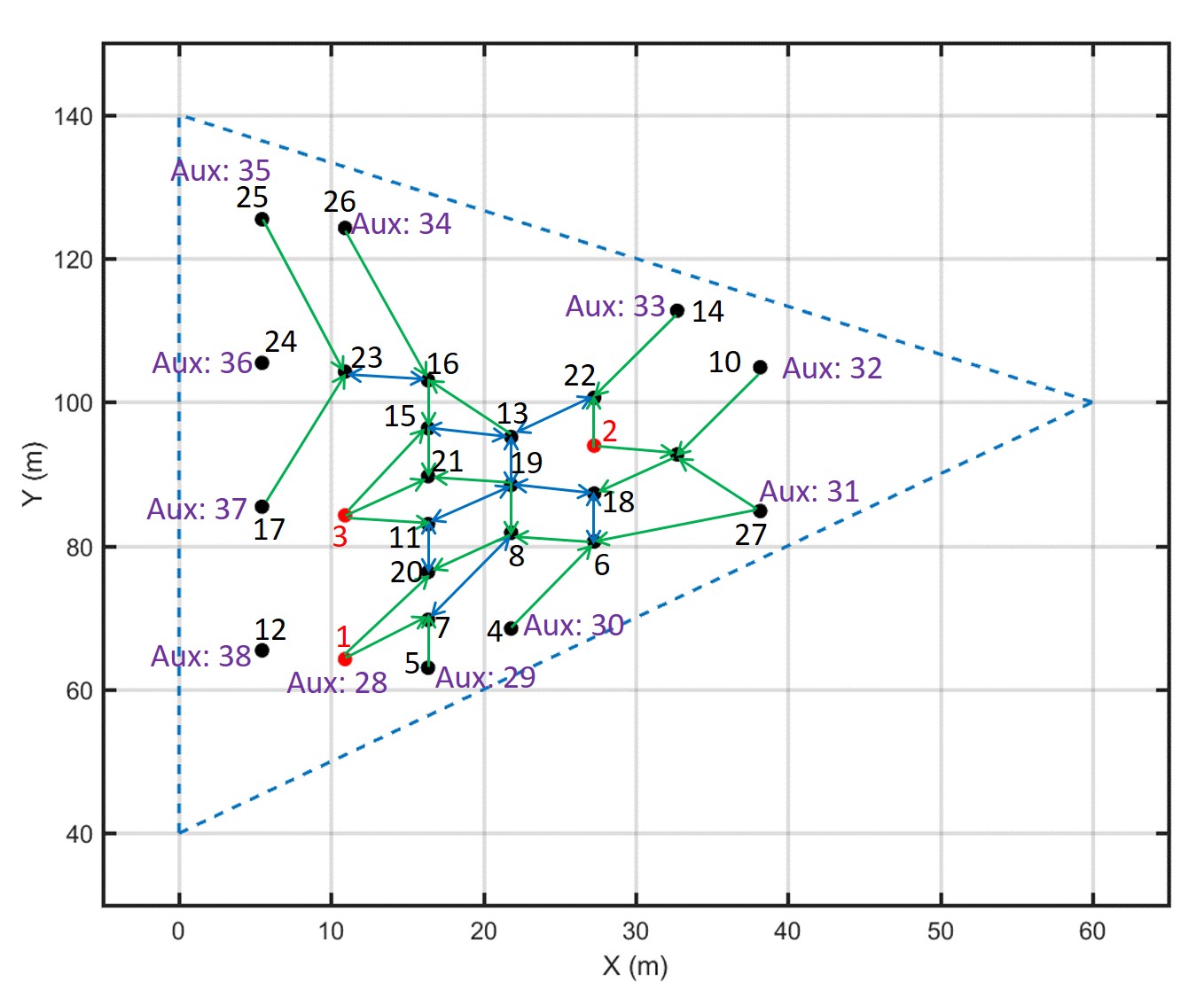}}
  \subfigure[$t=100s$]{\includegraphics[width=0.47\linewidth]{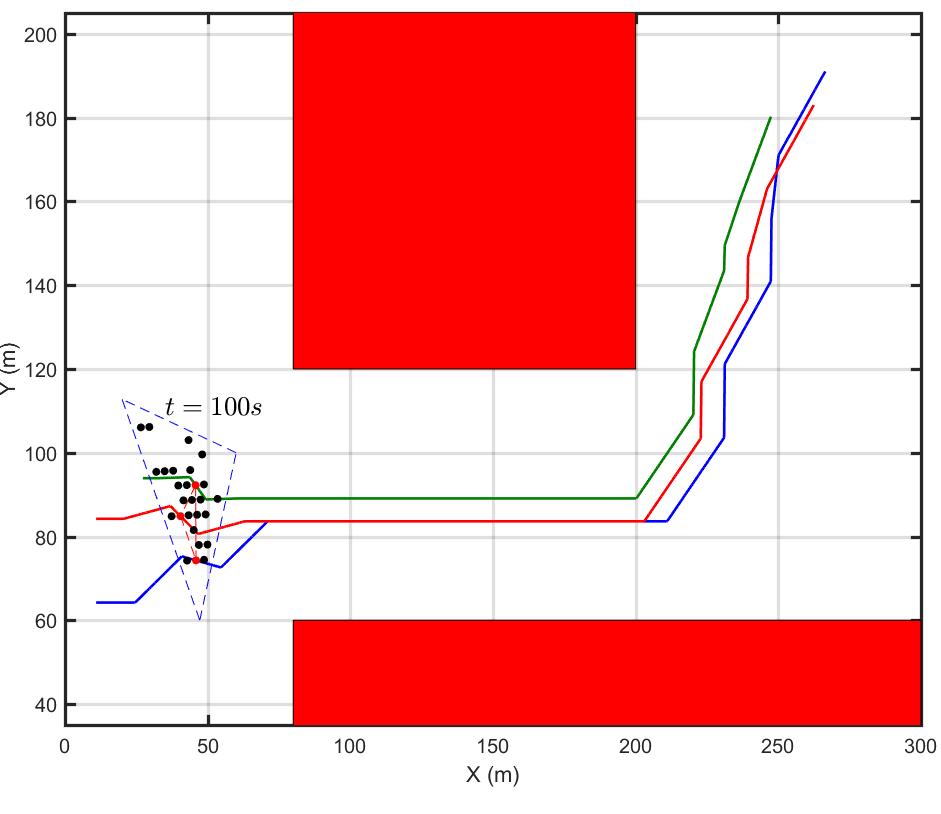}}
  \subfigure[$t=500s$]{\includegraphics[width=0.47\linewidth]{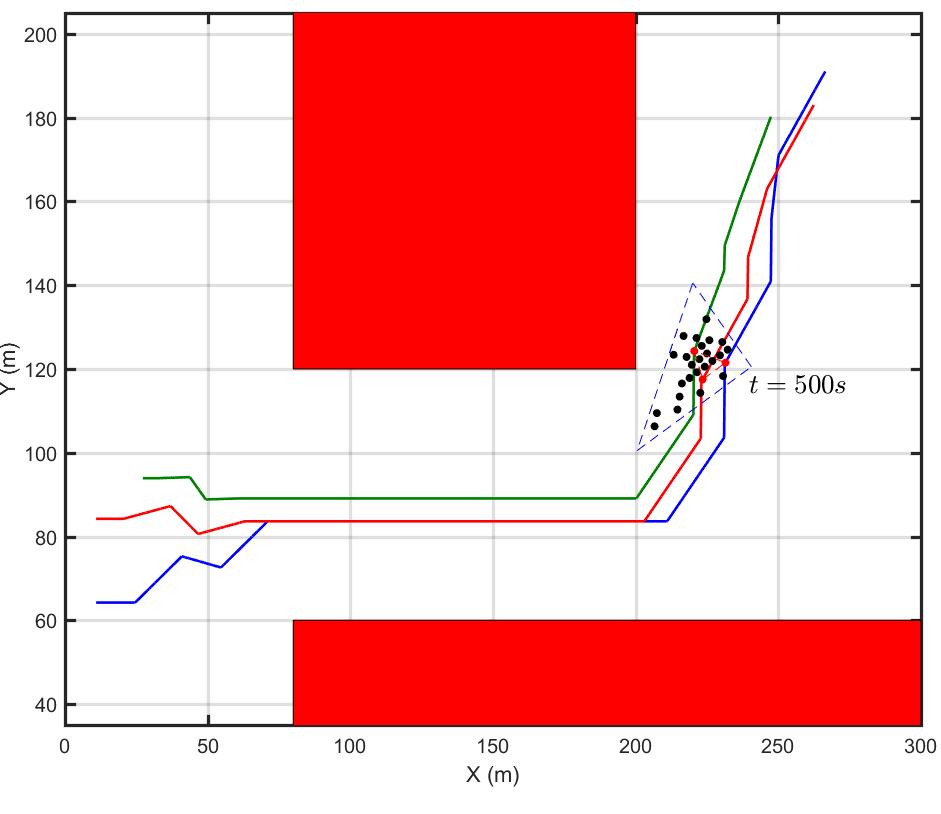}}
  \subfigure[$t=812s$]{\includegraphics[width=0.47\linewidth]{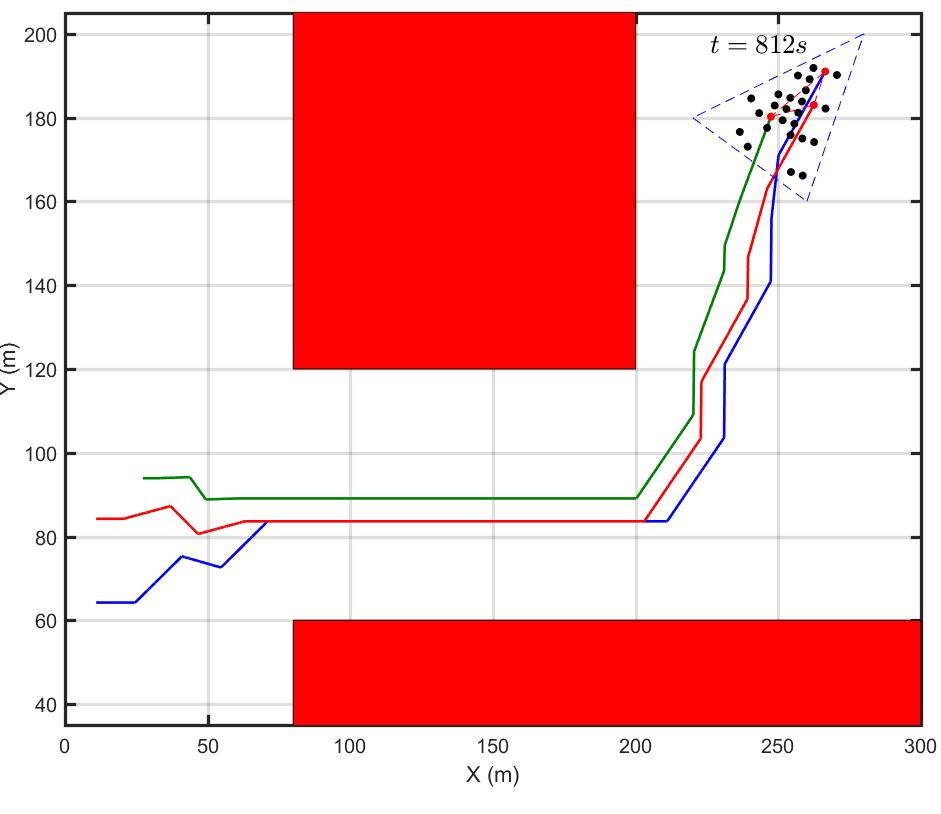}}
  \subfigure[Eigenvalues of matrix $\mathbf{U}_D$ versus time.]{\includegraphics[width=0.47\linewidth]{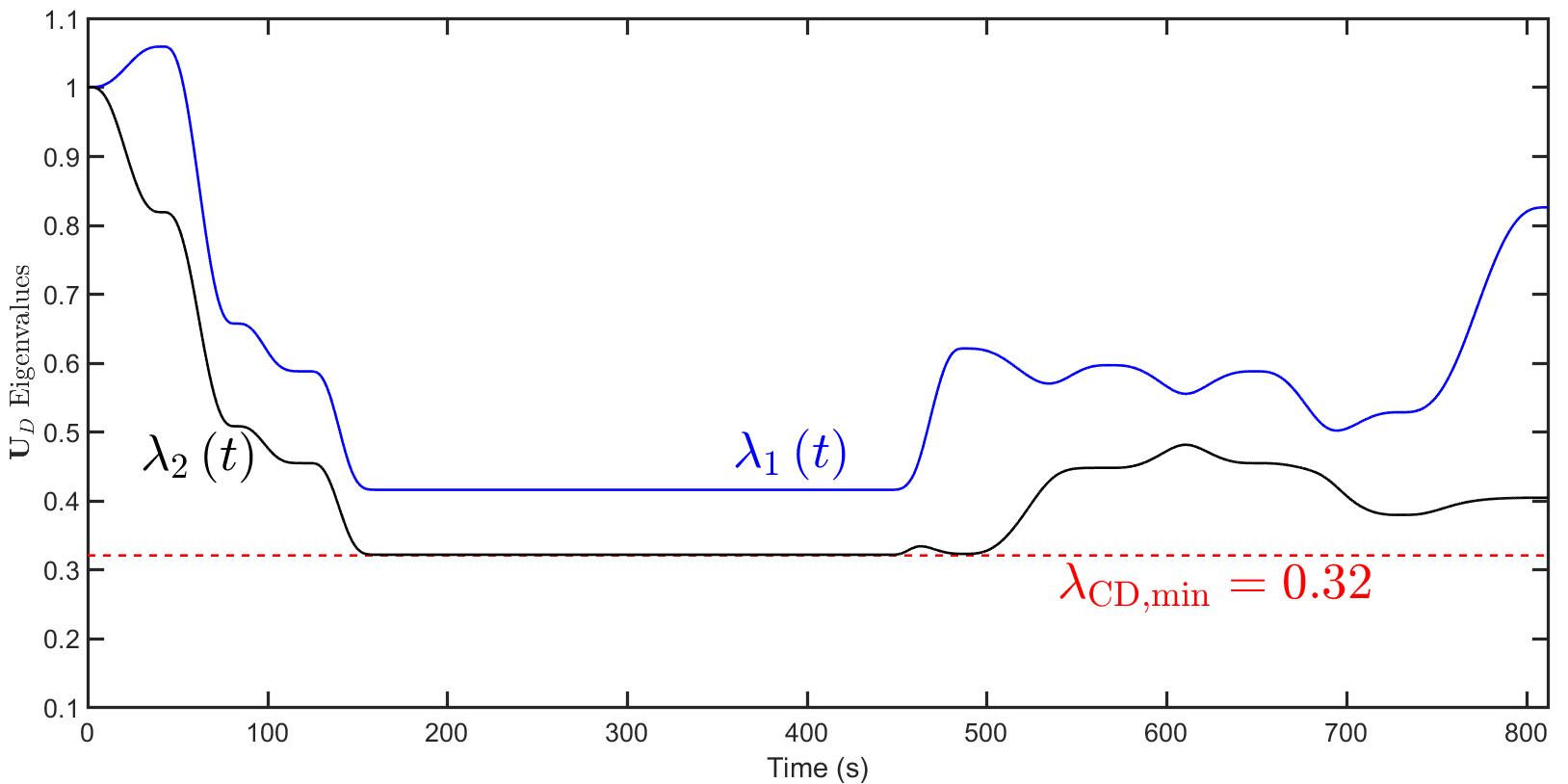}}
  \subfigure[Deviation of quadcopters from their global desired trajectories]{\includegraphics[width=0.47\linewidth]{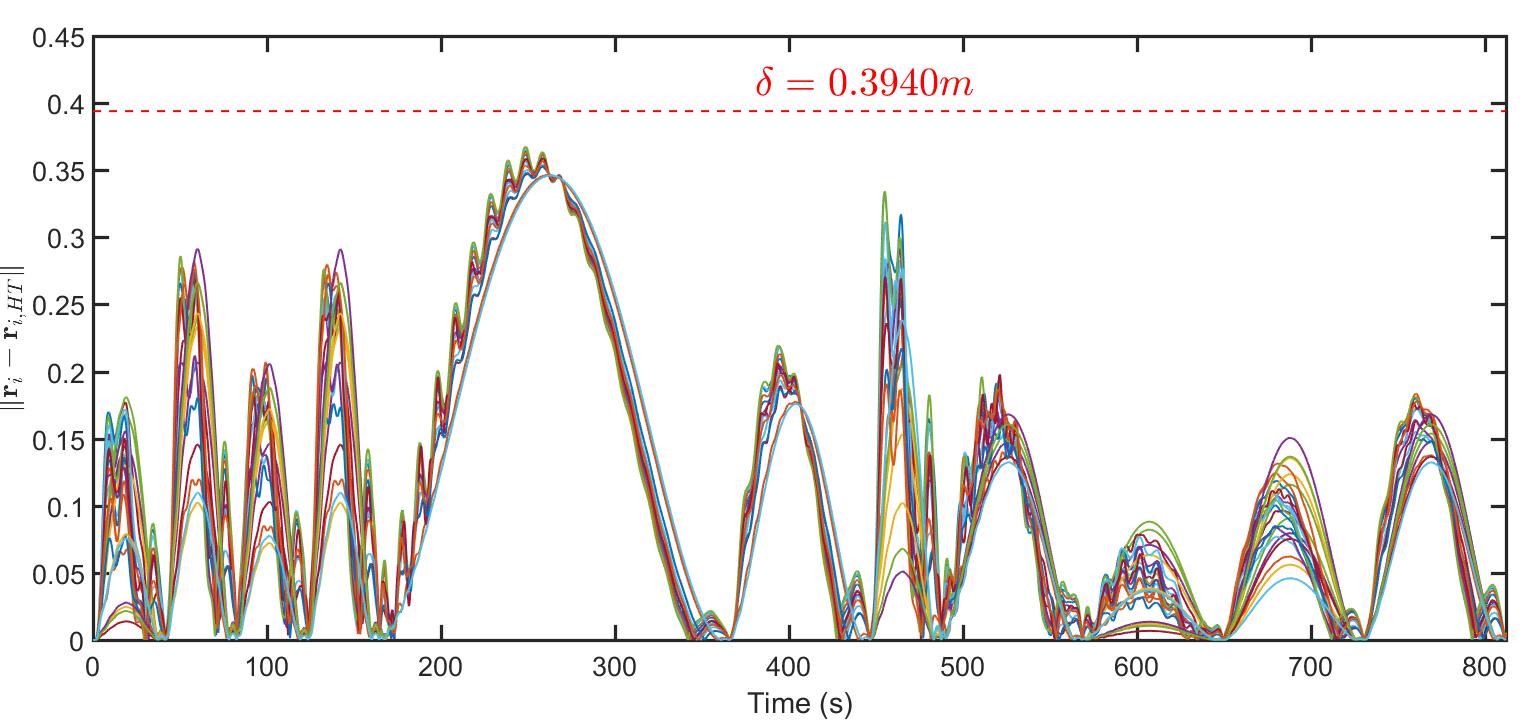}}
     \caption{(a-e) MVS at sample times $0s$, $100s$, $500s$, and $812s$. (e) $\mathbf{U}_D$ eigenvalues $\lambda_1$ and $\lambda_2$ versus time. (f) Deviation of each follower $i\in \mathcal{V}_F$ versus time. Note that $\sup\limits_{t}\|\mathbf{r}_i-\mathbf{r}_{i,HT}\|\leq \delta=0.3940m,~\forall i\in \mathcal{V}_F$.}
\label{obstacleladen}
\end{figure}
We consider MVS collective motion in an obstacle laden environment {\color{black}} shwon in Fig. \ref{obstacleladen}. Because {\color{black}the} $z$ component of {\color{black}the} MVS initial formation is $0$, reference and initial MVS formations are considered the same; $t_0=t_s=0$, $\mathbf{R}_s=\mathbf{I}_3\in \mathbb{R}^{3\times 3}$, and $\mathbf{d}_s=\mathbf{0}\in \mathbb{R}^{3\times 1}$. Optimal leaders' paths, minimizing travel distance between initial and target MVS formations, are determined using A* search and shown in Figs. \ref{obstacleladen} (a-h) (See Section \eqref{Obstacle-Laden Environment}). Given leaders' trajectories, followers use the communication graph shown in Fig. \ref{obstacleladen} (a) to acquire the desired continuum deformation through local communication. MVS formations at sample times $0s$, 
$100s$, 
$500s$, 
and $812s$ are shown in Figs. \ref{obstacleladen} (a-d). Blue and red triangles show the containing {\color{black}VCS} and leading triangle configurations in Figs. \ref{obstacleladen} (a-d). Given leaders' positions at reference time $t=0$ and current time time $t$, $\mathbf{U}_D$ eigenvalues, $\lambda_1$ and $\lambda_2$, are plotted versus time in Fig. \ref{obstacleladen} (e), e.g. $\lambda_3\left(t\right)=1$ at any time $t$. As shown $\lambda_{\mathrm{CD,min}}=0.32$ is {\color{black}a} lower limit for {\color{black}the} $\mathbf{U}_D$ eigenvalues. Substituting  $\lambda_{\mathrm{CD,min}}=0.32$, quadcopter size $\epsilon=0.5m$, $\delta_{\mathrm{max}}=2.2938m$ into Eq. \eqref{lcmin}, $\delta$ is obtained from:
\[
\delta=\lambda_{\mathrm{CD,min}}\left(\delta_{\mathrm{max}}+\epsilon\right)-\epsilon=0.3940m.
\]
Transient error $\|\mathbf{r}_i-\mathbf{r}_{i,HT}\|$ is plotted versus {\color{black}time} for every follower quadcopter $i\in \mathcal{V}_F$ in Fig. \ref{obstacleladen} (f). {\color{black}Note that the} deviation of every follower {\color{black}remains less than} $\delta=0.3940m$. Therefore, inequality constraints \eqref{deltainequality} and \eqref{NOinteragent} are satisfied{\color{black}, and}  inter-agent and obstacle collision avoidance, as well as quadcopter containment {\color{black}constraints} are {\color{black}satisfied}.

\begin{figure}
 \centering
  \subfigure[MVS initial formation]{\includegraphics[width=0.7\linewidth]{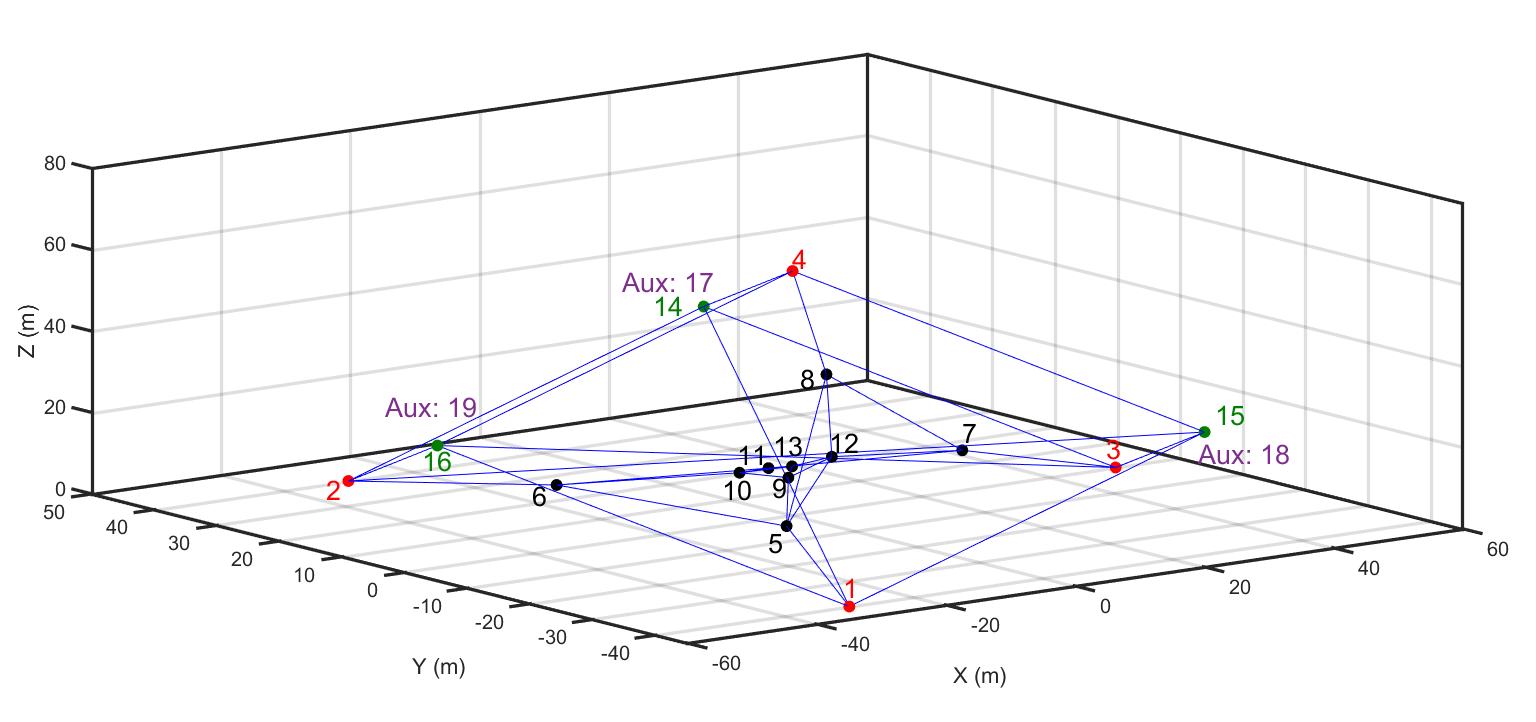}}
  \subfigure[MVS evolution]{\includegraphics[width=0.7\linewidth]{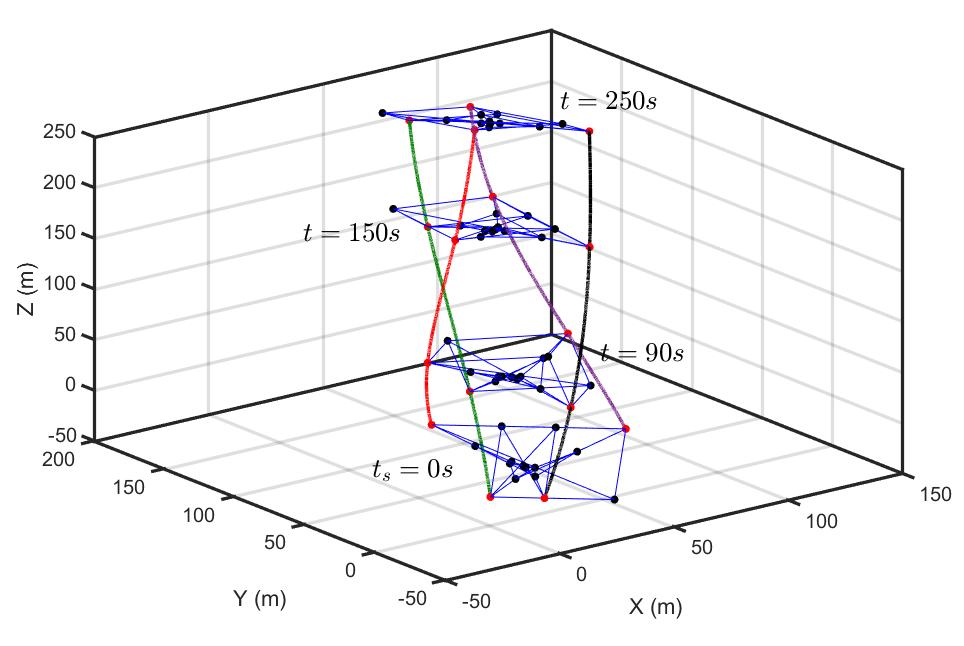}}
  \subfigure[Transient error]{\includegraphics[width=0.9\linewidth]{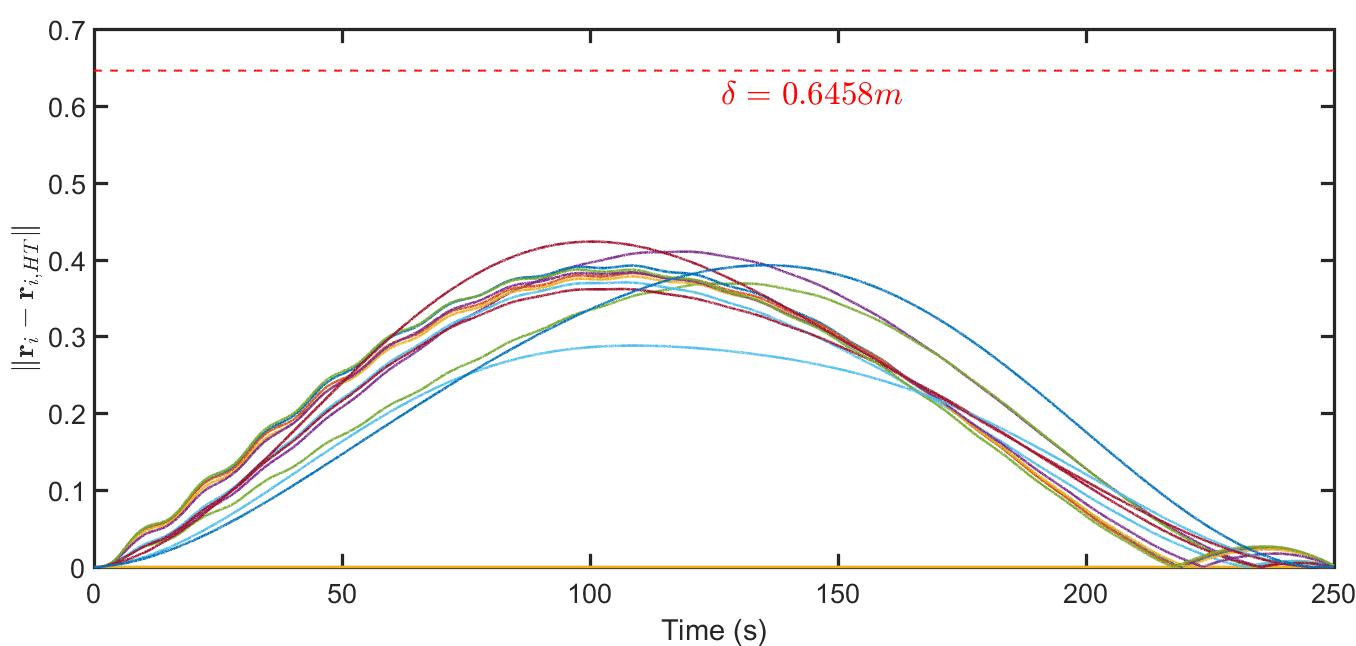}}
     \caption{(a) MVS initial formation and inter-agent communication in case-study 4. (b) MVS at times $0s$, $90s$, $150s$, $250s$. Red and black nodes show leaders and followers, respectively. 
     Leader $1$, $2$, $3$, $4$ paths are shown by black, {\color{black}green}, red, and pink curves. (c) Deviation of each follower $i\in \mathcal{V}_F$ versus time in case study 4. Note that $\sup\limits_{t}\|\mathbf{r}_i-\mathbf{r}_{i,HT}\|\leq \delta=0.6652m,~\forall i\in \mathcal{V}_F$.}
\label{InitialForm}
\end{figure}
\subsection{3-D Continuum Deformation Coordination}
\label{3-D Continuum Deformation Coordination}
We simulate collective takeoff with $N=16$ quadcopter MVS. Leaders are indexed by $1$ through $4$ ($\mathcal{V}_L=\{1,2,3,4\}$), followers are defined by set $\mathcal{V}_F=\{5,\cdots,16\}$ where $\mathcal{V}_B=\{1,2,3,4,14,15,16\}$ specifies the boundary quadcopters. $N_a=3$ auxiliary nodes are defined by $\mathcal{V}_{aux}=\{17,18,19\}$ where $\mathbf{r}_{14,0}=\mathbf{r}_{17,0}$, $\mathbf{r}_{15,0}=\mathbf{r}_{18,0}$, and $\mathbf{r}_{16,0}=\mathbf{r}_{19,0}$. Reference positions of auxiliary nodes $17$, $18$, and $19$ and boundary quadcopters $14$, $15$, and $16$ are the same. 
{\color{black}The} MVS initial formation is shown in Fig. \ref{InitialForm} (a) where leaders are illustrated by red and followers $5$ through $13$ are shown by black. In addition, follower quadcopters $14$, $15$, and $16$ (auxilliary nodes $17$, $18$, and $19$) are shown {\color{black}in} green in Fig. \ref{InitialForm} (b). Initial and reference formations are considered the same in a 3-D continuum deformation coordination. Follower communication weights are assigned using Eq. \eqref{communicationwitfollowers}.  Table  \ref{tab:2} lists initial positions and {\color{black}reference} communication weights per Eq. \eqref{communicationwitfollowers}.

\begin{table}
    \small
    \centering
        \caption{Quadcopters' initial positions and {\color{black}reference} communication weights {\color{black}of followers} in case study 4}
    \scalebox{0.69}{
    \begin{tabular}{|c| c| c| c| c| c| c| c| c |c |c |c| c| }
    \hline
    $i$&$x_{i,0}$&$y_{i,0}$&$z_{i,0}$&$i_1$&$i_2$&$i_3$&$i_4$&$w_{i,i_1}$& $w_{i,i_2}$&$w_{i,i_3}$& $w_{i,i_4}$ \\
    \hline
    1&  -30& -40 &   0 &   -&   - &   -    &-&-&-&-&-\\
    2&  -30&  40 &   0 &   -&   - &   -    &-&-&-&-&-\\
    3&  50&  0 &   0 &   -&   - &   -    &-&-&-&-&-\\
    4&  0&  0 &   60 &   -&   - &   -    &-&-&-&-&-\\
    \hline
    5&  -19.07&  -18.70 &   8.99&   1&    6 &   8    &9&0.5&${1\over6}$&${1\over6}$&${1\over6}$\\
    6&-19.43&   17.65&   5.16 &  2  &  5  & 10 &  11&0.5&${1\over6}$&${1\over6}$&${1\over6}$\\
    7& 25.05 &  -1.25 &  10.56 &   3  &  8  & 11&   12&0.5&${1\over6}$&${1\over6}$&${1\over6}$\\
    8&    1.06 &  -4.30 & 36.01 &   4 &   5  &  7&   12&0.5&${1\over6}$&${1\over6}$&${1\over6}$\\
    9& -6.03&   -5.54  & 12.77 &   5&   10&   12 &  13&0.2&0.2&0.2&0.4\\
    10& -6.35   & 1.94 & 11.17&   6  &  9  & 11 &  13&0.2&0.2&0.2&0.4\\
    11&-1.17 &  2.65&  10.80  &  6 &  10 &   7 &  13&0.2&0.2&0.2&0.4\\
    12&    0.39 &  -5.87 &  16.54  &  7  &  8 &   5  & 13&0.2&0.2&0.2&0.4\\
    13&   -2.55&-2.54&  13.56 &   9 &  10  & 11 &  12&0.2&0.2&0.2&0.4\\
    \hline
    14 (Aux: 17)    &   25  & 40  & 30  &  1  &  2    &3   & 4&0-0.5&0.5&0.5&0.5\\
    15 (Aux: 18)
    &   25 & -40 &  30  &  1 &   2&    3  &  4&0.5&-0.5&0.5&0.5\\
    16 (Aux: 19)
    &  -55 &   0 &  30  &  1 &   2 &   3 &   4&0.5&0.5&-0.5&0.5\\
    \hline
    \end{tabular}}
    \label{tab:2}
\end{table}

{\color{black}Given quadcopter initial reference positions, the shear deformation angles $\theta_{u,0}=-0.1721~\mathrm{rad}$ and $\psi_{u,0}=0.7130~\mathrm{rad}$ are obtained using Eq. \eqref{thspssassignment}. Furthermore,  we choose $\phi_{u,0}=0$ at any time $t$ (see Remark \ref{importantremark}).}
 We consider MVS collective motion over the time interval $[0,250]$ ($t_s=0s$ and $t_f=250$) where 
$\phi_r\left(t_f\right)=0~\mathrm{rad}$, $\theta_r\left(t_f\right)=0.0713~\mathrm{rad}$, $\psi_r\left(t_f\right)={\pi\over 2}~~\mathrm{rad}$, $d_1\left(t_f\right)=100m$, $d_2\left(t_f\right)=165m$, and $d_3\left(t_f\right)=200m$. Given $\mathbf{s}_{\mathrm{OF}}^3\left(0\right)$ and $\mathbf{s}_{\mathrm{OF}}^3\left(250\right)$, {\color{black}a} homogeneous transformation is defined by Eq. \eqref{homogtrans} and acquired by followers through local communication. MVS formations at sample times $t_s=0s$, $t=90s$, $t=150s$, and $t=250s$ are shown in Fig. \ref{InitialForm} (b).

\textbf{Inter-Agent Collision Avoidance:} Given {\color{black}the} MVS initial formation, $d_s=4.6607m$ is the minimum separation distance. Given $\lambda_{\mathrm{min}}=\lambda_1\left(t_f\right)=0.5$, deviation upper-bound $\delta$ becomes
$
\delta=\lambda_{\mathrm{CD,min}}\left(\delta_{\mathrm{max}}+\epsilon\right)-\epsilon=0.6458m.
$ Deviation of every follower from {\color{black}the} desired position, defined by the continuum deformation, is plotted versus time in Fig. \ref{InitialForm} (c). Because $\sup\limits_{t}\|\mathbf{r}_i(t)-\mathbf{r}_{i,HT}(t)\|\leq 0.6652$ ($\forall i\in \mathcal{V}_F$) inter-agent collision avoidance is gauranteed.


\section{Conclusion}
\label{Conclusion}
This paper advanced continuum deformation coordination by relaxing existing containment {\color{black}constraints}. We {\color{black}showed} that any $n+1$ agents forming an $n$-D simplex can be considered as leaders; followers can be placed inside or outside the leading simplex in an $n$-D homogeneous transformation ($n=1,2,3$). This paper also formulated continuum deformation coordination eigen-decomposition to determine a nonsingular mapping between leader position components and homogeneous transformation features assigned by continuum deformation eigen-decomposition. With this approach, leader trajectories ensuring collision avoidance and quadcopter containment can be {\color{black}safely} planned. {\color{black}Furthermore, this paper advances the exisiting condition for inter-agent collision avoidance in a large-scale continuum deformation. This new safety condition is much less restrictive and significantly advances the maneuverability and flexibility of the continuum deformation coordination.}

\bibliographystyle{IEEEtran}
\bibliography{reference}

\begin{thebibliography}{10}
\providecommand{\url}[1]{#1}
\csname url@samestyle\endcsname
\providecommand{\newblock}{\relax}
\providecommand{\bibinfo}[2]{#2}
\providecommand{\BIBentrySTDinterwordspacing}{\spaceskip=0pt\relax}
\providecommand{\BIBentryALTinterwordstretchfactor}{4}
\providecommand{\BIBentryALTinterwordspacing}{\spaceskip=\fontdimen2\font plus
\BIBentryALTinterwordstretchfactor\fontdimen3\font minus
  \fontdimen4\font\relax}
\providecommand{\BIBforeignlanguage}[2]{{%
\expandafter\ifx\csname l@#1\endcsname\relax
\typeout{** WARNING: IEEEtran.bst: No hyphenation pattern has been}%
\typeout{** loaded for the language `#1'. Using the pattern for}%
\typeout{** the default language instead.}%
\else
\language=\csname l@#1\endcsname
\fi
#2}}
\providecommand{\BIBdecl}{\relax}
\BIBdecl

\bibitem{lin2014distributed}
Z.~Lin, L.~Wang, Z.~Han, and M.~Fu, ``Distributed formation control of
  multi-agent systems using complex laplacian,'' \emph{IEEE Transactions on
  Automatic Control}, vol.~59, no.~7, pp. 1765--1777, 2014.

\bibitem{wang2015bswarm}
X.~Wang, J.~Ren, X.~Jin, and D.~Manocha, ``Bswarm: biologically-plausible
  dynamics model of insect swarms,'' in \emph{Proceedings of the 14th ACM
  SIGGRAPH/Eurographics Symposium on Computer Animation}.\hskip 1em plus 0.5em
  minus 0.4em\relax ACM, 2015, pp. 111--118.

\bibitem{boissier2013multi}
O.~Boissier, R.~H. Bordini, J.~F. H{\"u}bner, A.~Ricci, and A.~Santi,
  ``Multi-agent oriented programming with jacamo,'' \emph{Science of Computer
  Programming}, vol.~78, no.~6, pp. 747--761, 2013.

\bibitem{dong2015time}
X.~Dong, B.~Yu, Z.~Shi, and Y.~Zhong, ``Time-varying formation control for
  unmanned aerial vehicles: Theories and applications,'' \emph{IEEE Trans. on
  Control Systems Tech.}, vol.~23, no.~1, pp. 340--348, 2015.

\bibitem{nazari2016decentralized}
M.~Nazari, E.~A. Butcher, T.~Yucelen, and A.~K. Sanyal, ``Decentralized
  consensus control of a rigid-body spacecraft formation with communication
  delay,'' \emph{Journal of Guidance, Control, and Dynamics}, vol.~39, no.~4,
  pp. 838--851, 2016.

\bibitem{ren2002virtual}
W.~Ren and R.~Beard, ``Virtual structure based spacecraft formation control
  with formation feedback,'' in \emph{AIAA Guidance, Navigation, and control
  conference and exhibit}, 2002, p. 4963.

\bibitem{low2011flexible}
C.~B. Low and Q.~San~Ng, ``A flexible virtual structure formation keeping
  control for fixed-wing uavs,'' in \emph{Control and Automation (ICCA), 2011
  9th IEEE Intl. Conf. on}.\hskip 1em plus 0.5em minus 0.4em\relax IEEE, 2011,
  pp. 621--626.

\bibitem{ren2007information}
W.~Ren, R.~W. Beard, and E.~Atkins, ``Information consensus in multivehicle
  cooperative control,'' \emph{IEEE Control Systems Magazine}, vol.~27, pp.
  71--82, 2007.

\bibitem{feng2014group}
Y.~Feng, S.~Xu, and B.~Zhang, ``Group consensus control for double-integrator
  dynamic multiagent systems with fixed communication topology,''
  \emph{International Journal of Robust and Nonlinear Control}, vol.~24, no.~3,
  pp. 532--547, 2014.

\bibitem{zhu2010leader}
W.~Zhu and D.~Cheng, ``Leader-following consensus of second-order agents with
  multiple time-varying delays,'' \emph{Automatica}, vol.~46, no.~12, pp.
  1994--1999, 2010.

\bibitem{xiao2006state}
F.~Xiao and L.~Wang, ``State consensus for multi-agent systems with switching
  topologies and time-varying delays,'' \emph{International Journal of
  Control}, vol.~79, no.~10, pp. 1277--1284, 2006.

\bibitem{liu2018exponential}
H.~Liu, L.~Cheng, M.~Tan, and Z.-G. Hou, ``Exponential finite-time consensus of
  fractional-order multiagent systems,'' \emph{IEEE Transactions on Systems,
  Man, and Cybernetics: Systems}, 2018.

\bibitem{li2018nonlinear}
Y.~Li, C.~Tang, K.~Li, S.~Peeta, X.~He, and Y.~Wang, ``Nonlinear finite-time
  consensus-based connected vehicle platoon control under fixed and switching
  communication topologies,'' \emph{Transportation Research Part C: Emerging
  Technologies}, vol.~93, pp. 525--543, 2018.

\bibitem{cao2015leader}
W.~Cao, J.~Zhang, and W.~Ren, ``Leader--follower consensus of linear
  multi-agent systems with unknown external disturbances,'' \emph{Systems \&
  Control Letters}, vol.~82, pp. 64--70, 2015.

\bibitem{shao2018leader}
J.~Shao, W.~X. Zheng, T.-Z. Huang, and A.~N. Bishop, ``On leader-follower
  consensus with switching topologies: An analysis inspired by pigeon
  hierarchies,'' \emph{IEEE Transactions on Automatic Control}, 2018.

\bibitem{olfati2004consensus}
R.~Olfati-Saber and R.~M. Murray, ``Consensus problems in networks of agents
  with switching topology and time-delays,'' \emph{IEEE Transactions on
  automatic control}, vol.~49, no.~9, pp. 1520--1533, 2004.

\bibitem{zhang2010consensus}
Y.~Zhang and Y.-P. Tian, ``Consensus of data-sampled multi-agent systems with
  random communication delay and packet loss,'' \emph{IEEE Transactions on
  Automatic Control}, vol.~55, no.~4, pp. 939--943, 2010.

\bibitem{li2016containment}
B.~Li, Z.-q. Chen, Z.-x. Liu, C.-y. Zhang, and Q.~Zhang, ``Containment control
  of multi-agent systems with fixed time-delays in fixed directed networks,''
  \emph{Neurocomputing}, vol. 173, pp. 2069--2075, 2016.

\bibitem{li2015containment}
W.~Li, L.~Xie, and J.-F. Zhang, ``Containment control of leader-following
  multi-agent systems with markovian switching network topologies and
  measurement noises,'' \emph{Automatica}, vol.~51, pp. 263--267, 2015.

\bibitem{liu2014containment}
K.~Liu, G.~Xie, and L.~Wang, ``Containment control for second-order multi-agent
  systems with time-varying delays,'' \emph{Systems \& Control Letters},
  vol.~67, pp. 24--31, 2014.

\bibitem{wang2014distributed}
X.~Wang, S.~Li, and P.~Shi, ``Distributed finite-time containment control for
  double-integrator multiagent systems,'' \emph{IEEE Transactions on
  Cybernetics}, vol.~44, no.~9, pp. 1518--1528, 2014.

\bibitem{liu2015containment}
H.~Liu, L.~Cheng, M.~Tan, and Z.-G. Hou, ``Containment control of
  continuous-time linear multi-agent systems with aperiodic sampling,''
  \emph{Automatica}, vol.~57, pp. 78--84, 2015.

\bibitem{zhao2015finite}
Y.~Zhao and Z.~Duan, ``Finite-time containment control without velocity and
  acceleration measurements,'' \emph{Nonlinear Dynamics}, vol.~82, no. 1-2, pp.
  259--268, 2015.

\bibitem{zhao2015robust}
Y.-P. Zhao, P.~He, H.~Saberi~Nik, and J.~Ren, ``Robust adaptive synchronization
  of uncertain complex networks with multiple time-varying coupled delays,''
  \emph{Complexity}, vol.~20, no.~6, pp. 62--73, 2015.

\bibitem{notarstefano2011containment}
G.~Notarstefano, M.~Egerstedt, and M.~Haque, ``Containment in leader--follower
  networks with switching communication topologies,'' \emph{Automatica},
  vol.~47, no.~5, pp. 1035--1040, 2011.

\bibitem{rastgoftar2016continuum}
H.~Rastgoftar, \emph{Continuum Deformation of Multi-Agent Systems}.\hskip 1em
  plus 0.5em minus 0.4em\relax Birkh{\"a}user, 2016.

\bibitem{rastgoftar2017continuum}
H.~Rastgoftar and E.~M. Atkins, ``Continuum deformation of multi-agent systems
  under directed communication topologies,'' \emph{Journal of Dynamic Systems,
  Measurement, and Control}, vol. 139, no.~1, p. 011002, 2017.

\bibitem{lal2006continuum}
M.~Lal, D.~Maithripala, and S.~Jayasuriya, ``A continuum approach to global
  motion planning for networked agents under limited communication,'' in
  \emph{Information and Automation, 2006. ICIA 2006. International Conference
  on}.\hskip 1em plus 0.5em minus 0.4em\relax IEEE, 2006, pp. 337--342.

\bibitem{yu2010group}
J.~Yu and L.~Wang, ``Group consensus in multi-agent systems with switching
  topologies and communication delays,'' \emph{Systems \& Control Letters},
  vol.~59, no.~6, pp. 340--348, 2010.

\bibitem{peng2007distributed}
L.~Peng, J.~Yingmin, D.~Junping, and Y.~Shiying, ``Distributed consensus
  control for second-order agents with fixed topology and time-delay,'' in
  \emph{Chinese Control Conference}.\hskip 1em plus 0.5em minus 0.4em\relax
  IEEE, 2007, pp. 577--581.

\bibitem{ji2008containment}
M.~Ji, G.~Ferrari-Trecate, M.~Egerstedt, and A.~Buffa, ``Containment control in
  mobile networks,'' \emph{IEEE Transactions on Automatic Control}, vol.~53,
  no.~8, pp. 1972--1975, 2008.

\bibitem{lai2009introduction}
W.~M. Lai, D.~H. Rubin, D.~Rubin, and E.~Krempl, \emph{Introduction to
  continuum mechanics}.\hskip 1em plus 0.5em minus 0.4em\relax
  Butterworth-Heinemann, 2009.

\bibitem{rastgoftar2018asymptotic}
H.~Rastgoftar, H.~G. Kwatny, and E.~M. Atkins, ``Asymptotic tracking and
  robustness of mas transitions under a new communication topology,''
  \emph{IEEE Transactions on Automation Science and Engineering}, vol.~15,
  no.~1, pp. 16--32, 2018.

\bibitem{zhao2018affine}
S.~Zhao, ``Affine formation maneuver control of multiagent systems,''
  \emph{IEEE Transactions on Automatic Control}, vol.~63, no.~12, pp.
  4140--4155, 2018.

\bibitem{xu2018affine}
Y.~Xu, S.~Zhao, D.~Luo, and Y.~You, ``Affine formation maneuver control of
  multi-agent systems with directed interaction graphs,'' in \emph{2018 37th
  Chinese Control Conference (CCC)}.\hskip 1em plus 0.5em minus 0.4em\relax
  IEEE, 2018, pp. 4563--4568.

\bibitem{rastgoftar2014continuum}
H.~Rastgoftar and S.~Jayasuriya, ``Continuum evolution of multi agent systems
  under a polyhedral communication topology,'' in \emph{American Control
  Conference (ACC), 2014}.\hskip 1em plus 0.5em minus 0.4em\relax IEEE, 2014,
  pp. 5115--5120.

\bibitem{rastgoftar2018safe}
H.~Rastgoftar, E.~M. Atkins, and D.~Panagou, ``Safe multi-quadcopter system
  continuum deformation over moving frames,'' \emph{IEEE Transactions on
  Control of Network Systems}, 2018.

\bibitem{qu2009cooperative}
Z.~Qu, \emph{Cooperative control of dynamical systems: applications to
  autonomous vehicles}.\hskip 1em plus 0.5em minus 0.4em\relax Springer Science
  \& Business Media, 2009.

\bibitem{slotine1991applied}
J.-J.~E. Slotine, W.~Li \emph{et~al.}, \emph{Applied nonlinear control}.\hskip
  1em plus 0.5em minus 0.4em\relax Prentice hall Englewood Cliffs, NJ, 1991,
  vol. 199, no.~1.

\bibitem{rastgoftar2015swarm}
H.~Rastgoftar and S.~Jayasuriya, ``Swarm motion as particles of a continuum
  with communication delays,'' \emph{Journal of Dynamic Systems, Measurement,
  and Control}, vol. 137, no.~11, p. 111008, 2015.

\end{thebibliography}



 \appendices
 \section{$n$-D Homogeneous Deformation Decomposition}\label{APA}

By decomposition, homogeneous transformation features are uniquely related to leaders' position components at any time $t\geq t_0$ given leaders' reference positions. Decomposition is straightforward for a 3-D homogeneous transformation ($n=3$) when four leaders define a desired homogeneous transformation, e.g.  $12$ leaders position components can be simply related to the homogeneous transformation features. However, homogeneous transformation decomposition {\color{black}will not be straightforward} for 1-D and 2-D continuum deformation coordination, defined by $2$ and $3$ leaders, in a 3-D motion space because the number of leader position components differs from the number of homogeneous transformation features.

\subsubsection{1-D Homogeneous Transformation Decomposition}
\label{1-D Homogeneous Transformation Decomposition}
Assume that vehicles are {\color{black}treated as particles of a} 1-D deformable body distributed along a deformable line segment in a $3D$ motion space at a time $t$. The MVS transformation is guided by two leaders (agents $1$ and $2$) and the remaining agents are followers. For 1-D homogenous transformation, $\lambda_1\left(t\right)>0$, $\lambda_2\left(t\right)=\lambda_3\left(t\right)=1$ and $\phi_u\left(t\right)=\theta_u\left(t\right)=\psi_u\left(t\right)=\phi_r\left(t\right)=0$ at all {\color{black}times} $t$,  {\color{black}see} the first row of Table  \ref{tab:l}. 

\begin{proposition}\label{prop2} If agents are distributed along $\mathbf{U}_D$ eigenvector $\hat{\mathbf{u}}_1$ at time $t\geq t_0$,  they were distributed along {\color{black}$\hat{\mathbf{e}}_1=[1~0~0]^T$} at reference time $t_0$. 
\end{proposition}

\begin{proof}
Substituting $\phi_u=\theta_u=\psi_u=\phi_r=\theta_r=\psi_r=0$ in Eq. \eqref{DECOM},  $\hat{\mathbf{u}}_1=\hat{\mathbf{e}}_1=[1~0~0]^T$, $\hat{\mathbf{u}}_2=\hat{\mathbf{e}}_2=[0~1~0]^T$, and $\hat{\mathbf{u}}_3=\hat{\mathbf{e}}_3=[0~0~1]^T$. 
\end{proof}

Theorem \ref{theorem1} describes how to determine $\left(\lambda_1,\phi_r,\theta_r,d_1,d_2,d_3\right)$ {\color{black}leader  position components $\left(x_{1,HT},y_{1,HT},z_{1,HT},x_{2,HT},y_{2,HT},z_{2,HT}\right)$} at any time $t$.

\begin{theorem}\label{theorem1}
If leaders' reference positions ($\mathbf{r}_{1,0}$ and $\mathbf{r}_{3,0}$) and leaders' global desired positions {\color{black}$\mathbf{r}_{1,HT}\left(t\right)$ and $\mathbf{r}_{2,HT}\left(t\right)$} at current time $t\geq t_0$ are given, then,
\begin{equation}
\label{lambda1D}
    t\geq t_0,~~ \lambda_1\left(t\right)=\dfrac{\|\mathbf{r}_{2,HT}\left(t\right)-\mathbf{r}_{1,HT}\left(t\right)\|}{\|\mathbf{r}_{2,0}-\mathbf{r}_{1,0}\|}=\dfrac{\|\mathbf{r}_{2,HT}\left(t\right)-\mathbf{r}_{1,HT}\left(t\right)\|}{\big|x_{2,0}-x_{1,0}\big|}
\end{equation}
Also,
\begin{subequations}
\begin{equation}
    \hat{\mathbf{u}}_1\left(t\right)={\mathbf{r}_{2,HT}\left(t\right)-\mathbf{r}_{1,HT}\left(t\right)\over \mathbf{r}_{2,0}-\mathbf{r}_{1,0}},
\end{equation}
\begin{equation}
\label{thetaar}
    \theta_r=-\sin^{-1}\left(\hat{\mathbf{u}}_1\cdot\hat{\mathbf{e}}_3\right),
\end{equation}
\begin{equation}
\label{psiir}
    \psi_r=\tan^{-1}\left(\dfrac{\hat{\mathbf{u}}_1\cdot\hat{\mathbf{e}}_2}{\hat{\mathbf{u}}_1\cdot\hat{\mathbf{e}}_1}\right),
\end{equation}
\end{subequations}
where $\hat{\mathbf{e}}_1=[1~0~0]^T$, $\hat{\mathbf{e}}_2=[0~1~0]^T$, and $\hat{\mathbf{e}}_3=[0~0~1]^T$.
Additionally, $d_1$, $d_2$, and $d_3$ are assigned by Eq.  \eqref{homogtrans}: $\mathbf{d}\left(t\right)=[d_1\left(t\right)~d_2\left(t\right)~d_3\left(t\right)]^T=\mathbf{r}_{1,HT}\left(t\right)-\mathbf{Q}\mathbf{r}_{1,0}$. 

\end{theorem}
\begin{proof}
Leaders' desired positions satisfy Eq. \eqref{homogtrans}{\color{black};} thus, 
$
\forall t,\qquad \mathbf{r}_{2,HT}\left(t\right)-\mathbf{r}_{1,HT}\left(t\right)=\mathbf{Q}\left(\mathbf{r}_{2,0}-\mathbf{r}_{1,0}\right)
$
and 
\[
\resizebox{0.99\hsize}{!}{%
$
\begin{split}
\left(\mathbf{r}_{2,HT}-\mathbf{r}_{1,HT}\right)^T\left(\mathbf{r}_{2,HT}-\mathbf{r}_{1,HT}\right)=\left(\mathbf{r}_{2,0}-\mathbf{r}_{1,0}\right)^T\mathbf{R}_D^T\mathbf{U}_D^2\mathbf{R}_D\left(\mathbf{r}_{2,0}-\mathbf{r}_{1,0}\right),
\end{split}
$
}
\]
at any time $t\in [t_s,t_f]$, where $\mathbf{r}_{2,0}-\mathbf{r}_{1,0}=[\left(x_{2,0}-x_{1,0}\right)~0~0]^T$. Because $\mathbf{R}_D$ is orthogonal,  $\mathbf{R}_D^T\mathbf{R}_D=\mathbf{I}$ and $\lambda_1$ is obtained as given in Eq. \eqref{lambda1D}. 
Because vehicles' desired positions are distributed along $\hat{\mathbf{u}}_1\left(t\right)$ at any time $t$, $ \hat{\mathbf{u}}_1\left(t\right)=\dfrac{\mathbf{r}_{2,HT}\left(t\right)-\mathbf{r}_{1,HT}\left(t\right)}{\|\mathbf{r}_{2,HT}\left(t\right)-\mathbf{r}_{1,HT}\left(t\right)\|}.
$ 
Considering Proposition \ref{prop2}, vehicles were distributed along {\color{black}the} ${\color{black}\hat{\mathbf{e}}_1}$ axis at reference time $t_0$. Therefore,
\[
{\mathbf{r}}_{2,0}-{\mathbf{r}}_{1,0}=
\begin{bmatrix}
\|{\mathbf{r}}_{2,0}-{\mathbf{r}}_{1,0}\|&0&0
\end{bmatrix}
^T.
\]
Because $\mathbf{r}_{i,0}$ and $\mathbf{r}_{i,HT}\left(t\right)$ ($i=1,2$) satisfy Eq. \eqref{homogtrans}, we can write:
\[
\resizebox{0.99\hsize}{!}{%
$
\mathbf{r}_{2,HT}-\mathbf{r}_{1,HT}= \|\mathbf{r}_{2,HT}-\mathbf{r}_{1,HT}\|\hat{\mathbf{u}}_1=\lambda_1\|\mathbf{r}_{2,0}-\mathbf{r}_{1,0}\|
    \begin{bmatrix}
    \cos\theta_r\cos\psi_r\\\cos\theta_r\sin\psi_r\\-\sin\psi_r
    \end{bmatrix}
    .
    $
    }
\]
Therefore,
$
\hat{\mathbf{u}}_1=
\begin{bmatrix}
\cos\theta_r\cos\psi_r&\cos\theta_r\sin\psi_r&-\sin\theta_r
\end{bmatrix}
^T
$
and $\theta_r$ and $\psi_r$ {\color{black}are} per Eqs. \eqref{thetaar} and \eqref{psiir}. Given $\lambda_1\left(t\right)$, $\theta_r\left(t\right)$, and $\psi_r\left(t\right)$ at time $t$, $\mathbf{U}_D$ and $\mathbf{R}_D$ are determined using Eq. \eqref{ROTDEF}.  Substituting $\mathbf{r}_1\left(t\right)$ and $\mathbf{r}_{1,0}$ into Eq. \eqref{homogtrans}, 
$\mathbf{d}\left(t\right)=[d_1\left(t\right)~d_2\left(t\right)~d_3\left(t\right)]^T=\mathbf{r}_{1,HT}\left(t\right)-\mathbf{Q}\mathbf{r}_{1,0}$.
\end{proof}

\subsubsection{2-D Homogeneous Transformation Decomposition}
Assume that {\color{black}vehicles' global desired positions} are on a plane in a {\color{black}$3$-D} motion space at current time $t$. The MVS transformation is guided by three leaders (vehicles $1$, $2$, and $3$) and the remaining vehicles are followers. In a 2-D homogeneous transformation, $\lambda_3(t)=1$ and $\phi_u(t)=\theta_u(t)=0$ at any time $t$.  Therefore, $\mathbf{U}_D$ simplifies to
\begin{equation}
\label{matrixUD}
\begin{split}
    \mathbf{U}_D=
    \begin{bmatrix}
    \lambda_1\cos^2\psi_u+\lambda_2\sin^2\psi_u&\left(\lambda_1-\lambda_2\right)\sin\psi_u\cos\psi_u&0\\
    \left(\lambda_1-\lambda_2\right)\sin\psi_u\cos\psi_u&\lambda_1\sin^2\psi_u+\lambda_2\cos^2\psi_u&0\\
    0&0&1
    \end{bmatrix}
    .
\end{split}
\end{equation}
\begin{proposition}\label{prop3}
If vehicles are distributed on a plane with normal vector $\hat{\mathbf{u}}_3$ at any time $t$ ($\hat{\mathbf{u}}_1$, $\hat{\mathbf{u}}_2$, and $\hat{\mathbf{u}}_3$ are the eigenvectors of  $\mathbf{U}_D$),  they were on this plane at reference time $t_0$. 
\end{proposition}
\begin{proof}
Substituting $\phi_u=\theta_u=\psi_u=\phi_r=\theta_r=\psi_r=0$ in Eq. \eqref{ROTDEF},  $\hat{\mathbf{u}}_1=\hat{\mathbf{e}}_1=[1~0~0]^T$, $\hat{\mathbf{u}}_2=\hat{\mathbf{e}}_2=[0~1~0]^T$, and $\hat{\mathbf{u}}_3=\hat{\mathbf{e}}_3=[0~0~1]^T$. 
\end{proof}
\begin{remark}
Given leaders' reference positions $\mathbf{r}_{i,0}=[x_{i,0}~y_{i,0}~0]^T$ and $\mathbf{r}_{j,0}=[x_{j,0}~y_{j,0}~0]^T$ and leaders' global desired positions $\mathbf{r}_{i,HT}\left(t\right)=[x_{i,HT}\left(t\right)~y_{i,HT}\left(t\right)~z_{i,HT}\left(t\right)]^T$ and $\mathbf{r}_{j,HT}\left(t\right)=[x_{j,HT}\left(t\right)~y_{j,HT}\left(t\right)~z_{j,HT}\left(t\right)]^T$ satisfying homogeneous transformation condition \eqref{homogtrans}, the following relation holds:
\[
  \left(\mathbf{r}_{d,i}-\mathbf{r}_{j,HT}\right)^T\left(\mathbf{r}_{d,i}-\mathbf{r}_{j,HT}\right)=\left(\mathbf{r}_{i,0}-\mathbf{r}_{j,0}\right)^T\mathbf{U}_D^2\left(\mathbf{r}_{d,i}-\mathbf{r}_{j,HT}\right),
\]
where
\begin{equation}
    \mathbf{U}_D^2=
    \begin{bmatrix}
    \lambda_1^2\cos^2\psi_u+\lambda_2^2\sin^2\psi_u&\left(\lambda_1^2-\lambda_2^2\right)\sin\psi_u\cos\psi_u&0\\
    \left(\lambda_1^2-\lambda_2^2\right)\sin\psi_u\cos\psi_u&\lambda_1^2\sin^2\psi_u+\lambda_2^2\cos^2\psi_u&0\\
    0&0&1
    \end{bmatrix}
    .
\end{equation}
\end{remark}

\begin{theorem}
Assume vehicles are distributed on the plane normal to $\hat{\mathbf{u}}_3$. Given leaders' reference positions ($\mathbf{r}_{1,0}$, $\mathbf{r}_{2,0}$, and $\mathbf{r}_{3,0}$) and leaders' global desired positions at a time $t\geq t_0$ ($\mathbf{r}_{1,HT}$, $\mathbf{r}_{2,HT}$, and $\mathbf{r}_{3,HT}$), $\mathbf{U}_D$ eigenvalues $\lambda_1$ and $\lambda_2$ and deformation angle $\psi_u$ are obtained by
\begin{subequations}
\label{l1l2l3}
\begin{equation}
    \lambda_1=\sqrt{\dfrac{a+c}{2}+\sqrt{\big[{1\over 2}\left(a-c\right)\big]^2+b^2}}
\end{equation}
\begin{equation}
    \lambda_2=\sqrt{\dfrac{a+c}{2}-\sqrt{\big[{1\over 2}\left(a-c\right)\big]^2+b^2}}
\end{equation}
\begin{equation}
    \psi_u=\dfrac{1}{2}\tan^{-1}\left(\dfrac{2b}{a-c}\right),
\end{equation}
\end{subequations}
where
\begin{equation}
\label{abc}
\setlength\arraycolsep{0.5pt}
\begin{split}
&
    \begin{bmatrix}
    a\\
    b\\
    c
    \end{bmatrix}
    =
    \begin{bmatrix}
    \left(x_{2,0}-x_{1,0}\right)^2&~2\left(x_{2,0}-x_{1,0}\right)\left(y_{2,0}-y_{1,0}\right)~&\left(y_{2,0}-y_{1,0}\right)^2\\
    \left(x_{3,0}-x_{2,0}\right)^2&~2\left(x_{3,0}-x_{2,0}\right)\left(y_{3,0}-y_{2,0}\right)~&\left(y_{3,0}-y_{2,0}\right)^2\\
    \left(x_{1,0}-x_{3,0}\right)^2&~2\left(x_{1,0}-x_{3,0}\right)\left(y_{1,0}-y_{3,0}\right)~&\left(y_{1,0}-y_{3,0}\right)^2\\
    \end{bmatrix}
    ^{-1}
    \\
    &
    \begin{bmatrix}
    \left(x_{2,HT}-x_{1,HT}\right)^2+\left(y_{2,HT}-y_{1,HT}\right)^2+\left(z_{2,HT}-z_{1,HT}\right)^2\\
    \left(x_{3,HT}-x_{2,HT}\right)^2+\left(y_{3,HT}-y_{2,HT}\right)^2+\left(z_{3,HT}-z_{2,HT}\right)^2\\
    \left(x_{1,HT}-x_{3,HT}\right)^2+\left(y_{1,HT}-y_{3,HT}\right)^2+\left(z_{1,HT}-z_{3,HT}\right)^2\\
    \end{bmatrix}
    .
\end{split}
\end{equation}
\end{theorem}
\begin{proof}
Given leaders' reference positions $\mathbf{r}_{i,0}$ and $\mathbf{r}_{j,0}$ and leaders' current desired positions $\mathbf{r}_{d,i}\left(t\right)=[x_{d,i}\left(t\right)~y_{d,i}\left(t\right)~z_{d,i}\left(t\right)]^T$ ($i=1,2,3$), satisfying homogeneous transformation condition \eqref{homogtrans}, the following relation holds:
\begin{subequations}
\begin{equation}
\label{rd0rdt1}
\resizebox{0.99\hsize}{!}{%
$
\begin{split}
   \left(\mathbf{r}_{2,HT}-\mathbf{r}_{1,HT}\right)^T\left(\mathbf{r}_{2,HT}-\mathbf{r}_{1,HT}\right)=\left(\mathbf{r}_{2,0}-\mathbf{r}_{1,0}\right)^T\mathbf{U}_D^2\left(\mathbf{r}_{2,HT}-\mathbf{r}_{1,HT}\right),
\end{split}
$
}
\end{equation}
\begin{equation}
\label{rd0rdt2}
\resizebox{0.99\hsize}{!}{%
$
\begin{split}
    &\left(\mathbf{r}_{3,HT}-\mathbf{r}_{2,HT}\right)^T\left(\mathbf{r}_{3,HT}-\mathbf{r}_{2,HT}\right)=\left(\mathbf{r}_{3,0}-\mathbf{r}_{2,0}\right)^T\mathbf{U}_D^2\left(\mathbf{r}_{3,HT}-\mathbf{r}_{2,HT}\right),
\end{split}
$
}
\end{equation}
\begin{equation}
\label{rd0rdt3}
\resizebox{0.99\hsize}{!}{%
$
\begin{split}
    &\left(\mathbf{r}_{1,HT}-\mathbf{r}_{3,HT}\right)^T\left(\mathbf{r}_{1,HT}-\mathbf{r}_{3,HT}\right)=\left(\mathbf{r}_{1,0}-\mathbf{r}_{3,0}\right)^T\mathbf{U}_D^2\left(\mathbf{r}_{1,HT}-\mathbf{r}_{3,HT}\right),
\end{split}
$
}
\end{equation}
\end{subequations}
where
\begin{subequations}
\begin{equation}
    \mathbf{U}_D^2=
    \begin{bmatrix}
    a&b&0\\
    b&c&0\\
    0&0&1
    \end{bmatrix}
    ,
\end{equation}
\begin{equation}
    a=\lambda_1^2\cos^2\psi_u+\lambda_2^2\sin^2\psi_u,
\end{equation}
\begin{equation}
    b=\left(\lambda_1^2-\lambda_2^2\right)\sin\psi_u\cos\psi_u,
\end{equation}
\begin{equation}
    c=\lambda_1^2\sin^2\psi_u+\lambda_2^2\cos^2\psi_u.
\end{equation}
\end{subequations}

Considering Proposition \ref{prop3}, it is concluded that $z$ components of leaders' reference positions are zero  ($\mathbf{r}_{i,0}=[x_{i,0}~y_{i,0}~0]^T$ ($i=1,2,3$, $i\in\mathcal{V}_L$). Substituting $\mathbf{r}_{i,0}$ and $\mathbf{r}_{i,HT}\left(t\right)$ ($i=1,2,3$) into Eqs. \eqref{rd0rdt1}, \eqref{rd0rdt2}, and \eqref{rd0rdt3}, $a$, $b$, and $c$ are obtained as given in Eq. \eqref{abc}. Given $a$, $b$, and $c$, $\lambda_1$, $\lambda_2$, and $\lambda_3$ are obtained as given in Eq. \eqref{l1l2l3}.  
\end{proof}

\textbf{Rotation Matrix $\mathbf{R}_D$:} Given leaders' reference positions $\mathbf{r}_{1,0}$, $\mathbf{r}_{2,0}$, and $\mathbf{r}_{3,0}$, leaders desired positions $\mathbf{r}_{1,HT}\left(t\right)$, $\mathbf{r}_{2,HT}\left(t\right)$, and $\mathbf{r}_{3,HT}\left(t\right)$, {\color{black}the} pure deformation matrix $\mathbf{U}_D$ is computed using Eq. \eqref{matrixUD}. Under a homogeneous transformation, 
$\mathbf{v}_{1,0}=\mathbf{r}_{2,0}-\mathbf{r}_{1,0}$, $\mathbf{v}_{2,0}=\mathbf{r}_{3,0}-\mathbf{r}_{2,0}$, and $\mathbf{v}_{3,0}={\left(\mathbf{r}_{2,0}-\mathbf{r}_{1,0}\right)\times \left(\mathbf{r}_{2,0}-\mathbf{r}_{1,0}\right)\over \|\left(\mathbf{r}_{2,0}-\mathbf{r}_{1,0}\right)\times \left(\mathbf{r}_{2,0}-\mathbf{r}_{1,0}\right)\|}$  are transferred to $\mathbf{v}_{1,HT}=\mathbf{r}_{2,HT}-\mathbf{r}_{1,HT}$, $\mathbf{v}_{2,HT}=\mathbf{r}_{3,HT}-\mathbf{r}_{2,HT}$, and $\mathbf{v}_{3,HT}={\left(\mathbf{r}_{2,HT}-\mathbf{r}_{1,HT}\right)\times \left(\mathbf{r}_{2,HT}-\mathbf{r}_{1,HT}\right)\over \|\left(\mathbf{r}_{2,HT}-\mathbf{r}_{1,HT}\right)\times \left(\mathbf{r}_{2,HT}-\mathbf{r}_{1,HT}\right)\|}$,  where $\mathbf{v}_{i,HT}=\mathbf{Q}\mathbf{v}_{i,0}$.
Define 
\[
\mathbf{L}_0=
\begin{bmatrix}
\mathbf{v}_{1,0}&\mathbf{v}_{2,0}&\mathbf{v}_{3,0}
\end{bmatrix}
~\mathrm{and}~
\mathbf{L}_d=
\begin{bmatrix}
\mathbf{v}_{1,HT}&\mathbf{v}_{2,HT}&\mathbf{v}_{3,HT}
\end{bmatrix}
.
\]
If leaders $1$, $2$, and $3$ form a triangle at time $t$, {\color{black}the elements of $\mathbf{Q}$} are obtained by
\begin{equation}
    \mathrm{vec}\left(\mathbf{Q}^T\right)=\left(\mathbf{I}_3\otimes \mathbf{L}_0^T\right)^{-1}\mathrm{vec}\left(\mathbf{L}_d^T\right).
\end{equation}
Given $\mathbf{Q}$ and $\mathbf{U}_D$, $\mathbf{R}_D=\mathbf{Q}\mathbf{U}_D^{-1}$ {\color{black}can be} obtained. 
\\

\textbf{Rigid-Body Displacement Vector $\mathbf{d}$:} Given $\mathbf{Q}$, $\mathbf{r}_{i,0}$, and $\mathbf{r}_{d,i}$, $\mathbf{d}=\mathbf{r}_{d,i}-\mathbf{Q}\mathbf{r}_{i,0}$.

\subsubsection{3-D Homogeneous Transformation Decomposition}
When an MVS is guided by four leaders and leaders form a tetrahedron at any time $t$, leaders' global desired positions satisfy the following rank condition:
\begin{equation}
\begin{split}
    &\forall t\geq t_0,\\
    &\mathrm{Rank}\left(
    \begin{bmatrix}
    \mathbf{r}_{2,HT}-\mathbf{r}_{1,HT}&\mathbf{r}_{3,HT}-\mathbf{r}_{1,HT}&\mathbf{r}_{d,4}-\mathbf{r}_{1,HT}
    \end{bmatrix}
    \right)
    =3.
\end{split}
\end{equation}
Define
\[
\mathbf{P}_0=
\begin{bmatrix}
{x}_{1,0}&{y}_{1,0}&{z}_{1,0}\\
{x}_{2,0}&{y}_{2,0}&{z}_{2,0}\\
{x}_{3,0}&{y}_{3,0}&{z}_{3,0}\\
{x}_{4,0}&{y}_{4,0}&{z}_{4,0}\\
\end{bmatrix}
~\mathrm{and}~
\mathbf{P}_d=
\begin{bmatrix}
{x}_{1,HT}&{y}_{1,HT}&{z}_{1,HT}\\
{x}_{2,HT}&{y}_{2,HT}&{z}_{2,HT}\\
{x}_{3,HT}&{y}_{3,HT}&{z}_{3,HT}\\
{x}_{d,4}&{y}_{d,4}&{z}_{d,4}\\
\end{bmatrix}
.
\]
Elements of $\mathbf{Q}$ and $\mathbf{d}$ are uniquely related to leaders' position components by \cite{rastgoftar2015swarm}:
\begin{equation}
\label{QQDD}
\begin{bmatrix}
\mathrm{vec}\left(\mathbf{Q}^T\left(t\right)\right)\\
\mathbf{d}\left(t\right)
\end{bmatrix}
=
\begin{bmatrix}
\mathbf I_3\otimes \mathbf{P}_0&\mathbf I_3\otimes \mathbf{1}_4
\end{bmatrix}
^{-1}
\mathrm{vec}\left(\mathbf{P}_d\left(t\right)\right).
\end{equation}
where $\mathbf I_3\in \mathbb{R}^{3\times 3}$ is an identity matrix and $\mathbf 1_4\in \mathbb{R}^{4\times 1}$ is a ones vector.

\begin{IEEEbiography}[{\includegraphics[width=1in,height=1.25in,clip,keepaspectratio]{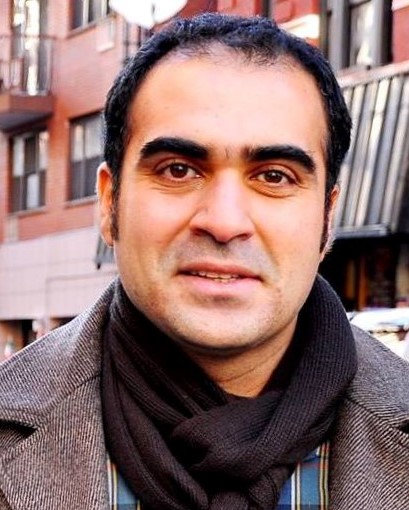}}]
{\textbf{Hossein Rastgoftar}} is an Assistant Research Scientist in the Department of Aerospace Engineering at the Uinveristy of Michigan. He received the B.Sc. degree in mechanical engineering-thermo-fluids from Shiraz University, Shiraz, Iran, the M.S. degrees in mechanical systems and solid mechanics from Shiraz University and the University of Central Florida, Orlando, FL, USA, and the Ph.D. degree in mechanical engineering from Drexel University, Philadelphia, in 2015. His current research interests include dynamics and control, multiagent systems, cyber-physical systems, and optimization and Markov decision processes.
\end{IEEEbiography}
\begin{IEEEbiography}[{\includegraphics[width=1in,height=1.25in,clip,keepaspectratio]{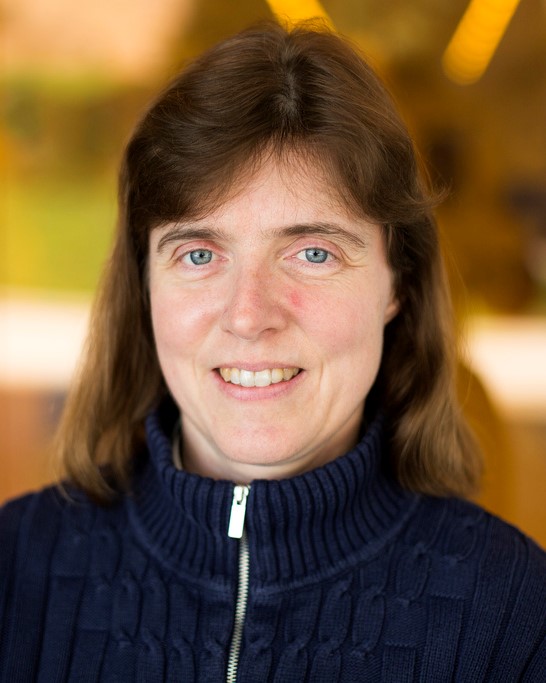}}]
{\textbf{Ella Atkins}} is a Full Professor of Aerospace Engineering at the University of Michigan, where she directs the Autonomous Aerospace Systems Lab and is Associate Director of Graduate Programs for the Robotics Institute.  Dr. Atkins holds B.S. and M.S. degrees in Aeronautics and Astronautics from MIT and M.S. and Ph.D. degrees in Computer Science and Engineering from the University of Michigan.  She is past-chair of the AIAA Intelligent Systems Technical Committee and has served on the National Academy's Aeronautics and Space Engineering Board, the Institute for Defense Analysis Defense Science Studies Group, and an NRC committee to develop an autonomy research agenda for civil aviation. She pursues research in Aerospace system autonomy and safety.
\end{IEEEbiography}
\begin{IEEEbiography}[{\includegraphics[width=1in,height=1.2in,clip,keepaspectratio]{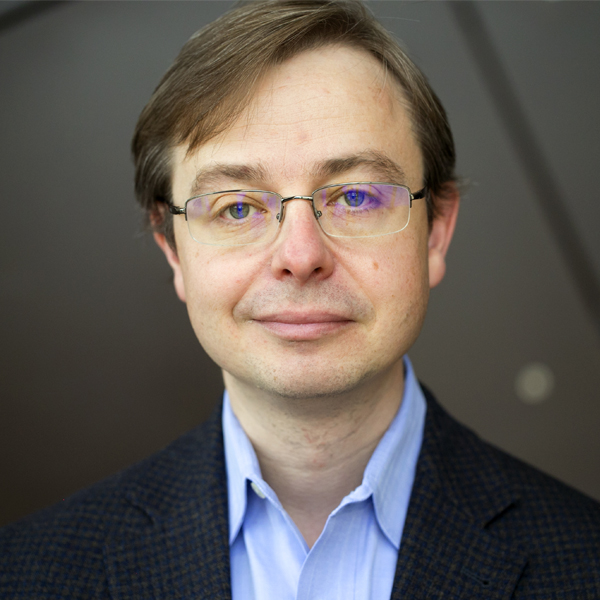}}]
{\textbf{Ilya V. Kolmanovsky}} received M.S. and Ph.D. degrees in aerospace engineering and the M.A. degree in mathematics from the University of Michigan, Ann Arbor, MI, USA, in 1993, 1995, and 1995, respectively.  Between 1995 and 2009, he was with Ford Research and Advanced Engineering, Dearborn, MI, USA. He is currently a Full Professor with the Department of Aerospace Engineering, University of Michigan. His  research interests include control theory for systems with state and control constraints, and control applications to aerospace and automotive systems. Dr. Kolmanovsky was a recipient of the Donald P. Eckman Award of American Automatic Control Council and two IEEE Transactions on Control Systems Technology Outstanding Paper Awards. Dr. Kolmanovsky is an IEEE Fellow and {\color{black}AIAA Associate Fellow}.
\end{IEEEbiography}

\end{document}